\documentclass[a4paper,12pt,numbers=endperiod,headinclude,abstract=on]{scrartcl}

\usepackage[american]{babel}
\usepackage[T1]{fontenc}
\usepackage{lmodern}
\usepackage{microtype}
\usepackage{graphicx}
\usepackage{xcolor}
\usepackage{amsmath,amsfonts,amssymb,amsxtra}
\usepackage[thmmarks,amsmath]{ntheorem}
\usepackage[para]{footmisc}
\usepackage{setspace}
\usepackage{mathrsfs}
\usepackage{dsfont}
\usepackage{verbatim}
\usepackage{relsize}
\usepackage{float}
\usepackage{upgreek}
\usepackage{enumitem}
\usepackage{bm}
\usepackage{fixcmex}

\usepackage{geometry}
\usepackage{bookmark} 
\hypersetup{
	colorlinks,
	linkcolor={red!50!black},
	citecolor={blue!50!black},
	urlcolor={blue!50!black}
}

\theoremstyle{break}
\newtheorem{theorem}{Theorem}[section]
\newtheorem{proposition}[theorem]{Proposition}
\newtheorem{corollary}[theorem]{Corollary}
\newtheorem{lemma}[theorem]{Lemma}

\theorembodyfont{\normalfont}
\newtheorem{definition}[theorem]{Definition}

\theoremstyle{nonumberbreak}
\theoremheaderfont{\bfseries}
\theorembodyfont{\normalfont}
\theoremseparator{.}
\theoremsymbol{\linebreak[1]\hspace*{.5em plus 1fill}\ensuremath{\square}}
\newtheorem{proof}{Proof}

\numberwithin{equation}{section} 

\allowdisplaybreaks 

\usepackage{accents}
\let\utilde\undertilde

\DeclareMathOperator{\ran}{ran}
\DeclareMathOperator{\dom}{dom}
\DeclareMathOperator{\tr}{tr}
\DeclareMathOperator{\id}{id}
\DeclareMathOperator{\sgn}{sgn}
\DeclareMathOperator{\spec}{spec}

\DeclareMathOperator{\diag}{diag}
\DeclareMathOperator{\dGamma}{d\Gamma}
\DeclareMathOperator{\Fol}{Fol}
\newcommand{\slot}{\,\cdot\,}
\newcommand{\Cinf}{C_{\mathrm{c}}^{\infty}}
\newcommand{\Wcal}{\mathcal{W}}
\newcommand{\Hcal}{\mathcal{H}}
\newcommand{\Mcal}{\mathcal{M}}
\newcommand{\Ncal}{\mathcal{N}}
\newcommand{\Fcal}{\mathcal{F}}
\newcommand{\Bcal}{\mathcal{B}}
\newcommand{\Acal}{\mathcal{A}}
\newcommand{\Rcal}{\mathcal{R}}
\newcommand{\Kcal}{\mathcal{K}}
\newcommand{\Gcal}{\mathcal{G}}
\newcommand{\Ocal}{\mathcal{O}}

\newcommand{\Ffrak}{\mathfrak{F}}
\newcommand{\Ftfrak}{\widetilde{\mathfrak{F}}}
\newcommand{\Tfrak}{\mathfrak{T}}
\newcommand{\Poincare}{\mathscr{P}_+^\uparrow}
\newcommand{\Lorentz}{\mathscr{L}_+^\uparrow}
\newcommand{\Transl}{\mathscr{T}}
\newcommand{\R}{\mathds{R}}
\newcommand{\C}{\mathds{C}}
\newcommand{\M}{\mathds{M}}
\newcommand{\WR}{\mathds{W}}
\newcommand{\restr}{\!\!\upharpoonright}
\newcommand{\WF}{W^{\mathcal{F}}}
\newcommand{\kfbeta}{\boldsymbol{\mathsf{k}}_\beta}
\newcommand{\kf}{\boldsymbol{\mathsf{k}}}
\newcommand{\gf}{\boldsymbol{\mathsf{g}}}
\newcommand{\Kfbeta}{\boldsymbol{\mathsf{K}}_\beta}
\newcommand{\acc}{\mathsf{a}}
\newcommand{\e}{\mathrm{e}}
\newcommand{\Cc}{\mathtt{C}}
\newcommand{\ham}{\mathsf{h}}
\renewcommand*{\Re}{\operatorname{Re}}
\renewcommand*{\Im}{\operatorname{Im}}
\renewcommand*{\vec}[1]{\underline{#1}}
\newcommand{\xu}{\underline{x}}
\newcommand{\ku}{\underline{k}}
\newcommand{\JPmap}{\mathfrak{J}\mkern-2.1mu\mathfrak{P}}
\newcommand{\JP}{Jak\v{s}i\'{c}-Pillet}
\newcommand*\diff{\mathop{}\!\mathrm{d}}
\newcommand*\Diff[1]{\mathop{}\!\mathrm{d}^#1}

\newcommand{\frameu}{\mathsf{u}}
\newcommand{\framew}{\mathsf{w}}

\newcommand*\class[1]{[\![#1]\!]}
\newcommand{\Csymp}{C_{\M}}

\newcommand{\mysepline}{\noindent\makebox[\linewidth]{\resizebox{0.3333\linewidth}{1pt}{$\bullet$}}\bigskip}

\usepackage{authblk}

\renewcommand{\thefootnote}{\fnsymbol{footnote}}
\setkomafont{dedication}{\normalsize\itshape}

\title{{\Large Disjointness of inertial KMS states and the \\ role of Lorentz symmetry in thermalization}}

\author[1,2]{Albert Georg Passegger\thanks{e-mail: \texttt{albert\_georg.passegger@uni-leipzig.de}}}
\author[1,3]{Rainer Verch\thanks{e-mail: \texttt{rainer.verch@uni-leipzig.de}}}
\affil[1]{Institut f\"{u}r Theoretische Physik, Universit\"{a}t Leipzig, Leipzig, Germany}
\affil[2]{Max Planck Institute for Mathematics in the Sciences, Leipzig, Germany}
\affil[3]{CY Advanced Studies, CY Cergy Paris Universit\'{e}, Neuville-sur-Oise, France}

\date{\vspace{-5mm}}
\dedication{Dedicated to Detlev Buchholz on the occasion of his 80th birthday}

\begin{document}
	
\pagenumbering{arabic}

\maketitle

\vspace{-5mm}
\renewcommand{\thefootnote}{\arabic{footnote}}

\begin{abstract}
\noindent For any local, translation-covariant quantum field theory on Minkowski spacetime, we prove that two distinct states that are invariant under the inertial time evolutions in different inertial reference frames are disjoint, i.e.\ neither state is a perturbation of the other, if the states are primary, have separating Gelfand-Naimark-Segal (GNS) vectors, and satisfy a timelike cluster property called the mixing property. These conditions are fulfilled by the inertial Kubo-Martin-Schwinger (KMS) states of the free scalar field, thus showing that a state satisfying the KMS condition relative to one inertial frame is far from thermal equilibrium relative to other inertial frames. We review the property of return to equilibrium (RTE) in open quantum systems theory and discuss the implications of disjointness on the asymptotic behavior of detector systems coupled to states of a free massless scalar field. We argue that the coupled system of an Unruh-DeWitt detector moving with constant velocity relative to the field in a KMS state, or an excitation thereof, cannot thermalize under generic conditions. This leads to an illustration of the physical differences between heat baths in inertial systems and the alleged ``heat bath'' of the Unruh effect. This paper also sketches the construction and RTE property of the quantum dynamical system of an Unruh-DeWitt detector coupled to a massless scalar field in a KMS state relative to the inertial rest frame of the detector.
\bigskip

\noindent {\footnotesize 2020 \textit{Mathematics Subject Classification}: \href{https://zbmath.org/classification/?q=46L60}{46L60}, \href{https://zbmath.org/classification/?q=81T05}{81T05}, \href{https://zbmath.org/classification/?q=82B10}{82B10}, \href{https://zbmath.org/classification/?q=82C10}{82C10}}
\end{abstract}

\newpage

\section{Introduction}
\label{sec:intro}

Thermal equilibrium states in possibly infinitely extended quantum systems are characterized by the Kubo-Martin-Schwinger (KMS) condition \cite{Kubo1957,Martin-Schwinger1959,HHW1967,Bratteli-Robinson2}. It can be obtained from the thermodynamic limit of Gibbs ensembles and is expressed in the algebraic framework of quantum statistical mechanics as an analyticity condition for state functionals on the C$^\ast$-algebra of observables of a quantum system. States satisfying the KMS condition, called KMS states, are characterized by properties that render them as faithful representatives of global thermal equilibrium, such as time invariance, stability under perturbations \cite{Haag-Kastler-Trych-Pohlmeyer}, and a thermodynamic interpretation \cite{Kossakowski1977,Pusz-Woronowicz}. From a mathematical perspective, the KMS condition has a close connection to the theory of faithful normal states on von Neumann algebras, which is the topic of Tomita-Takesaki modular theory \cite{Takesaki1970,Bratteli-Robinson1,Borchers-Tomita}. A wide range of applications may be found in the monographs \cite{Bratteli-Robinson2,Haag1996} and the references cited therein.\medskip

KMS states have also been studied in the context of (relativistic) quantum field theory. For linear quantum fields on static spacetimes, KMS states and their mixtures satisfy the microlocal spectrum condition (Hadamard condition) \cite{Radzikowski1996,Sahlmann-Verch}, a regularity criterion on small scales that distinguishes physical states and is, under certain conditions, equivalent to the implementation of quantum energy inequalities for such states \cite{Fewster-Verch2003}. 

Mimicking the intuitive measurement process in an experimental setup, thermal phenomena in quantum field theory such as the Unruh \cite{Unruh1976,Fulling1973,Davies1975,Sewell1982,Crispino-Higuchi-Matsas} and Hawking effect \cite{Hawking1974,Hawking1975,Wald1975,Fredenhagen-Haag1990,Bachelot} (see also \cite{Wald-book,Hollands-Wald-2015}), which employ the KMS condition in their description, may be investigated by coupling a ``detector'' system to the observables of the field. A simple detector model that is commonly considered in this context is the Unruh-DeWitt detector \cite{Unruh1976,DeWitt,Unruh-Wald1984,Birrell-Davies}, a point-like (or sufficiently spatially localized) non-relativistic quantum system with finitely many (commonly two) energy levels, whose coupling to the field is realized, for example, by a monopole interaction. The process of measuring observables of a quantum field using ``probes'' can be formulated in a model-independent, local and covariant way \cite{FV-measurement,FV-measurement-enc}. One expects that Unruh-DeWitt type detector models can be regarded as a special case of this general theory of quantum field measurement processes; first steps in that direction were taken in \cite[Sec.\ 5.3]{FV-measurement}, and recent developments and discussions can be found in \cite{PoloGomez-Garay-MM2022,Perche-PoloGomez-Torres-MM2024,Torres2024,Papageorgiou-Fraser2024}.

An interesting question is how an Unruh-DeWitt detector system responds to the state of a quantum field depending on the motion of the detector along some timelike worldline. There are two physically different ways of observing what may be attributed to the system as thermal behavior: recording thermal transition rates or probabilities between the energy levels of the detector in time-dependent perturbation theory (``particle detector'' excitation), and late time convergence of an initial state into a thermal equilibrium state under the time evolution (``thermometer'' equilibration), a process referred to as \textit{thermalization}.\medskip

The Unruh effect is a prominent instance in which such setups have been considered. The operational statement of the Unruh effect is that a linearly and uniformly accelerating detector with proper acceleration $\acc>0$ in Minkowski spacetime coupled to the Minkowski vacuum state of a quantum field theory detects the Unruh temperature $\acc/2\pi$. (A review of multiple facets of what is called ``Unruh effect'' in the literature can be found in \cite{Earman2011}.) For an Unruh–DeWitt detector linearly coupled to a free massless scalar field one obtains, in first-order perturbation theory, a transition rate (or probability) corresponding to a Planck (black body) spectrum at the Unruh temperature \cite{Unruh1976,Birrell-Davies,Crispino-Higuchi-Matsas,Unruh-Wald1984,Takagi1986}. In this approach, thermality of a two-level Unruh-DeWitt detector can be expressed as the detailed balancing of the quotient of excitation and de-excitation probabilities (see, e.g., \cite{Takagi1986,Waiting-for-Unruh,Perche-thermalization}, and the recent \cite{Lima-Patterson-Tjoa-Mann2023} for a discussion of peculiarities in higher dimensional detector systems).

The asymptotic state behavior is a fundamentally different property of the system. For instance, \cite{Moustos-Anastopoulos2017,Moustos2018} showed that the transition rate of a uniformly accelerated detector coupled to the Minkowski vacuum can be non-Planckian at early times, depending on the type of coupling, but the Unruh effect is still observed in the sense that the asymptotic (reduced density matrix) state of the detector at late times is a Gibbs equilibrium state at the Unruh temperature. Moreover, the shape of the response function of the uniformly accelerated detector coupled to the Minkowski vacuum can depend on the spacetime dimension: For the scalar field in odd spacetime dimensions the distribution has a Fermi-Dirac form, as opposed to the Planck-type (Bose-Einstein) distribution that appears in even dimensions \cite{Takagi1986,Unruh1986}. (By contrast, the response function of the detector coupled to the Dirac field is Planckian in all spacetime dimensions \cite{Louko-Toussaint2016}.) Despite this ``inversion of statistics'', the detector still thermalizes at late times \cite{DeB-M,inversion-of-statistics-2021,Moustos2022}. 

The behavior of Unruh-DeWitt detectors and more general ``particle detectors'' is a quite active research area. A selection of recent investigations is given by \cite{Moustos-Anastopoulos2017,Moustos2018} (open quantum systems analysis \cite{Breuer-Petruccione} of the reduced density matrix of uniformly accelerated detectors), \cite{Biermann-et-al2020,Good-JA-Moustos-Temirkhan2020,Bunney-Parry-Perche-Louko2024} (Unruh-like effects for Unruh-DeWitt detectors along circular and other stationary accelerated trajectories in Minkowski spacetime), and \cite{JuarezAubry2019,Perche-thermalization} (detectors following timelike worldlines while being coupled to quantum fields in curved spacetimes in KMS states relative to time evolution along the worldline); see the references therein for further literature.\medskip 

A strong form of thermalization for the full detector-field system (not just the detector alone) that will be the focus in our work has been considered by De Bi\`{e}vre \& Merkli \cite{DeB-M} in the case of the Unruh effect: The open quantum system that is obtained from linearly coupling a two-level detector (using a monopole) to the free scalar field satisfies the property of ``return to equilibrium'' (RTE), an ergodicity property of quantum dynamical systems (see, e.g., \cite{JP2,JP-thermal-relaxation,BFS2000,Pillet-Attal2006}, and Section \ref{sec:th-rte} for more references). The coupled detector-field system converges to the KMS state corresponding to the Unruh temperature, which entails the approach of the reduced state of the detector to the associated Gibbs state. The result does not only hold if the initial state is given by the detector in the ground state and the field in the Minkowski vacuum state, but also if the initial detector-field state is an excitation of that. Therefore, the asymptotic state behavior of the coupled system is stable in a quite strong sense, as thermal equilibrium is reached from a whole folium of initial states.\medskip

Even though the Unruh effect appears to hint towards a thermal interpretation for the Minkowski vacuum, one might ask if the Unruh temperature is a quantity that actually pertains to a property of the quantum field as a heat bath of ``Rindler particles'' in the vacuum state from the perspective of a uniformly accelerated observer (as claimed in, e.g., \cite{Crispino-Higuchi-Matsas}). The comparison with characteristics of conventional heat baths in inertial reference frames (Lorentz frames) can be taken as a litmus test in that regard. A fundamental result against the heat bath picture of the Unruh effect is obtained from the interpretation of KMS states. The Unruh effect is related to the property that the Minkowski vacuum state is a KMS state with KMS parameter $2\pi$ with respect to the dynamics given by the Lorentz boost isometries, which leave the right wedge of Minkowski spacetime invariant \cite{Fulling1973,Unruh1976,Bisognano-Wichmann,Bisognano-Wichmann2,Sewell-BW1980,Sewell1982}. However, while the interpretation of inertial KMS states representing thermal equilibrium at positive temperature in the sense of the zeroth and the second law is well-established \cite{Kossakowski1977,Pusz-Woronowicz} (see also \cite{Haag1996,Bratteli-Robinson2}), the inverse of the KMS parameter does not represent a temperature in non-inertial situations \cite{Buchholz-Solveen,Solveen2012}. It turns out that under the influence of external forces, such as the accelerating force in the Unruh effect, Unruh-DeWitt detectors cannot be used as reliable ``thermometers'' for the quantum field in the sense of the zeroth law of thermodynamics (i.e.\ measuring devices attaining their temperature by purely thermal contact) \cite{Buchholz-Solveen,Buchholz-Verch2015,Buchholz-Verch2016}. Still, it is remarkable that detector models can exhibit thermality under uniform acceleration (cf.\ the remarks in \cite[Sec.\ 1]{DeB-M}). Hence Unruh-DeWitt detectors are useful instruments to study conceptual questions on idealized detector systems coupled to quantum field theories.\medskip 

Complementary to the above results, this paper treats the behavior of Unruh-DeWitt detectors under the action of Lorentz boosts to reveal further differences between the Unruh effect and thermalization in inertial heat baths. It has been noted in \cite{Candelas-Deutsch-Sciama} that the thermality of a uniformly accelerated detector is preserved under change of the Lorentz frame: 

\begin{quotation}
	\noindent ``For example, the temperature depends on the acceleration of the observer but is independent of his velocity. Thus two neighbouring observers with the same acceleration, but with velocities differing by nearly the velocity of light would observe the \textit{same} temperature–-there would be no Doppler transformation from one temperature to the other.'' \cite[p.\ 334]{Candelas-Deutsch-Sciama}
\end{quotation}

\noindent In other words, while the detector follows the orbit of a Lorentz boost describing its uniformly accelerated motion and is therefore stationary relative to the Killing flow (in turn, stationary relative to the Minkowski vacuum to which it couples), changing the Lorentz frame by boosting to a different velocity relative to an inertial observer does not affect the asymptotic temperature registered by the detector. 

The goal of this paper is to contrast that statement with the situation of an inertially moving detector coupled to an inertial heat bath. We do this at the level of asymptotic state properties of coupled detector-field systems (in the spirit of \cite{DeB-M}). This further reveals the innate differences between proper, inertial heat baths and the heuristic heat bath picture of the Unruh effect.\medskip

For our discussion we prove the following algebraic property of states: Let $\sigma$ and $\omega$ be two distinct states of a local quantum field theory on Minkowski spacetime that are invariant under the time translations in different inertial frames (time evolution along two timelike vectors that are related by a Lorentz boost). Then the Gelfand-Naimark-Segal (GNS) representations associated to the two states are disjoint, i.e.\ $\sigma$ and $\omega$ are not in the same folium of states and thus are not a perturbation of one another, if they have separating GNS vectors and satisfy certain physically reasonable properties: primarity, which formalizes a ``pure phase'', and an asymptotic timelike cluster property called the mixing property. The algebraic conditions of primarity and mixing are characterized by spectral properties of the generator of the time evolution on the GNS Hilbert space, which are shown to be fulfilled by the quasi-free inertial KMS states of the free scalar quantum field on four-dimensional Minkowski spacetime. Hence two such KMS states with respect to different inertial frames are disjoint.

The result can be expressed as Lorentz symmetry breaking in primary, mixing states that are not Lorentz boost invariant. We formulate this explicitly by showing that Lorentz boosts are not implementable by unitary operators in the GNS representation of such states. In the case of inertial KMS states, for which the breakdown of Lorentz symmetry has been discussed before by Ojima \cite{Ojima1986}, one can interpret this as the formalization of the observation that there is no Lorentz transformation law relating temperatures relative to different inertial frames \cite{Landsberg-Matsas1996,Landsberg-Matsas2004,Sewell2008,Sewell-rep2009} (see also \cite{Sewell2010}, where this is shown with respect to macroscopic properties).\footnote{The question how temperature transforms between different inertial frames is a fundamental, highly debated topic in relativistic thermodynamics, and there are several conflicting results based on various approaches and arguments. A brief overview that includes a list of references to some relevant works can be found in \cite{Sewell2008} and the review \cite{Farias-Pinto-Moya}.}\medskip

The disjointness of inertial KMS states strengthens the result by Sewell \cite{Sewell2008,Sewell-rep2009} that a state of a quantum field theory which satisfies the KMS condition with respect to time translations in one inertial frame cannot be a KMS state with respect to time translations in a different, Lorentz boosted frame. Sewell concluded that a detector moving relative to an inertial heat bath with constant non-zero velocity does not approach a KMS state, since the state of the heat bath does not satisfy the KMS condition relative to the rest frame of the detector. For a heat bath given by a massless scalar field, this is supported by the non-thermal transition rate resulting from the Doppler shifted quanta that are registered by the detector \cite{Costa-Matsas-background1995,Costa-Matsas1995}; changing to an inertial frame that is moving relative to the rest frame of the heat bath thus has an influence on the excitation spectrum of the detector.\medskip

In our work we present arguments corroborating these results for the full, coupled detector-field system in the Hamiltonian approach to open quantum systems theory \cite{Attal2006}. As the initial state of the heat bath is far from equilibrium with respect to the notion of thermal equilibrium determined by the rest frame of the inertially moving detector, it is suggested that the coupled detector-field system does not approach a KMS state for generic couplings. We motivate this conclusion by combining the disjointness of inertial KMS states with the RTE property of an Unruh-DeWitt detector coupled to an inertial heat bath (massless scalar field) in its rest frame. By contrast, the Minkowski vacuum state is Poincar\'{e}-invariant, which implies that an accelerated detector will return to equilibrium in the sense of \cite{DeB-M} in every boosted Lorentz frame. Our discussion thereby places the above statement from \cite{Candelas-Deutsch-Sciama} in a broader context and illustrates that the uniformly accelerated detector in the Unruh effect does not behave as if it was coupled to a heat bath.\medskip

This work is structured as follows. In Section \ref{sec:prelim} we introduce necessary notions from algebraic quantum field theory and quantum statistical mechanics, in particular primarity and the mixing property, with complementary material being collected in Appendix \ref{appendix:qds-add}. Section \ref{sec:disjointness} contains the disjointness results for invariant states relative to different inertial reference frames, first for general primary, mixing states (Theorem \ref{thm:primary-disjoint}), then applied to inertial KMS states of the free scalar field (Theorem \ref{thm:KMS-disjoint}). The formulation in terms of the breakdown of Lorentz symmetry in these states is presented in Section \ref{sec:breaking}. We discuss the relation to the literature and compare our disjointness result for KMS states with a similar result due to \cite{Herman-Takesaki1970} in Section \ref{sec:discussion-result}. In Section \ref{sec:detectors} the suggested implications of disjointness on the asymptotic state properties of Unruh-DeWitt detectors are discussed. After an introduction to the algebraic formalization of thermalization, we discuss the RTE property for a detector that rests relative to a heat bath of a massless scalar field (Section \ref{sec:inertial-rte}). The necessary field theoretic formalism and a useful implementation of the GNS representation for KMS states of the massless scalar field (the \JP\ glued representation) are compiled in Appendix \ref{appendix:scalar-field} and Appendix \ref{appendix:JP-glued-rep}, which shows the connection to literature on RTE for spin-boson models. If the detector moves inertially relative to such a heat bath, the disjointness of inertial KMS states indicates that the coupled detector-field system will not thermalize (Section \ref{sec:moving-inertial}). In Section \ref{sec:unruh} we compare the results with the Unruh effect. An outlook and concluding remarks pointing out open questions are given in Section \ref{sec:conclusions}.

\paragraph{Conventions} We use the ``mostly minus'' metric signature $(+,-,\ldots,-)$, and physical units where the vacuum speed of light, reduced Planck constant, and Boltzmann constant are set to $1$. By default or assumption, all Hilbert spaces are complex and separable, operator algebras are unital (with unit element $\mathds{1}$), and morphisms between them are unit-preserving.

\section{Setting and preliminaries}
\label{sec:prelim}

In this section we introduce the setting and concepts needed in this work.

\subsection{Algebraic quantum field theory}
\label{sec:prelim-aqft}

We consider a local quantum field theory on Minkowski spacetime, described in the model-independent algebraic framework of Haag \& Kastler \cite{Haag-Kastler,Haag1996,Borchers1996} (see also \cite{Fredenhagen-Intro-AQFT} for a concise introduction). The quantum field theory is given by a so-called ``net'' of C$^\ast$-algebras: To every open, bounded, causally convex subset $O\subset\M$ of $(1+d)$-dimensional Minkowski spacetime\footnote{$\M$ will be used to denote Minkowski spacetime as a manifold, but also the underlying set (which can be identified with $\R^{1+d}$) without notational distinction.} $\M=\R^{1,d}$ (for spatial dimension $d\geq 1$), referred to as \textit{region}, one associates a unital C$^\ast$-algebra $\Acal(O)$, called a \textit{local observable algebra}, subject to the following physically motivated assumptions.
\begin{itemize}[leftmargin=2em]
	\item[(A)] \textbf{Isotony}: If $O_1 \subset O_2$ then $\Acal(O_1)\subset\Acal(O_2)$ (i.e.\ there is a $^\ast$-monomorphism $\Acal(O_1)\hookrightarrow\Acal(O_2)$, and $\Acal(O_1)$ is identified with the image under this map).
\end{itemize}
The C$^\ast$-algebra $\Acal = \overline{\bigcup_O \Acal(O)}$ of the theory, called the \textit{quasi-local algebra}, is the norm-closed direct limit of the family $\{\Acal(O)\}_{O\subset\M}$, where $O$ runs over the direct system of all regions in $\M$ (with set inclusion relation). The C$^\ast$-algebras $\Acal(O)$ can be viewed as $^\ast$-subalgebras of $\Acal$.
\begin{itemize}[leftmargin=2em]
	\item[(B)] \textbf{Locality} (Einstein causality): If $O_1 \subset O^\perp$, where $O^\perp$ is the causal complement of $O$, then $[\Acal(O_1),\Acal(O)]=\{0\}$. This means that operators associated to spacelike separated regions commute.
	\item[(C)] \textbf{Translation covariance}: The additive group of spacetime translations $\Transl := (\R^{1+d},+)$ is represented on $\Acal$ as a group of $^\ast$-automorphisms $\{\updelta_x\}_{x\in\R^{1+d}}$ such that $\updelta_x(\Acal(O))=\Acal(O+x)$ for all regions $O\subset\M$ and $x\in\R^{1+d}$, where $O+x:=\{y+x : y\in O\}$.
\end{itemize}
Covariance under spacetime translations is sufficient for our main statements on the disjointness of states, where we will be interested in $^\ast$-automorphisms $\{\updelta_{t\frameu}\}_{t\in\R}$ representing the time evolution along the time direction of an inertial frame prescribed by a future-directed timelike unit vector $\frameu$. Assumption (C) encompasses theories that are covariant under the full symmetry group of Minkowski spacetime, the proper orthochronous Poincar\'{e} group $\Poincare$, of which $\Transl$ forms an abelian subgroup. This stronger form of covariance will be needed in the discussion of Lorentz symmetry breaking and the application to the thermalization of detector systems.
\begin{itemize}[leftmargin=2em]
	\item[(C')] \textbf{Poincar\'{e} covariance}: The proper orthochronous Poincar\'{e} group $\Poincare$ of $\M$ is represented on $\Acal$ as a group of $^\ast$-automorphisms $\{\updelta_{(x,\Uplambda)}\}_{(x,\Uplambda)\in\Poincare}$ such that $\updelta_{(x,\Uplambda)}(\Acal(O))=\Acal(\Uplambda O+x)$ for all regions $O\subset\M$ and $(x,\Uplambda)\in\Poincare$, where $\Uplambda O+x:=\{\Uplambda y+x : y\in O\}$.
\end{itemize}
Here we identified the group elements of the semidirect product $\Poincare=\Transl\rtimes\Lorentz$ with pairs $(x,\Uplambda)$ for vectors $x\in\R^{1+d}$ of spacetime translations and proper orthochronous Lorentz transformations $\Uplambda\in\Lorentz$ on $\M$, with the usual group law $(x,\Uplambda)\cdot(x',\Uplambda')=(x+\Uplambda x',\Uplambda\Uplambda')$. (An introduction to these groups can be found in, e.g., \cite{Sexl-Urbantke}.) The $^\ast$-automorphisms of spacetime translations from assumption (C) are obtained from (C') via $\updelta_x = \updelta_{(x,\mathds{1})}$ for $x\in\R^{1+d}$. The representation of Lorentz transformations on $\Acal$ will be denoted by $\updelta^{\Uplambda} := \updelta_{(0,\Uplambda)}$. From the group homomorphism property of $(x,\Uplambda)\mapsto\updelta_{(x,\Uplambda)}$ and the group law of $\Poincare$ we get $\updelta^{\Uplambda^{-1}} = (\updelta^{\Uplambda})^{-1}$ and the relation
\begin{gather}
	\label{eq:boost-translation-transf}
	\updelta^{\Uplambda} \circ \updelta_x = \updelta_{\Uplambda x} \circ \updelta^{\Uplambda} \quad \text{for all}\ (x,\Uplambda)\in\Poincare \, .
\end{gather}

Other typical assumptions on local quantum field theories, like the spectrum condition, nuclearity, or a dynamical law (time slice axiom) \cite{Haag1996,Fewster-Verch-AQFT2015}, will not be needed in this work. The axioms (A)--(C') fit into the general underlying structure of quantum field theories on globally hyperbolic spacetimes \cite{BFV2003,Fewster-Verch-AQFT2015} (with covariance being a consequence of functoriality of locally covariant quantum field theory) and can be shown to be fulfilled, for example, for the free scalar (Klein-Gordon) field \cite{Borchers1996,Fewster-Rejzner-AQFT,Dimock1980} and for free Dirac (spin $\frac{1}{2}$) and Proca (spin $1$) fields \cite{Borchers1996,Dimock1982,Sanders2010,Dappiaggi2011}.

\subsection{Quantum dynamical systems and invariant states}
\label{sec:prelim-qds}

For the reader's convenience and to fix notation we summarize some terminology and results from quantum statistical mechanics. Basic familiarity with operator algebras will be assumed. Additional material is put into Appendix \ref{appendix:qds-add}. For a thorough exposition, including proofs of all statements, we refer to \cite{Bratteli-Robinson1,Bratteli-Robinson2,Emch,Haag1996,Pillet-Attal2006}.\medskip

Throughout this section, $\Acal$ is a complex C$^\ast$-algebra, $\Bcal(\Hcal)$ is the C$^\ast$-algebra of bounded operators on a Hilbert space $\Hcal$, and von Neumann algebras, i.e.\ $^\ast$-subalgebras of some $\Bcal(\Hcal)$ closed in weak operator topology, will be denoted $\Mcal$. A detailed discussion of operator topologies can be found in \cite[Sec.\ 2.1.e]{Emch}, \cite[Sec.\ 2.4.1]{Bratteli-Robinson1}, and \cite[Secs.\ 1.7, 1.8, 1.15]{Sakai1998}. Here we just recall that the \textit{weak operator topology} on $\Bcal(\Hcal)$ is the weakest (coarsest) topology such that $\langle v, (\slot)w\rangle: \Bcal(\Hcal)\to\C$ is continuous for every $v,w\in\Hcal$, so a net $(A_i)_{i\in I}$ in $\Bcal(\Hcal)$ converges to $A\in\Bcal(\Hcal)$ in weak operator topology if and only if $\lim_{i\in I} \langle v,A_i w\rangle = \langle v,Aw\rangle$ for all $v,w\in\Hcal$ in the standard topology of $\C$. The \textit{ultraweak} (also called weak-$^\ast$ or $\sigma$-weak) \textit{topology} on $\Bcal(\Hcal)$ is the weakest topology such that the functionals $\tr(\rho(\slot)):\Bcal(\Hcal)\to\C$ are continuous for every trace class operator $\rho$ on $\Hcal$. The ultraweak topology is stronger (finer) than the weak operator topology, the topologies agree on norm-bounded subsets, and the closures of a $^\ast$-subalgebra of $\Bcal(\Hcal)$ in weak operator and ultraweak topology coincide.\medskip

Algebraic states are positive, normalized, linear functionals on an operator algebra. Any state can be represented as a vector in a Hilbert space.

\begin{proposition}[GNS representation {\cite[Thm.\ 2.3.16]{Bratteli-Robinson1}}]
	For every state $\omega:\Acal\to\C$ there is a $^\ast$-representation $\pi_\omega : \Acal \to \Bcal(\Hcal_\omega)$ on a Hilbert space $\Hcal_\omega$ and a unit vector $\Omega_\omega \in \Hcal_\omega$ such that $\omega=\langle\Omega_\omega , \pi_\omega(\slot)\Omega_\omega \rangle$ and $\pi_\omega(\Acal)\Omega_\omega$ is norm dense in $\Hcal_\omega$. The triple $(\pi_\omega,\Hcal_\omega,\Omega_\omega)$, or $\pi_\omega$ for short, is called the \emph{GNS (Gelfand-Naimark-Segal) representation} of $\omega$. It is unique up to unitary equivalence.
\end{proposition}

\begin{definition}[Commutant, cyclic, separating, normal, faithful]
	The \textit{commutant} of a $^\ast$-subalgebra $\Rcal\subseteq\Bcal(\Hcal)$ is $\Rcal':=\{X\in\Bcal(\Hcal) \, : \, XY=YX \; \text{for all}\ Y\in\Rcal\}\subseteq\Bcal(\Hcal)$. If $\Mcal\subseteq\Bcal(\Hcal)$ is a von Neumann algebra, a vector $\Omega\in\Hcal$ is \textit{cyclic} for $\Mcal$ if $\Mcal\Omega$ is dense in $\Hcal$, and \textit{separating} if $X\in\Mcal$ and $X\Omega=0$ implies $X=0$. A state $\varphi$ on $\Mcal$ is \textit{normal} if it is given by a density operator, i.e.\ $\varphi=\tr(\rho(\slot))$ for a self-adjoint, positive trace class operator $\rho$ on $\Hcal$ with unit trace. A normal state $\varphi$ is \textit{faithful} if $\varphi (X^\ast X)=0$ implies $X=0$. 
\end{definition}
A vector $\Omega\in\Hcal$ is separating for $\Mcal$ if and only if it is cyclic for $\Mcal'$. If $\varphi=\langle\Psi_\varphi , (\slot)\Psi_\varphi \rangle$ is a vector state for a unit vector $\Psi_\varphi \in \Hcal$, then $\varphi$ is faithful if and only if $\Psi_\varphi$ is separating for $\Mcal$. 

\begin{definition}[Induced von Neumann algebra]
	Let $\omega$ be a state on $\Acal$ with GNS representation $(\pi_\omega,\Hcal_\omega,\Omega_\omega)$. The \textit{von Neumann algebra induced by} $\omega$ is $\Mcal_\omega := \pi_\omega(\Acal)''\subseteq\Bcal(\Hcal_\omega)$, where $(\slot)'' = ((\slot)')'$ is the bicommutant.
\end{definition}
The bicommutant of a $^\ast$-subalgebra of some $\Bcal(\Hcal)$ coincides with the closure in weak operator topology by von Neumann's bicommutant theorem (see \cite[Thm.\ 2.4.11]{Bratteli-Robinson1}).

\begin{definition}[Folium, quasi-equivalence, disjointness]
	\label{def:folium}
	Let $\omega$ be a state on $\Acal$ with GNS representation $\pi_\omega$. The \textit{folium} $\Fol(\omega)$ of $\omega$ is the collection of density operator states $\tr(\rho\pi_\omega (\slot))$, called states \textit{normal to} $\omega$. Two states $\sigma,\omega$ on $\Acal$ are \textit{quasi-equivalent} if $\Fol(\sigma)=\Fol(\omega)$, and \textit{disjoint} if $\Fol(\sigma)\cap\Fol(\omega)=\emptyset$.
\end{definition} 
The folium $\Fol(\omega)$ of $\omega$ can be characterized as the norm closure of the convex hull of states $\omega(B^\ast (\slot) B)/\omega(B^\ast B)$ for $B\in\Acal$ such that $\omega(B^\ast B)\neq 0$, which are interpreted as ``perturbations'' or ``excitations'' of $\omega$. A state $\sigma$ is in the folium $\Fol(\omega)$ if and only if there is a unique normal state $\tilde{\sigma}$ on $\Mcal_\omega = \pi_\omega (\Acal)''$ such that $\sigma=\tilde{\sigma} \circ \pi_\omega$. This establishes an isomorphism between $\Fol(\omega)$ and normal states on $\Mcal_\omega$ (see the discussion preceding \cite[Def.\ 2.20]{Pillet-Attal2006}).

\begin{definition}[Quantum and W$^\ast$-dynamical system, invariant state]
	\label{def:dynamical-system}
	A \textit{quantum dynamical system} $(\Acal,\alpha)$ is a pair consisting of a C$^\ast$-algebra $\Acal$ and a one-parameter group $\alpha=\{\alpha_t\}_{t\in\R}$ of $^\ast$-automorphisms on $\Acal$, called the \textit{dynamics} on $\Acal$. If $\Mcal\subseteq\Bcal(\Hcal)$ is a von Neumann algebra equipped with dynamics $\gamma=\{\gamma_t\}_{t\in\R}$ that is ultraweakly continuous, i.e.\ $t \mapsto \tr(\rho \gamma_t (X))$ is continuous for all positive trace class operators $\rho$ on $\Hcal$ and $X\in\Mcal$, the quantum dynamical system $(\Mcal,\gamma)$ is called a \textit{W$^\ast$-dynamical system} with W$^\ast$-dynamics $\gamma$. A state $\omega$ on $(\Acal,\alpha)$ is called \textit{invariant under $\alpha$} (or \textit{$\alpha$-invariant}) if $\omega\circ\alpha_t = \omega$ for all $t\in\R$.
\end{definition}
As done in the above definition, if $\omega$ is a state on a C$^\ast$-algebra $\mathcal{A}$, and $(\mathcal{A},\alpha)$ is a quantum dynamical system, we will say that $\omega$ is a state on $(\mathcal{A},\alpha)$. The one-parameter group $\alpha$ is usually interpreted as ``time evolution'' in the Heisenberg picture, acting on the system with ``observable algebra'' $\Acal$. In general we do not demand any continuity for $t\mapsto\alpha_t$. A particular family of invariant states on a quantum dynamical system $(\Acal,\alpha)$ is given by \textit{$(\alpha,\beta)$-KMS (Kubo-Martin-Schwinger) states} \cite{Kubo1957,Martin-Schwinger1959,HHW1967,Bratteli-Robinson2}, which are used to describe quantum systems in thermal equilibrium at temperature $\beta^{-1}$.

\begin{definition}[KMS state]
	\label{def:kms-state}
	A state $\omega$ on a quantum dynamical system $(\Acal,\alpha)$ is called a \textit{$(\alpha, \beta)$-KMS state} (or $\beta$-KMS state with respect to $\alpha$, or $\beta$-KMS state on $(\Acal,\alpha)$) for $\beta\in\R\setminus\{0\}$, the \textit{KMS parameter}, if for every $A, B \in \Acal$ there exists a function $z \mapsto F_{A,B}^\beta (z)$ that is bounded and continuous on $\{z \in \C \, : \, 0 \leq \sgn(\beta) \Im (z) \leq |\beta| \}$ and analytic on the interior of this strip, such that $F_{A,B}^\beta (t) = \omega(A \alpha_t (B))$ and $F_{A,B}^\beta (t+i\beta) = \omega(\alpha_t (B) A)$ for all $t\in\R$.
\end{definition}

\begin{definition}[Unitary implementation]
	\label{def:unitary-imp}
	Let $\omega$ be a state on a C$^\ast$-algebra $\Acal$ with GNS representation $(\pi_\omega,\Hcal_\omega,\Omega_\omega)$. A $^\ast$-automorphism $\kappa$ on $\Acal$ is said to be \textit{unitarily implemented in the GNS representation of $\omega$} if there exists a unitary operator $U$ on $\Hcal_\omega$ such that $\pi_\omega (\kappa(A)) = U\pi_\omega (A)U^{-1}$ for all $A\in\Acal$. A group that is represented on $\Acal$ by a group of $^\ast$-automorphisms is \textit{unitarily implemented in $\omega$} if every representing $^\ast$-automorphism is unitarily implemented in the GNS representation of $\omega$ (resulting in a one-parameter group of unitary operators).
\end{definition}

The following proposition is compiled from \cite[Cor.\ 2.3.17]{Bratteli-Robinson1}, \cite[Lemma 4.26]{Pillet-Attal2006}, \cite[Prop.\ 5.1.17]{Khavkine-Moretti}, and \cite[Thm.\ 13.3]{Takesaki1970} (see also \cite[Cor.\ 5.3.4]{Bratteli-Robinson2} for the special case of strongly continuous dynamics).
\begin{proposition}[Induced W$^\ast$-dynamical system, Liouvillean]
	\label{prop:weak-Liouvillean}
	Let $\omega$ be an $\alpha$-invariant state on a quantum dynamical system $(\Acal,\alpha)$ with GNS representation $(\pi_\omega , \Hcal_\omega , \Omega_\omega)$. Then there is a one-parameter group $\{U_\omega (t)\}_{t\in\R}$ of unitary operators on $\Hcal_\omega$ that unitarily implements $\alpha$ in the GNS representation of $\omega$, uniquely determined by 
	\begin{gather}
		\label{eq:unitary-dynamics}
		\pi_\omega (\alpha_t (A)) = U_\omega (t) \pi_\omega (A) U_\omega (t)^{-1} \, , \quad U_\omega (t)\Omega_\omega = \Omega_\omega \quad \text{for all}\ A\in\Acal \, , \; t\in\R \, ,
	\end{gather}
	which uniquely defines a one-parameter group $\gamma=\{\gamma_t\}_{t\in\R}$ of $^\ast$-automorphisms
	\begin{gather*}
		\gamma_t = U_\omega (t) (\slot) U_\omega(t)^{-1}
	\end{gather*}
	on $\Mcal_\omega = \pi_\omega (\Acal)''$. The GNS vector state $\langle \Omega_\omega , (\slot) \Omega_\omega \rangle$, which represents $\omega$ on $\Mcal_\omega$, is a $\gamma$-invariant state.
	\begin{itemize}[leftmargin=*]
		\item The one-parameter group $\{U_\omega (t)\}_{t\in\R}$ is strongly continuous if and only if
		\begin{gather}
			\label{eq:continuity-condition}
			\R\ni t \mapsto \omega(A^\ast \alpha_t (A))\in\C \; \, \text{is continuous for all}\ A\in\Acal \, .
		\end{gather}
		In this case, $\gamma$ is generated by a unique (possibly unbounded) self-adjoint operator $L$ on $\Hcal_\omega$ by Stone's theorem, i.e.\ $\gamma_t = \e^{iLt} (\slot) \e^{-iLt}$ for all $t\in\R$, such that, by Eq.\ \eqref{eq:unitary-dynamics},
		\begin{gather}
			\label{eq:Liouvillean}
			\gamma_t (\pi_\omega (A)) = \pi_\omega (\alpha_t (A)) = \e^{iLt} \pi_\omega (A) \e^{-iLt} \quad \text{for all}\ A\in\Acal \, , \; t\in\R \, , \qquad  L\Omega_\omega = 0 \, ,
		\end{gather}
		and $\gamma$ is a W$^\ast$-dynamics on $\Mcal_\omega$. The operator $L$ is called the \emph{Liouvillean} of the W$^\ast$-dynamical system $(\Mcal_\omega , \gamma)$ induced by $\omega$.
		\item If $\omega$ is a $\beta$-KMS state on $(\Acal,\alpha)$ for $\beta \in \R\setminus\{0\}$, then $\gamma$ is generated by a Liouvillean as above, and $\langle \Omega_\omega , (\slot) \Omega_\omega \rangle$ is a faithful, normal $\beta$-KMS state on the W$^\ast$-dynamical system $(\Mcal_\omega , \gamma)$. In particular, the GNS vector $\Omega_\omega$ of the KMS state $\omega$ is cyclic and separating for $\Mcal_\omega$.
	\end{itemize}
\end{proposition}

The operator $L$ is also referred to as the ``$\omega$-Liouvillean'' or ``$\Omega_\omega$-Liouvillean'' in the theory of standard representations \cite{DJP2003,Pillet-Attal2006}, which coincides in the present case with the concept of a ``standard Liouvillean'' (see \cite[Prop.\ 2.4]{DJP2003}). For systems at positive temperature the Liouvillean is sometimes also called the ``thermal Hamiltonian'' \cite{BFS2000}. At zero temperature one typically writes $H$ instead of $L$ and refers to it as the ``Hamiltonian''; if $H\geq 0$, the state $\omega$ is called a \textit{ground state} \cite{Sanders2013,Khavkine-Moretti}. The continuity condition \eqref{eq:continuity-condition} is fulfilled, for instance, if $\alpha$ is strongly continuous, i.e.\ $\lim_{t\to 0} \|\alpha_t (A) -A\| = 0$ for all $A\in\Acal$ (for any invariant state $\omega$), if $\omega$ is a KMS state (for any dynamics $\alpha$, cf.\ Definition \ref{def:kms-state}), or if $(\Acal,\alpha)$ is a W$^\ast$-dynamical system and $\omega$ is a normal invariant state (cf.\ \cite[Sec.\ 4.4]{Pillet-Attal2006}).

\subsection{Primarity and mixing property}
\label{sec:prelim-mixing-primary}

Two main assumptions for the states we consider in our theorems are primarity \cite{Bratteli-Robinson1,Kadison-Ringrose-II,Haag1996} and the (strong) mixing property \cite{Bratteli-Robinson1,JP2,Pillet-Attal2006}. We introduce these notions, discuss their interpretation, and present sufficient criteria for their validity that can be checked in concrete models. Let $(\Acal,\alpha)$ be a quantum dynamical system.

\begin{definition}[Primary state]
	\label{def:primary}
	Let $\omega$ be a state on $\Acal$ with GNS representation $(\pi_\omega,\Hcal_\omega,\Omega_\omega)$. The state is called \textit{primary} (or a ``factor state'') if the induced von Neumann algebra $\Mcal_\omega = \pi_\omega(\Acal)''\subseteq\Bcal(\Hcal_\omega)$ is a \textit{factor}, i.e.\ has trivial center $\Mcal_\omega \cap {\Mcal_\omega}' = \C\mathds{1}$.
\end{definition}
Since pure states have irreducible GNS representations, they are primary. Primary states are precisely those states that do not possess a non-trivial decomposition into disjoint states. Two primary states are either disjoint or quasi-equivalent (see \cite[Prop.\ 10.3.12 (ii)]{Kadison-Ringrose-II}).

For strongly continuous dynamics $\alpha$, i.e.\ when $t\mapsto\alpha_t (A)$ is norm-continuous in $\Acal$ for all $A\in\Acal$, the convex set $K_{(\alpha,\beta)}$ of $(\alpha,\beta)$-KMS states on $\Acal$ is a weak-$^\ast$ compact subset of the set of states, and an $(\alpha,\beta)$-KMS state is primary if and only if it is an extremal point of $K_{(\alpha,\beta)}$ \cite[Thm.\ 5.3.30]{Bratteli-Robinson2}. In that case, primary $(\alpha,\beta)$-KMS states are those which are not a mixture of other $(\alpha,\beta)$-KMS states, and thus describe pure thermodynamic phases \cite[Sec.\ V.1.5]{Haag1996} (see also \cite{Emch-Knops-Verboven1970}). Without the continuity assumption on $\alpha$ (in particular for W$^\ast$-dynamics), $K_{(\alpha,\beta)}$ is not necessarily weak-$^\ast$ compact and may lack extremal points \cite[Sec.\ §4]{Takesaki-Winnink1973} (see also the comments in \cite[p.\ 118]{Bratteli-Robinson2}). \medskip

Without any assumption on the continuity of the dynamics, a sufficient, spectral condition for the primarity of a KMS state is given by the simplicity of the zero eigenvalue of the Liouvillean, which signifies the uniqueness (up to a phase) of the invariant vector under the induced W$^\ast$-dynamics (see Proposition \ref{prop:cone}). For completeness we include a proof in Appendix \ref{appendix:qds-add} (Proposition \ref{prop:cond-primary-KMS-app}). 
\begin{proposition}[Criterion for primarity of KMS states]
	\label{prop:cond-primary-KMS}
	Let $\omega$ be a KMS state on $(\Acal,\alpha)$ with GNS representation $(\pi_\omega , \Hcal_\omega , \Omega_\omega)$ and Liouvillean $L$ generating the dynamics $\gamma$ of the induced W$^\ast$-dynamical system $(\Mcal_\omega,\gamma)$. If $L$ has a simple eigenvalue at $0$, i.e.\ $\ker L=\C\Omega_\omega$, then $\omega$ is primary.
\end{proposition}

Now we turn to the cluster property that will be of interest in this work.
\begin{definition}[Mixing property]
	\label{def:mixing}
	An $\alpha$-invariant state $\omega$ on $(\Acal,\alpha)$ is called \textit{mixing}, or the triple $(\Acal,\alpha,\omega)$ \textit{satisfies the mixing property}, if
	\begin{gather*}
		\lim\limits_{t\to\infty} \omega(B\alpha_t (A))=\omega(A)\omega(B) \quad \text{for all}\ A,B\in\Acal \, .
	\end{gather*}
\end{definition}
The mixing property is one of many cluster properties, i.e.\ asymptotic factorization properties for expectation values of states under some dynamics (see \cite[Sec.\ 4.3]{Bratteli-Robinson1} and the introductory lecture notes \cite{Verboven1969}). In thermal quantum field theory this property appears as the ``time-clustering'' of correlation functions \cite{Bros-Buchholz1996}. Variants of the mixing property for Galilei-invariant fermionic systems have been investigated in \cite{Narnhofer-Thirring_mixing} (see also \cite{Narnhofer-Thirring_Galilei1991}). For a local quantum field theory on Minkowski spacetime it has been shown in \cite{Jaekel-Narnhofer-Wreszinski} that certain states can be weakly approximated by states that satisfy weak asymptotic abelianness, a property closely related to the mixing property (see at the end of this section).\medskip 

The mixing property is the quantum version of a similar property for classical ergodic systems \cite{Arnold-Avez}, where it constitutes one level of randomness in a whole hierarchy \cite{Berkovitz2006}. Ergodicity properties of invariant states on quantum dynamical systems can be characterized by spectral properties of the generator of the dynamics, a correspondence referred to as ``quantum Koopmanism'' \cite{JP2,Pillet-Attal2006} due to the analogy to Koopman's classical results \cite{Koopman1931} (see \cite[Sec.\ §9]{Arnold-Avez} and also \cite[Sec.\ 3.2]{Pillet-Attal2006}, \cite[Sec.\ VII.4]{Reed-Simon-I}).
\begin{proposition}[Quantum Koopman mixing criterion]
	\label{prop:cond-mixing}
	Let $\omega$ be an $\alpha$-invariant state on $(\Acal,\alpha)$ with GNS representation $(\pi_\omega , \Hcal_\omega , \Omega_\omega)$, such that $\Omega_\omega$ is separating for $\Mcal_\omega = \pi_\omega (\Acal)''$ (i.e., $\langle \Omega_\omega , (\slot)\Omega_\omega \rangle$ is faithful) and $\omega$ satisfies the continuity condition \eqref{eq:continuity-condition}. Let $L$ be the Liouvillean of the induced W$^\ast$-dynamical system $(\Mcal_\omega,\gamma)$. If the spectrum of $L$ is absolutely continuous apart from a simple eigenvalue at $0$, i.e.\ $\ker L=\C\Omega_\omega$, then $\lim_{t\to\infty} \e^{iLt} = |\Omega_\omega \rangle\langle \Omega_\omega |$ in weak operator topology and thus $(\Acal,\alpha,\omega)$ satisfies the mixing property.
\end{proposition}
The proof can be found in \cite[Sec.\ 4]{JP2}, among the spectral characterizations of other ergodicity properties; see also \cite[Lemma III.1]{BFS2000} and \cite[Cor.\ 4.66]{Pillet-Attal2006}. In \cite[Def.\ 4.1]{JP2} the mixing property is formulated for the Liouvillean in the GNS representation of a faithful, normal, invariant state on a W$^\ast$-dynamical system. In Proposition \ref{prop:cond-mixing}, the GNS representation of the faithful, normal, $\gamma$-invariant state $\langle \Omega_\omega , (\slot) \Omega_\omega \rangle$ on $\Mcal_\omega\subseteq\Bcal(\Hcal_\omega)$ is the identity morphism on $\Mcal_\omega$, hence the mixing property for $(\Mcal_\omega , \gamma , \langle \Omega_\omega , (\slot) \Omega_\omega \rangle)$, and thereby for $(\Acal,\alpha,\omega)$, is obtained from \cite[Thm.\ 4.2 \& Cor.\ 4.3]{JP2} by setting $\pi\equiv\id$.

We mention that the faithfulness of the state $\langle \Omega_\omega , (\slot) \Omega_\omega \rangle$ that extends $\omega$ to $\Mcal_\omega$ is essential for this spectral characterization of the mixing property. If $\omega$ is a $\beta$-KMS state for some $\beta>0$, faithfulness is guaranteed by Proposition \ref{prop:weak-Liouvillean}, and \eqref{eq:continuity-condition} is fulfilled, so Proposition \ref{prop:cond-mixing} can be applied. In the ``zero temperature'', ground state limit $\beta\to\infty$ faithfulness is lost (see \cite[Prop.\ 2]{Requardt-Wreszinski}), and Proposition \ref{prop:cond-mixing} is not valid for such states (see Remark 1 in \cite[Sec.\ 4]{JP2}). However, there is a generalization that applies to possibly non-faithful invariant states and concerns the spectral properties of a certain restriction of the Liouvillean \cite[Thm.\ 4.65]{Pillet-Attal2006}. \medskip 

For KMS states the mixing property is closely related to the stability under some perturbation of the dynamics \cite{Haag-Kastler-Trych-Pohlmeyer,Bratteli-Kishimoto-Robinson}. Most important for our purposes, the mixing property of KMS states is equivalent to the RTE property (see \cite[Thm.\ 4.2]{JP2}), to be discussed later in Section \ref{sec:th-rte} in the context of thermalization of Unruh-DeWitt detectors coupled to KMS states of a quantum field. Under the conditions of Proposition \ref{prop:cond-mixing}, it holds $\ker L=\C\Omega_\omega$ if and only if $(\Acal,\alpha,\omega)$ satisfies RTE in the ergodic mean (see Definition \ref{def:th-rte-general}) \cite[Thm.\ 4.2]{JP2}, which is weaker than mixing (respectively, RTE).\medskip

A cluster property that is closely related to mixing is weak asymptotic abelianness \cite{Doplicher-Kadison-Kastler-Robinson1967,Doplicher-Kastler-Stormer1969,Bratteli-Robinson1} (here we use the formulation appearing in, e.g., \cite{Narnhofer-Thirring_mixing,Narnhofer-Thirring_Galilei1991,Jaekel-Narnhofer-Wreszinski}). The triple $(\Acal,\alpha,\omega)$ is \textit{weakly asymptotically abelian} for an $\alpha$-invariant state $\omega$ if
\begin{gather}
	\label{eq:weak-asymp-abelian}
	\lim\limits_{t\to\infty} \omega(C[\alpha_t (A),B]C')=0 \quad \text{for all}\ A,B,C,C'\in\Acal \, ,
\end{gather}
or, equivalently, $\lim_{t\to\infty} [\pi_\omega (\alpha_t (A)),\pi_\omega (B)] = 0$ for all $A,B\in\Acal$ in weak operator topology under the GNS representation $\pi_\omega$ of $\omega$. If $\omega$ is primary and has a separating GNS vector, then $(\Acal,\alpha,\omega)$ is mixing if and only if it is weakly asymptotically abelian (see Lemma \ref{lem:weak-asymp-ab}, cf.\ \cite{Narnhofer-Thirring_Galilei1991} or \cite[p.\ 403]{Bratteli-Robinson1}). In our main Theorem \ref{thm:primary-disjoint} the mixing property can therefore be replaced by weak asymptotic abelianness. Nevertheless, the proofs will be formulated using the mixing property, which has a useful spectral characterization by Proposition \ref{prop:cond-mixing}.

\section{Disjointness of invariant states in different inertial frames}
\label{sec:disjointness}

In this section we present our main results on the disjointness of invariant states with respect to different inertial reference frames. For the proofs we will need some results regarding limits of $^\ast$-automorphism groups in quantum dynamical systems and algebraic quantum field theory, which are discussed in Section \ref{sec:lemmata-limits}. After that, we prove the disjointness of primary, mixing states in Section \ref{sec:disjoint-primary}. In Section \ref{sec:disjoint-KMS} we show that the result applies to states of the free scalar field satisfying the KMS condition with respect to inertial time translations.

\subsection{Three lemmata on limits}
\label{sec:lemmata-limits}

A cluster point (accumulation point) of a net (Moore-Smith sequence) in a topological space is a point for which there exists a subnet converging to it. There are several non-equivalent definitions for the notion of a subnet \cite[7.14 ff.]{Schechter}. For the sake of a concrete definition we may consider Kelley subnets: Let $(I,\preceq_I),(J,\preceq_J)$ be directed sets, and $(a_i)_{i\in I}$ a net in a topological space. A (Kelley) \textit{subnet} of $(a_i)_{i\in I}$ is given by $(a_{f(j)})_{j\in J}$, where $f:J\to I$ is a map such that for every $i \in I$ there exists $j_0 \in J$ so that $i \preceq_I f(j)$ whenever $j_0 \preceq_J j$. It is easy to see that if a net converges to a point, so does every subnet. Moreover, in a compact space every net has a cluster point. These fundamental properties are shared by any of the common definitions of a subnet, as explained in \cite[7.19, 15.38, 17.2]{Schechter}.\medskip

The first lemma is a consequence of the mixing property (Definition \ref{def:mixing}). The result and variants of it have been presented before in different contexts, see for example \cite{Haag-Kastler-Trych-Pohlmeyer,Narnhofer-Thirring_Galilei1991,Buchholz-Verch2015}. For systems with ``discrete'' dynamics given by a single non-trivial $^\ast$-automorphism an analogous statement is proven in \cite{Narnhofer-Thirring-Wiklicky} (see also \cite[Prop.\ (3.1)]{Thirring1992}).
\begin{lemma}
	\label{lem:mixing-limit}
	Let $\omega$ be a mixing state on a quantum dynamical system $(\Acal,\alpha)$ with a GNS representation $(\pi_\omega , \Hcal_\omega , \Omega_\omega)$ such that $\Omega_\omega$ is separating for $\Mcal_\omega = \pi_\omega (\Acal)''$. Then
	\begin{gather*}
		\lim\limits_{t\to\pm\infty} \pi_\omega (\alpha_t (A)) = \omega(A)\mathds{1} \quad \text{for all}\ A\in\Acal
	\end{gather*}
	in weak operator topology.
\end{lemma}
\begin{proof}
	Let $A\in\Acal$. We first prove the statement for the limit $t\to\infty$. The net $(\pi_\omega(\alpha_t (A)))_{t\in\R}$ is norm-bounded by $\|\pi_\omega(\alpha_t (A))\|\leq\|\alpha_t (A)\|=\|A\|$ (as $^\ast$-homomorphisms between C$^\ast$-algebras are norm-decreasing, and $^\ast$-automorphisms are isometric), thus it is contained in a norm-closed ball in $\Mcal_\omega$, which is compact in ultraweak (weak-$^\ast$) topology by the Banach-Alaoglu theorem. Hence the net has a cluster point $X \in\Mcal_\omega$ with respect to weak operator topology, i.e.\ there exists some directed set $I$ and a map $f:I\to\R$ defining a subnet $(\pi_\omega(\alpha_{f(i)} (A)))_{i\in I}$ that converges to $X$ in weak operator topology. In particular,
	\begin{gather*}
		\lim\limits_{i\in I} \langle \pi_\omega (B) \Omega_\omega , \pi_\omega(\alpha_{f(i)} (A)) \Omega_\omega \rangle = \langle \pi_\omega (B) \Omega_\omega , X \Omega_\omega \rangle
	\end{gather*}
	for all $B\in\Acal$. On the other hand, the mixing property of $\omega$ implies that
	\begin{gather*}
		\lim\limits_{t\to\infty} \langle \pi_\omega (B) \Omega_\omega , \pi_\omega(\alpha_{t} (A)) \Omega_\omega \rangle = \omega(A)\omega(B^\ast) \equiv \langle \pi_\omega (B) \Omega_\omega , \omega(A)\mathds{1} \Omega_\omega \rangle
	\end{gather*}
	for all $B\in\Acal$. Since $(\langle \pi_\omega (B) \Omega_\omega , \pi_\omega(\alpha_{f(i)} (A)) \Omega_\omega \rangle)_{i\in I}$ is a subnet of the convergent net $(\langle \pi_\omega (B) \Omega_\omega , \pi_\omega(\alpha_{t} (A)) \Omega_\omega \rangle)_{t\in\R}$ the limits must coincide, i.e.\ $\langle \pi_\omega (B) \Omega_\omega , (X-\omega(A)\mathds{1}) \Omega_\omega \rangle = 0$. This holds for all $B\in\Acal$, and $\Omega_\omega$ is cyclic, so $(X-\omega(A)\mathds{1}) \Omega_\omega = 0$, which implies $X=\omega(A)\mathds{1}$ by the assumption that $\Omega_\omega$ is separating. As this applies to all cluster points, $\omega(A)\mathds{1}$ is the unique cluster point of $(\pi_\omega(\alpha_t (A)))_{t\in\R}$, which is a net in a weakly compact subspace and thus has to converge to $\omega(A)\mathds{1}$.
	
	For the limit $t\to -\infty$ we show the equivalent statement that $\lim_{t\to\infty} \pi_\omega (\alpha_{-t} (A)) = \omega(A)\mathds{1}$ in weak operator topology for $A\in\Acal$. Notice that the mixing property is equivalent to $\lim_{t\to\infty} \omega(\alpha_{-t} (A)B)=\omega(A)\omega(B)$ for all $A,B\in\Acal$, since $\omega(\alpha_{-t} (A)B) = \omega(A\alpha_t (B))$ by the $\alpha$-invariance of $\omega$. The proof therefore proceeds in a similar way as before. 
\end{proof}

Our second lemma concerns dynamical systems of von Neumann algebras. A somewhat analogous result for KMS states is \cite[Sec.\ §1,\ Thm.\ 4]{Herman-Takesaki1970} (see the discussion in Section \ref{sec:herman-takesaki}). We formulate the statement in a slightly more general form than needed for the proof of our main theorem. From now on we choose to write the concatenation of automorphisms as juxtaposition to simplify notation. Furthermore, two $^\ast$-automorphism groups $\{\alpha_t\}_{t\in\R}$ and $\{\alpha'_t\}_{t\in\R}$ are said to commute if $\alpha_t \alpha'_s = \alpha'_s \alpha_t$ for all $s,t\in\R$.
\begin{lemma}
	\label{lem:equality}
	Let $\Mcal\subseteq\Bcal(\Hcal)$ be a von Neumann algebra of operators on a Hilbert space $\Hcal$. Let $\varphi=\langle\Psi_\varphi , (\slot)\Psi_\varphi \rangle$ and $\psi=\langle\Psi_\psi , (\slot)\Psi_\psi \rangle$ for $\Psi_\varphi , \Psi_\psi \in \Hcal$ be vector states on $\Mcal$ that are invariant under commuting one-parameter $^\ast$-automorphism groups $\gamma^\varphi = \{\gamma^\varphi_t\}_{t\in\R}$ and $\gamma^\psi = \{\gamma^\psi_t\}_{t\in\R}$, respectively. Assume that there is a $^\ast$-subalgebra $\Ncal \subseteq \Mcal$ with $\Ncal''=\Mcal$ such that for every $X\in\Ncal$ the nets $(\gamma^\varphi_t (X))_{t\in\R}$, $(\gamma^\psi_t (X))_{t\in\R}$, respectively $(\gamma^\psi_{-t} \gamma^\varphi_t (X))_{t\in\R}$, have cluster points in $\C\mathds{1}$ in weak operator topology as $t\to\pm\infty$, respectively $t\to\infty$. Then $\lim_{t\to\pm\infty} \gamma^\varphi_t (X) = \varphi(X) \mathds{1}$ and $\lim_{t\to\pm\infty} \gamma^\psi_t (X) = \psi(X) \mathds{1}$ for all $X\in\Ncal$ in weak operator topology, and $\varphi=\psi$ on $\Mcal$.
\end{lemma}

\begin{proof}
	Fix $X\in\Ncal$. The existence of cluster points for the stated nets is guaranteed from the outset by the Banach-Alaoglu theorem; for example, the net $(\gamma^\varphi_t (X))_{t\in\R}$ over the directed set $(\R,\leq)$ (i.e.\ we consider the limit $t\to\infty$) is norm-bounded by $\|\gamma^\varphi_t (X)\|=\|X\|$ (as $^\ast$-automorphisms are isometric). The assumption is that the cluster points are multiples of the identity, so there exists some directed set $I$ and a map $f:I\to\R$ defining a subnet $(\gamma^\varphi_{f(i)} (X))_{i\in I}$ of $(\gamma^\varphi_t (X))_{t\in\R}$ such that $\lim_{i\in I} \langle v, \gamma^\varphi_{f(i)} (X)w\rangle = c^\varphi_X \langle v, w\rangle$ for some $c^\varphi_X \in\C$ for all $v,w\in\Hcal$. In particular, we get $c^\varphi_X = \lim_{i\in I} \langle \Psi_\varphi , \gamma^\varphi_{f(i)} (X) \Psi_\varphi \rangle = \varphi (X)$, as $\varphi$ is $\gamma^\varphi$-invariant. This shows that $\varphi(X)\mathds{1}$ is the unique cluster point and therefore (as in the proof of Lemma \ref{lem:mixing-limit}) $(\gamma^\varphi_t (X))_{t\in\R}$ converges to $\varphi(X)\mathds{1}$. Similarly, the net $(\gamma^\psi_t (X))_{t\in\R}$ converges to $\psi(X)\mathds{1}$ in weak operator topology. The proof for the limit $t\to -\infty$ is identical.
	
	For $X\in\Ncal$ the net $(\gamma^\psi_{-t} \gamma^\varphi_t (X))_{t\in\R}$ has a cluster point $c^{\psi\varphi}_X \mathds{1}$ for some $c^{\psi\varphi}_X \in\C$. Let $(\gamma^\psi_{-g(j)} \gamma^\varphi_{g(j)} (X))_{j\in J}$ be a subnet (given by a directed set $J$ and a map $g:J\to\R$) converging to $c^{\psi\varphi}_X \mathds{1}$. Then 
	\begin{gather*}
		c^{\psi\varphi}_X = \lim\limits_{j\in J} \langle \Psi_\psi , \gamma^\psi_{-g(j)} \gamma^\varphi_{g(j)} (X) \Psi_\psi \rangle = \lim\limits_{j\in J} \langle \Psi_\psi , \gamma^\varphi_{g(j)} (X) \Psi_\psi \rangle = \varphi(X)
	\end{gather*}
	by $\gamma^\psi$-invariance of $\psi$ and the weak operator limit properties (note that $(\gamma^\varphi_{g(j)} (X))_{j\in J}$ is a subnet of the convergent net $(\gamma^\varphi_t (X))_{t\in\R}$). On the other hand, 
	\begin{gather*}
		c^{\psi\varphi}_X = \lim\limits_{j\in J} \langle \Psi_\varphi , \gamma^\psi_{-g(j)} \gamma^\varphi_{g(j)} (X) \Psi_\varphi \rangle = \lim\limits_{j\in J} \langle \Psi_\varphi , \gamma^\psi_{-g(j)} (X) \Psi_\varphi \rangle = \psi(X) \, ,
	\end{gather*}
	using that $\gamma^\varphi$ and $\gamma^\psi$ commute, $\varphi$ is $\gamma^\varphi$-invariant, and the limit properties. One therefore concludes that $\varphi(X)=c^{\psi\varphi}_X =\psi(X)$ for all $X\in\Ncal$. Since $\Ncal\subseteq\Mcal$ is ultraweakly dense by von Neumann's bicommutant theorem \cite[Thm.\ 2.4.11]{Bratteli-Robinson1}, and the states $\varphi$ and $\psi$ are normal, they are already uniquely determined on $\Ncal$ and so $\varphi\restr_{\Ncal}=\psi\restr_{\Ncal}$ implies $\varphi=\psi$.
\end{proof}

Finally, we state (without proof) a cluster property regarding the asymptotic commutativity of observables of a local quantum field theory under spacelike translations.
\begin{lemma}[{\cite[Lemma IV.4.2]{Borchers1996}}]
	\label{lem:spacelike-abelian}
	Let $\Acal$ be the quasi-local algebra of an algebraic quantum field theory on $(1+d)$-dimensional Minkowski spacetime $\M\cong\R^{1+d}$ ($d \geq 1$), satisfying isotony, locality and translation covariance (assumptions (A)--(C) in Section \ref{sec:prelim-aqft}), where spacetime translations are represented on $\Acal$ by $^\ast$-automorphisms $\{\updelta_x\}_{x\in\R^{1+d}}$. Let $\pi:\Acal\to\Bcal(\Hcal)$ be a $^\ast$-representation of $\Acal$ on a Hilbert space $\Hcal$, with induced von Neumann algebra $\Mcal:=\pi(\Acal)''$. Then for every $A\in\Acal$, $X'\in(\Mcal\cap\Mcal')'$, and spacelike vector $a\in\R^{1+d}$, one has
	\begin{gather*}
		\lim\limits_{t\to\infty} [\pi(\updelta_{ta} (A)),X'] = 0
	\end{gather*}
	in weak operator topology.
\end{lemma}

\subsection{Disjointness of primary, mixing states}
\label{sec:disjoint-primary}

Consider a quantum field theory on $(1+d)$-dimensional Minkowski spacetime $\M=\R^{1,d}$ ($d \geq 1$), given by a family $\{\Acal(O)\}_{O\subset\M}$ of C$^\ast$-algebras satisfying (A)--(C) as introduced in Section \ref{sec:prelim-aqft}, with quasi-local algebra $\Acal$. We want to compare states that are invariant according to different inertial observers. Each observer measures time relative to the respective inertial rest frame of reference (Lorentz frame). The worldline of an inertial observer located at some point $x\in\M$ at time $0$ is given by $t\mapsto x + t\frameu$, where $\frameu$ is a future-directed timelike unit vector specifying the time direction of the observer's reference frame, which we also call $\frameu$ to keep the notation simple. Time evolution by the amount $t\in\R$ relative to $\frameu$ is therefore described by spacetime translation with $t\frameu$. By assumption (C) there is a one-parameter group of $^\ast$-automorphisms $\{\updelta_{t\frameu}\}_{t\in\R}$ on $\Acal$ describing the time evolution relative to $\frameu$ for observables.\medskip

Let $\frameu$ and $\framew$ be two linear independent future-directed timelike unit vectors, specifying the time directions of two inertial reference frames that are related by a non-trivial Lorentz boost. (Without loss of generality one may realize these reference frames using the isomorphism $\M\cong\R^{1+d}$ with coordinates $(x_0 , x_1 , x_2 , \ldots , x_d)$ on $\R^{1+d}$ adapted to the frame $\frameu$, such that $\frameu=(1,0,0,\ldots,0)$ and $\framew=(1-v_{\mathrm{rel}}^2)^{-1/2}(1,-v_{\mathrm{rel}},0,\ldots,0)$ for relative velocity $-1<v_{\mathrm{rel}}<1$, $v_{\mathrm{rel}} \neq 0$, along the $x_1$-axis.) The time translations relative to $\frameu$ and $\framew$ are represented by commuting one-parameter groups $\alpha^{(\frameu)} = \{\alpha_t^{(\frameu)}\}_{t\in\R}$ and $\alpha^{(\framew)} = \{\alpha_t^{(\framew)}\}_{t\in\R}$ of $^\ast$-automorphisms
\begin{gather*}
	\alpha_t^{(\frameu)} = \updelta_{t\frameu} \, , \quad \alpha_t^{(\framew)} = \updelta_{t\framew}
\end{gather*}
on $\Acal$, respectively. The proof of the following theorem, which provides conditions for the disjointness of two distinct invariant states on the quantum dynamical systems $(\Acal,\alpha^{(\frameu)})$ and $(\Acal,\alpha^{(\framew)})$, relies on the three lemmata presented in the previous section.

\begin{theorem}[Disjointness of primary, mixing states]
	\label{thm:primary-disjoint}
	Let $\sigma,\omega$ be two distinct states on $\Acal$ that are invariant under the commuting one-parameter $^\ast$-automorphism groups $\alpha^{(\frameu)},\alpha^{(\framew)}$, respectively. Assume that $\sigma,\omega$ are primary, $(\Acal,\alpha^{(\frameu)},\sigma)$ and $(\Acal,\alpha^{(\framew)},\omega)$ have the mixing property, and the GNS vectors of $\sigma,\omega$ are separating for the respective induced von Neumann algebras. Then $\sigma$ and $\omega$ are disjoint.
\end{theorem}

\begin{proof}
	Since $\sigma$ and $\omega$ are primary, they are either disjoint or quasi-equivalent (see \cite[Prop.\ 10.3.12 (ii)]{Kadison-Ringrose-II}). Suppose that $\sigma$ and $\omega$ are quasi-equivalent.\medskip 
	
	Let $(\pi_\sigma,\Hcal_\sigma,\Omega_\sigma)$ and $(\pi_\omega,\Hcal_\omega,\Omega_\omega)$ be the GNS representations of $\sigma$ and $\omega$, respectively, and denote the induced von Neumann algebras by $\Mcal_\sigma := \pi_\sigma(\Acal)''\subseteq\Bcal(\Hcal_\sigma)$ and $\Mcal_\omega := \pi_\omega(\Acal)''\subseteq\Bcal(\Hcal_\omega)$. Then there exists a $^\ast$-isomorphism $\Theta:\Mcal_\omega \to\Mcal_\sigma$ such that $\Theta\circ\pi_\omega = \pi_\sigma$ \cite[Thm.\ 2.4.26]{Bratteli-Robinson1}. By \cite[Thm.\ 7.2.9]{Kadison-Ringrose-II}, a $^\ast$-isomorphism between von Neumann algebras with cyclic and separating vectors is unitarily implemented, so quasi-equivalent states with separating GNS vectors are unitarily equivalent. Hence there is a unitary $U:\Hcal_\omega \to \Hcal_\sigma$ such that $\Theta = U(\slot)U^{-1}$ and $\pi_\sigma = U \pi_\omega (\slot)U^{-1}$. 
	
	Define the states 
	\begin{gather*}
		\tilde{\sigma} = \langle U^{-1} \Omega_\sigma , (\slot)U^{-1} \Omega_\sigma \rangle \, , \quad \tilde{\omega} = \langle\Omega_\omega , (\slot)\Omega_\omega \rangle
	\end{gather*}
	on $\Mcal_\omega$, which represent $\sigma$ and $\omega$ in the sense that $\sigma=\tilde{\sigma} \circ \pi_\omega$ and $\omega=\tilde{\omega} \circ \pi_\omega$ on $\Acal$. Since $U$ is unitary, $U^{-1} \Omega_\sigma$ is cyclic and separating for $\Mcal_\omega$, and the GNS representation of $\sigma$ is given by $(\pi_\omega , \Hcal_\omega , U^{-1} \Omega_\sigma )$.
	
	The induced one-parameter $^\ast$-automorphism groups $\gamma^{(\frameu)}$ and $\gamma^{(\framew)}$ on $\Mcal_\omega$ are defined by $\pi_\omega \circ \alpha_t^{(\frameu)} = \gamma_t^{(\frameu)} \circ \pi_\omega$ (by transporting the induced dynamics on $\Mcal_\sigma$ to $\Mcal_\omega$ via $\Theta$, or by directly applying Proposition \ref{prop:weak-Liouvillean}) and $\pi_\omega \circ \alpha_t^{(\framew)} = \gamma_t^{(\framew)} \circ \pi_\omega$. They commute, i.e.\ $\gamma_t^{(\frameu)} \gamma_s^{(\framew)} = \gamma_s^{(\framew)} \gamma_t^{(\frameu)}$ for all $s,t\in\R$, because $\alpha^{(\frameu)}$ and $\alpha^{(\framew)}$ commute, and the states $\tilde{\sigma}$ and $\tilde{\omega}$ are invariant under $\gamma^{(\frameu)}$ and $\gamma^{(\framew)}$, respectively.\medskip
	
	Let $X=\pi_\omega(A)\in\pi_\omega(\Acal)$. Then $\gamma_t^{(\frameu)} (X) = \pi_\omega(\alpha_t^{(\frameu)} (A))$ and $\gamma_t^{(\framew)} (X) = \pi_\omega(\alpha_t^{(\framew)} (A))$, and since $\sigma$ and $\omega$ are mixing and have separating GNS vectors, Lemma \ref{lem:mixing-limit} implies that $\lim_{t\to\pm\infty} \gamma_t^{(\frameu)} (X) = \tilde{\sigma}(X)\mathds{1}$ and $\lim_{t\to\pm\infty} \gamma_t^{(\framew)} (X) = \tilde{\omega}(X)\mathds{1}$. Furthermore, as $\frameu$ and $\framew$ are linear independent future-directed timelike vectors, $\alpha_t^{(\frameu,\framew)} := \alpha_{-t}^{(\framew)} \alpha_t^{(\frameu)} = \updelta_{t(\frameu-\framew)}$, which is represented under $\pi_\omega$ by $\gamma_t^{(\frameu,\framew)} := \gamma_{-t}^{(\framew)} \gamma_t^{(\frameu)}$ on $\Mcal_\omega$, represents a translation along the spacelike vector $\frameu-\framew$ for $t\in\R$. The net $(\gamma_t^{(\frameu,\framew)} (X))_{t\in\R}$, where $\gamma_t^{(\frameu,\framew)} (X) = \pi_\omega(\alpha_t^{(\frameu,\framew)} (A))$, is norm-bounded, so the Banach-Alaoglu theorem provides the existence of a weak cluster point $Y$, which is the limit of a subnet $(\pi_\omega(\alpha_{f(i)}^{(\frameu,\framew)} (A)))_{i\in I}$ for some directed set $I$ and map $f:I\to\R$. Spacetime translations in spacelike directions act asymptotically abelian by Lemma \ref{lem:spacelike-abelian}, thus $\lim_{t\to\infty} [\pi_\omega(\alpha_t^{(\frameu,\framew)} (A)),X'] = 0$ in weak operator topology for every $X'\in(\Mcal_\omega \cap {\Mcal_\omega}')'=\Bcal(\Hcal_\omega)$, where we use the primarity of $\omega$. For the subnet this implies
	\begin{gather*}
		[Y,X']=\lim\limits_{i\in I} [\pi_\omega(\alpha_{f(i)}^{(\frameu,\framew)} (A)),X'] = 0 \quad \text{for all}\ X'\in\Bcal(\Hcal_\omega)
	\end{gather*}
	in weak operator topology, and therefore $Y \in \C\mathds{1}$. Hence, for every $X\in\pi_\omega(\Acal)$, the net $(\gamma_t^{(\frameu,\framew)} (X))_{t\in\R}$ has a cluster point in weak operator topology that lies in $\C\mathds{1}$.\medskip
	
	Now we can apply Lemma \ref{lem:equality} (for $\Mcal\equiv\Mcal_\omega$, $\Ncal\equiv\pi_\omega(\Acal)$, $\varphi\equiv\tilde{\sigma}$, $\psi\equiv\tilde{\omega}$, $\gamma^\varphi \equiv \gamma^{(\frameu)}$, $\gamma^{\psi} \equiv \gamma^{(\framew)}$), which gives $\sigma=\tilde{\sigma} \circ \pi_\omega = \tilde{\omega} \circ \pi_\omega = \omega$, a contradiction to the assumption that $\sigma,\omega$ are distinct states on $\Acal$.
\end{proof}

\subsection{Disjointness of KMS states}
\label{sec:disjoint-KMS}

The disjointness result proven in Theorem \ref{thm:primary-disjoint} can be applied, in particular, to two KMS states relative to different (linear independent) inertial time directions $\frameu,\framew$. By Propositions \ref{prop:cond-primary-KMS} \& \ref{prop:cond-mixing}, primarity and the mixing property are satisfied if the Liouvillean (Proposition \ref{prop:weak-Liouvillean}) of each state has absolutely continuous spectrum, except for a simple zero eigenvalue. The other assumptions put forward in Theorem \ref{thm:primary-disjoint} are automatically fulfilled: KMS states have separating GNS vectors for the respective induced von Neumann algebras (see Proposition \ref{prop:weak-Liouvillean}). Furthermore, one usually does not have to assume that the two states are distinct, because (under a certain assumption we discuss below) an inertial KMS state of a local quantum field theory with respect to time translations in some inertial frame cannot satisfy the KMS condition with respect to time translations in any other inertial frame moving with constant non-zero relative velocity, irrespective of the KMS parameters, as shown by Sewell in \cite[Prop.\ 3.1]{Sewell2008} (see also \cite{Sewell-rep2009}).

Sewell's result relies on the rather mild assumption that the automorphism group of spacetime translations acts non-trivially in the GNS representation of the state, i.e.\ if $\omega$ is a state on the quasi-local algebra $\Acal$ of the theory and $x$ is a non-zero vector, then there exists some $A\in\Acal$ and $t\in\R$ such that $\pi_\omega (\updelta_{tx} (A)) \neq \pi_\omega (A)$ (see \cite[Def.\ 3.1]{Sewell2008}). We note that non-triviality of translations in spacelike directions always holds for primary states, except in the trivial case that the induced von Neumann algebra is $\C\mathds{1}$. Indeed, let $\omega$ be a primary state on $\Acal$ and $a$ any spacelike vector, and suppose that $\pi_\omega (\updelta_{ta} (A)) = \pi_\omega (A)$ for all $A\in\Acal$ and all $t\in\R$. For fixed $A\in\Acal$, primarity of $\omega$ and Lemma \ref{lem:spacelike-abelian} imply that the net $(\pi_\omega (\updelta_{ta} (A)))_{t\in\R}$ has a cluster point in weak operator topology lying in $\C\mathds{1}$ (using the same argument as in the proof of Theorem \ref{thm:primary-disjoint}). By assumption the corresponding subnet converging to that cluster point is a constant net with value $\pi_\omega (A)$, hence $\pi_\omega (A)\in\C\mathds{1}$. As this holds for all $A\in\Acal$ one concludes that the translations in a spacelike direction in the GNS representation of a primary state $\omega$ can only be trivial if $\pi_\omega (\Acal)'' = \C\mathds{1}$.\medskip

We therefore arrive at the following corollary to Theorem \ref{thm:primary-disjoint}.
\begin{corollary}[Disjointness of inertial KMS states]
	\label{cor:disjointness-kms}
	Let $\sigma,\omega$ be two KMS states on $\Acal$ with respect to the inertial time translation $^\ast$-automorphism groups $\alpha^{(\frameu)},\alpha^{(\framew)}$ associated to two inertial frames with linear independent time directions specified by $\frameu,\framew$, respectively. If $\sigma$ and $\omega$ are primary, $(\Acal,\alpha^{(\frameu)},\sigma)$ and $(\Acal,\alpha^{(\framew)},\omega)$ satisfy the mixing property, and the states are distinct, then they are disjoint.
\end{corollary}
One can see this as an extension of Sewell's result: Under the stated conditions, a KMS state with respect to $\frameu$ is not only non-KMS with respect to $\framew$, but it is not even ``close to equilibrium'' in the sense of being in the folium of a KMS state relative to $\framew$. We remark that the degree of ``non-equilibrium'' of a KMS state viewed from a moving inertial frame can be quantified by the maximum amount of work that can be performed in a cyclic process, which is formalized using the property of ``semipassivity'' \cite{Kuckert2002}.\medskip

As an example for the disjointness of inertial KMS states, we now consider the quasi-free KMS states of the free real scalar (Klein-Gordon) field of mass $m\geq 0$ on four-dimensional Minkowski spacetime $\M=\R^{1,3}$, given by a net $\{\Acal_F (O)\}_{O\subset\M}$ of C$^\ast$-algebras satisfying assumptions (A)--(C). The construction of the theory and its inertial KMS states in terms of one-particle structures is summarized in Appendix \ref{appendix:scalar-field}, and we use the notation introduced there. The KMS states are defined in Appendix \ref{appendix:scalar-field}, without loss of generality, using coordinates such that the time direction is $(1,\vec{0})$ (Eq.\ \eqref{eq:dynamics-KG}). The free scalar field on $\M$ also satisfies the Poincar\'{e} covariance condition (C') of Section \ref{sec:prelim-aqft}. By applying Lorentz boosts one obtains KMS states with respect to time translations in different inertial frames. If $\omega_{F,\beta}$ is the unique $\beta$-KMS state in a frame $\frameu$, and $\Uplambda\in\Lorentz$ is a boost represented on $\Acal$ by the $^\ast$-automorphism $\updelta^{\Uplambda}$, then $\omega_{F,\beta} \circ (\updelta^\Uplambda)^{-1}$ is the unique $\beta$-KMS state with respect to time translations $\updelta^{\Uplambda} \circ \updelta_{t\frameu} \circ (\updelta^\Uplambda)^{-1} = \updelta_{t\Uplambda\frameu}$ in the frame $\Uplambda\frameu$ (cf.\ Eq.\ \eqref{eq:boost-translation-transf} and Definition \ref{def:kms-state}). A consequence of the disjointness of inertial KMS states is that the automorphisms $\updelta^{\Uplambda}$ relating the KMS states of different inertial frames are not implemented as unitary operators in the GNS representation (see also Section \ref{sec:breaking}). Note also that the proof of disjointness only makes use of the translation covariance condition (C) (by virtue of Theorem \ref{thm:primary-disjoint}); Poincar\'{e} covariance (C') is not needed. 

As a side remark, we also mention that in relativistic field theory the natural concept for KMS states is the relativistic KMS condition, first introduced in \cite{Bros-Buchholz}, as it incorporates the time direction of the reference frame under consideration. The analyticity requirements in the relativistic KMS condition are stronger than in the regular KMS condition, but they are satisfied by the KMS states of the free scalar field \cite{Bros-Buchholz}.\medskip

Let us now check that the conditions of primarity and mixing are satisfied by the inertial KMS states of the scalar field. We do this by employing the criteria from Propositions \ref{prop:cond-primary-KMS} \& \ref{prop:cond-mixing}, which demonstrates the usage of these general results.
\begin{lemma}[KMS states of the free scalar field are primary and mixing]
	\label{lem:KMS-scalar-primary-mixing}
	Let $\omega$ be a quasi-free KMS state on the quasi-local algebra $\Acal_F$ of the free scalar field on four-dimensional Minkowski spacetime $\M$ with respect to the $^\ast$-automorphism group $\alpha$ of time translations in an inertial reference frame. Then $\omega$ is primary, and $(\Acal_F , \alpha , \omega)$ satisfies the mixing property.
\end{lemma}

\begin{proof}
	As in Appendix \ref{appendix:scalar-field} we choose Minkowski coordinates $(t , \xu)$ on $\M\cong\R^4$ such that the time direction of the reference frame is given by $(1,\vec{0})$ with dynamics $\alpha_F$ (Eq.\ \eqref{eq:dynamics-KG}), and we use the setup presented in that appendix, with $\omega=\omega_{F,\beta}$ defined by Eq.\ \eqref{eq:field-KMS-state} (for any $\beta>0$), and $m\geq 0$. The Schr\"{o}dinger operator $-\Delta+m^2$ on $L^2 (\R^3,\Diff 3\xu)$ has the spectrum $[m^2,\infty)$ which is absolutely continuous (see, e.g., \cite{Teschl2014}), hence the self-adjoint operator $\ham_0 = (-\Delta+m^2)^{1/2}$ has purely absolutely continuous spectrum $[m,\infty)$ by spectral calculus. The one-particle Hamiltonian $\ham=\ham_0 \oplus -\ham_0$ on $\Kcal = L^2 (\R^3,\Diff 3\xu) \oplus L^2 (\R^3,\Diff 3\xu)$ (Eq.\ \eqref{eq:one-particle-Hamiltonian-KMS}) also has purely absolutely continuous spectrum. In the Fock space GNS representation of $\omega$, the Liouvillean is given by $L=\dGamma(\ham)$ (Eq.\ \eqref{eq:Liouv-Fock}), which, by \cite[Thm.\ 5.3]{Arai}, has the spectrum $\R$ that is absolutely continuous except for a simple (embedded) eigenvalue at $0$ corresponding to the Fock vacuum vector. Primarity and mixing hence follow from Propositions \ref{prop:cond-primary-KMS} \& \ref{prop:cond-mixing}. 
\end{proof}
A complementary elaboration of the mixing property in terms of one-particle structures is given in Appendix \ref{appendix:comments-mixing}.

\begin{theorem}[Disjointness of inertial KMS states of the free scalar field]
	\label{thm:KMS-disjoint}
	Let $\sigma,\omega$ be two quasi-free KMS states on the quasi-local algebra $\Acal_F$ of the free scalar field on four-dimensional Minkowski spacetime with respect to the inertial time translation $^\ast$-automorphism groups $\alpha^{(\frameu)},\alpha^{(\framew)}$ associated to two inertial frames with linear independent time directions specified by $\frameu,\framew$, respectively. Then $\sigma$ and $\omega$ are disjoint.
\end{theorem}

\begin{proof}
	Lemma \ref{lem:KMS-scalar-primary-mixing} shows that the KMS states $\sigma,\omega$ are primary and $(\Acal_F , \alpha^{(\frameu)} , \sigma)$ and $(\Acal_F , \alpha^{(\framew)} , \omega)$ satisfy the mixing property. Furthermore, $\sigma\neq\omega$ due to \cite[Prop.\ 3.1]{Sewell2008}, as explained at the beginning of this section. The statement therefore follows from Corollary \ref{cor:disjointness-kms}.
\end{proof}

\subsection{Lorentz boosts are not unitarily implemented}
\label{sec:breaking}

Temperature is a quantity to be properly assigned to a system in its rest frame, a conclusion that has been drawn from different perspectives in, e.g., \cite{Costa-Matsas1995,Landsberg-Matsas1996,Landsberg-Matsas2004,Sewell2008,Sewell-rep2009,Sewell2010} (cf.\ respective remarks in the early literature \cite{Cavalleri-Salgarelli,Callen-Horwitz}). The idea that the (instantaneous) rest frame has a distinguished role in the formulation of thermal properties also underlies the concept of local thermal equilibrium states \cite{Buchholz-Ojima-Roos}, which is used to determine the temperature in accelerated frames \cite{Buchholz-Solveen} (see Section \ref{sec:unruh}). By contrast, the Minkowski vacuum state is invariant under the Poincar\'{e} group (see Eq.\ \eqref{eq:vacuum-poincare-inv}) and therefore serves as a ground state in every inertial reference frame. 

These observations express the spontaneous breakdown of Lorentz symmetry in KMS states (cf.\ \cite[Sec.\ V.1.5]{Haag1996}, and \cite{Ojima2004,Ojima2005} for a general, abstract discussion of broken symmetries in quantum systems): The Fourier transformed form of the KMS condition reveals a Lorentz non-invariant factor, which implies that inertial KMS states are not invariant under Lorentz boosts; therefore, boosts cannot be implemented by unitary operators in the GNS representation such that the GNS vector is invariant, as boosts do not commute with inertial time translations \cite{Ojima1986}. The disjointness of primary, mixing states relative to different inertial reference frames allows to show that a unitary implementation of Lorentz boosts in these states is generally prohibited, regardless of the failure of invariance, i.e.\ even when permitting unitary implementations that do not necessarily leave the GNS vector invariant (cf.\ Definition \ref{def:unitary-imp}). A similar result for automorphisms commuting with the dynamics of a KMS state is presented in \cite{Narnhofer1977}.\medskip

We again consider a quantum field theory $\{\Acal(O)\}_{O\subset\M}$ on $\M$ satisfying (A)--(C) stated in Section \ref{sec:prelim}. In addition, we assume the stronger covariance condition (C') under the proper orthochronous Poincar\'{e} group. The proper orthochronous Lorentz group $\Lorentz$ is represented on the quasi-local algebra $\Acal$ by $^\ast$-automorphisms $\{\updelta^{\Uplambda}\}_{\Uplambda\in\Lorentz}$. We say that a state $\omega$ is Lorentz boost invariant if $\omega\circ\updelta^{\Uplambda} = \omega$ for all boosts $\Uplambda\in\Lorentz$, in which case every $\updelta^{\Uplambda}$ is implemented by a unique unitary operator $U^{\Uplambda}$ on the GNS Hilbert space that leaves the GNS vector invariant, i.e.\ $U^{\Uplambda} \Omega_\omega = \Omega_\omega$ (see \cite[Cor.\ 2.3.17]{Bratteli-Robinson1}). As mentioned above, \cite{Ojima1986} showed that inertial KMS states are never Lorentz boost invariant. The following result is a corollary of Theorem \ref{thm:primary-disjoint}.
\begin{corollary}[Lorentz boosts are not unitarily implemented]
	\label{cor:lorentz-breaking}
	$\Lorentz$ is not unitarily implemented in states on $\Acal$ that are primary, mixing with respect to time translations in some inertial frame, have a separating GNS vector, and are not Lorentz boost invariant. In particular, this applies to the quasi-free inertial KMS states of the free scalar field on four-dimensional Minkowski spacetime.
\end{corollary}
\begin{proof}
	Let $\omega$ be a state on $\mathcal{A}$ that is not Lorentz boost invariant, with GNS representation $(\pi_\omega , \Hcal_\omega , \Omega_\omega)$ and separating $\Omega_\omega$. Assume that $\omega$ is primary, invariant under the one-parameter group $\alpha^{(\framew)} = \{\alpha_t^{(\framew)}\}_{t\in\R}$ of $^\ast$-automorphisms given by $\alpha_t^{(\framew)} = \updelta_{t\framew}$ for some inertial frame with time direction (future-directed timelike unit vector) $\framew$, and $(\Acal,\alpha^{(\framew)},\omega)$ is mixing. Let $\Uplambda\in\Lorentz$ be a Lorentz boost between the (linear independent) time directions $\frameu=\Uplambda^{-1} \framew$ and $\framew$ such that $\omega\circ\updelta^{\Uplambda} \neq \omega$. Suppose that $\Lorentz$ is unitarily implemented in $\omega$. Then there exists a unitary operator $U^\Uplambda$ on $\Hcal_\omega$ such that $\pi_\omega (\updelta^{\Uplambda} (A)) = U^\Uplambda \pi_\omega (A) (U^{\Uplambda})^{-1}$ for all $A\in\Acal$. The state $\sigma:=\omega\circ\updelta^{\Uplambda}$ on $\Acal$ is invariant under $\alpha^{(\frameu)} = \{\alpha_t^{(\frameu)}\}_{t\in\R}$ for $\alpha_t^{(\frameu)} = \updelta_{t\frameu}$ and has the GNS representation $(\pi_\omega , \Hcal_\omega , (U^{\Uplambda})^{-1} \Omega_\omega)$, thus $\sigma$ is also primary with separating GNS vector. The mixing property of $(\Acal,\alpha^{(\framew)},\omega)$ implies that $(\Acal,\alpha^{(\frameu)},\sigma)$ is mixing, for if $A,B\in\Acal$ we have
	\begin{gather*}
		\lim\limits_{t\to\infty} \sigma(B\alpha_t^{(\frameu)} (A)) = \lim\limits_{t\to\infty} \omega\left(\updelta^{\Uplambda} (B) \alpha_t^{(\framew)}  (\updelta^{\Uplambda} (A))\right) = \omega(\updelta^{\Uplambda} (A)) \omega(\updelta^{\Uplambda} (B)) \equiv \sigma(A)\sigma(B) \, .
	\end{gather*}
	Under these conditions Theorem \ref{thm:primary-disjoint} tells us that $\sigma$ and $\omega$ have to be disjoint. But $\sigma$ is unitarily equivalent (in particular, quasi-equivalent) to $\omega$ by construction, which proves the claim. The last sentence of the corollary follows from Theorem \ref{thm:KMS-disjoint}.
\end{proof}

\section{Disjointness: Discussion and review of literature}
\label{sec:discussion-result}

In this section some consequences of our results are presented and related literature is reviewed.

\subsection{Comments on disjointness}

\begin{itemize}[leftmargin=*]
	\item Theorem \ref{thm:primary-disjoint} establishes a superselection rule in that inertial reference frames partition the theory into different sectors given by quasi-equivalence classes (folia) of primary, mixing states (cf.\ \cite{Ojima2004,Ojima2005}).
	\item The disjointness of primary, mixing states is a global (large-scale) property of states on the quasi-local algebra of the theory. However, if these states are quasi-free Hadamard states \cite{Radzikowski1996,Khavkine-Moretti} (such as the inertial KMS states of the free scalar field \cite{Sahlmann-Verch}), then they are quasi-equivalent when restricted to local observable algebras associated to open, relatively compact regions \cite{Verch1994,DAntoni-Hollands}.
	\item The disjointness result for KMS states is independent of the KMS parameters. Even if the two KMS states in Corollary \ref{cor:disjointness-kms} or Theorem \ref{thm:KMS-disjoint} have the same KMS parameter $\beta<\infty$ (inverse temperature relative to the respective frame), they are disjoint. This is no longer the case when taking the limit $\beta\to\infty$ in the weak-$^\ast$ sense, since the KMS states in any frame approach the Poincar\'{e}-invariant Minkowski vacuum state $\omega_{\M}$:
	\begin{gather}
		\label{eq:vacuum-poincare-inv}
		\omega_{\M} \circ \updelta_{(x,\Uplambda)} = \omega_{\M} \quad \text{for all}\ (x,\Uplambda)\in\Poincare
	\end{gather}
	The Minkowski vacuum is therefore a ground state for every inertial observer \cite{Haag1996} (it is the unique quasi-free pure state that is $\Poincare$-invariant, see \cite[Prop.\ 5.2.33]{Khavkine-Moretti}).
	\item Complementary to the disjointness with respect to different inertial reference frames (irrespective of the KMS parameter), two KMS states with respect to the same dynamics but with different KMS parameters are disjoint as well, provided the von Neumann algebra induced by one of the states is of type $\mathrm{III}$ \cite{Takesaki1970-disj} (see also \cite[Thm.\ 5.3.35]{Bratteli-Robinson2}). The type $\mathrm{III}$ property (see \cite[Sec.\ V.2.4]{Haag1996} and \cite{Sorce-type-vN}, and references cited there) is typically fulfilled by the local and global von Neumann algebras induced by states of a local quantum field \cite{Yngvason2005}. For the inertial KMS states (of finite temperature) of the free scalar field on Minkowski spacetime the type $\mathrm{III}$ property of the induced von Neumann algebra can be obtained from the characterization proven in \cite{Hugenholtz-factor}, using that the von Neumann algebra is a factor (Lemma \ref{lem:KMS-scalar-primary-mixing}). (From the spectral properties of the Liouvillean one can even prove that the type is $\mathrm{III}_1$ using the classification of factors due to Connes \cite{Connes1973}, see \cite[Thm.\ 1.5]{kay-double-wedge}.) Some other results on the disjointness of states on a single dynamical system are given in \cite{Mueller-Herold} (disjointness of primary $\beta$-KMS states with different chemical potentials for every $\beta$), \cite[Thm.\ 5.3.30]{Bratteli-Robinson2} (disjointness of distinct primary $\beta$-KMS states with respect to strongly continuous dynamics for every $\beta$), and \cite[Thm.\ 4.3.19]{Bratteli-Robinson1} (disjointness of distinct ``centrally ergodic'' states with respect to a representation of a group by $^\ast$-automorphisms). 
	\item Corollary \ref{cor:lorentz-breaking} can be compared to a similar situation in supersymmetric theories. Given a (graded) C$^\ast$-dynamical system with an asymptotically abelian time evolution that commutes with a supersymmetry transformation, it has been shown in \cite{Buchholz-Ojima-supersymmetry,Buchholz-supersymmetry} that any time-invariant state for which supersymmetry can be implemented by operators on the GNS Hilbert space (called a ``super-regular'' state) has to be an excitation of a ground state. Consequently, supersymmetry is not implemented in the GNS representation of KMS states and their mixtures (the ``spontaneous collapse of supersymmetry'' in thermal states \cite{Buchholz-Ojima-supersymmetry}). Similarly, the Lorentz group is unitarily implemented in the Minkowski vacuum (and excitations or perturbations thereof, i.e.\ states in the folium), since it is $\Poincare$-invariant (Eq.\ \eqref{eq:vacuum-poincare-inv}). KMS states, on the other hand, single out a particular reference frame, which leads to the breakdown of Lorentz symmetry as well as supersymmetry.
\end{itemize}

\subsection{Comparison with the literature}
\label{sec:herman-takesaki}

A statement that is closely related to Corollary \ref{cor:disjointness-kms} has been previously considered by Herman \& Takesaki \cite{Herman-Takesaki1970}. 
\begin{proposition}[{\cite[Sec.\ §1,\ Thm.\ 4]{Herman-Takesaki1970}}]
	If $\sigma,\omega$ are two quasi-equivalent states on a C$^\ast$-algebra $\Acal$ satisfying the $(\beta=-1)$-KMS condition with respect to commuting one-parameter $^\ast$-automorphism groups, such that
	\begin{gather*}
		F_\omega := \{X \in \Mcal_\omega \, : \, \tau_t^\omega (X)=X \; \text{for all}\ t\in\R\}=\C\mathds{1} \, ,
	\end{gather*}
	where $\tau^\omega$ is the modular automorphism group associated to $\omega$ by Tomita-Takesaki theory, then $\sigma=\omega$.
\end{proposition}
From this result one can obtain the disjointness of primary, mixing KMS states of a quantum field theory with respect to time translations in different inertial frames, or even any commuting one-parameter $^\ast$-automorphism groups of dynamics. To retain the specific physical setup of our theorems, we again consider commuting inertial time translation $^\ast$-automorphism groups $\alpha^{(\frameu)},\alpha^{(\framew)}$ associated to distinct reference frames $\frameu,\framew$, respectively. Let $\sigma,\omega$ be quasi-equivalent, primary, mixing KMS states with respect to $\alpha^{(\frameu)},\alpha^{(\framew)}$, respectively. If $\sigma$ has the KMS parameter $\beta$, then it is a $(-1)$-KMS state with respect to the rescaled dynamics $\widetilde{\alpha}^{(\frameu)} :=\{\alpha_{-\beta t}^{(\frameu)}\}_{t\in\R}$. Similarly, $\omega$ can be viewed as $(-1)$-KMS state with respect to some rescaled dynamics $\widetilde{\alpha}^{(\framew)}$ of time translations in the frame $\framew$. The assumed quasi-equivalence allows to represent both $\sigma$ and $\omega$ and the dynamics $\widetilde{\alpha}^{(\frameu)},\widetilde{\alpha}^{(\framew)}$ on the induced von Neumann algebra $\Mcal_\omega$ of $\omega$ (cf.\ the proof of Theorem \ref{thm:primary-disjoint}). The KMS property implies that the GNS vector $\Omega_\omega$ is cyclic and separating, hence there are Tomita-Takesaki modular objects associated to the pair $(\Mcal_\omega , \Omega_\omega)$ such that the corresponding modular automorphism group $\tau^\omega$ equals the W$^\ast$-dynamics induced by $\widetilde{\alpha}^{(\framew)}$ (see Propositions \ref{prop:modular} \& \ref{prop:modular-KMS}). The mixing property for $(\Mcal_\omega , \tau^\omega , \langle\Omega_\omega , (\slot)\Omega_\omega \rangle)$ and Lemma \ref{lem:mixing-limit} show that $F_\omega \subseteq \C\mathds{1}$ for the fixed point subalgebra $F_\omega$ defined above, thus $F_\omega = \C\mathds{1}$ and \cite[Sec.\ §1,\ Thm.\ 4]{Herman-Takesaki1970} implies $\sigma=\omega$. Therefore, if $\sigma$ and $\omega$ are supposed to be distinct states, primarity implies that they have to be disjoint. For the free scalar field this gives an alternative proof of Theorem \ref{thm:KMS-disjoint}.\medskip

However, our results for KMS states in Section \ref{sec:disjoint-KMS} are derived from Theorem \ref{thm:primary-disjoint}, which shows the disjointness of any two primary states of an isotonous, local, translation-covariant algebraic quantum field theory with respect to inertial time translations, provided the mixing property is fulfilled. We point out sufficient conditions for these properties, which can be verified in concrete examples. Explicit properties of the modular automorphism group, or some relation between modular automorphisms and the dynamics, are not required. This sets our proof apart from the above argumentation based on \cite{Herman-Takesaki1970}, which applies to any commuting dynamics, but necessarily relies on the KMS property of the involved states. Let us explain this point in more detail by summarizing the proof of \cite[Sec.\ §1,\ Thm.\ 4]{Herman-Takesaki1970}:
\begin{itemize}[leftmargin=*]
	\item The dynamics of the KMS states $\sigma,\omega$ extend to the respective modular automorphism groups on the induced von Neumann algebras. Since the states are assumed to be quasi-equivalent, $\sigma$ can be implemented on the induced von Neumann algebra $\Mcal_\omega$ of $\omega$, and the modular automorphism group of $\sigma$ can be transported to $\Mcal_\omega$. By assumption, the resulting modular automorphism groups on $\Mcal_\omega$ commute.
	\item The commutativity implies the invariance of each state under the modular automorphism group of the other (see \cite[Sec.\ §1,\ Thm.\ 2]{Herman-Takesaki1970}). Hence there is a self-adjoint operator $H>0$ affiliated with $F_\omega$ so that $\sigma$ is represented on $\Mcal_\omega$ by the vector $H\Omega_\omega$ (see \cite[Thm.\ 15.2]{Takesaki1970}), which implies $\sigma=\omega$, as $F_\omega = \C\mathds{1}$.
\end{itemize}
The proof of \cite[Sec.\ §1,\ Thm.\ 4]{Herman-Takesaki1970} rests on the fact that for the two quasi-equivalent KMS states the dynamics are implemented on the von Neumann algebra level by commuting modular automorphism groups. This cannot be guaranteed if we instead start with quasi-equivalent, primary, mixing states $\sigma,\omega$ with separating GNS vectors, not necessarily satisfying the KMS condition, like in the proof of Theorem \ref{thm:primary-disjoint}. These states induce faithful normal states on $\Mcal_\omega$ which are KMS states with respect to the corresponding modular automorphism groups (cf.\ Proposition \ref{prop:modular}). But the induced W$^\ast$-dynamics for $\sigma,\omega$ on $\Mcal_\omega$ do not have to be related to their modular automorphism groups (as would be the case for KMS states by Proposition \ref{prop:modular-KMS}), and while the induced dynamics of each of the states on $\Mcal_\omega$ commutes with its respective modular automorphism group (see \cite[Sec.\ §1,\ Thm.\ 1]{Herman-Takesaki1970} or the general result \cite[Thm.\ 2.4]{Sirugue-Winnink1970}), and the induced dynamics commute by assumption, there is no reason for the two modular automorphism groups themselves to commute.\medskip 

Theorem \ref{thm:primary-disjoint} thereby provides a more general result on the disjointness with respect to different inertial reference frames, and we demonstrate the application to KMS states in the case of the scalar field. Our proof uses Lemma \ref{lem:equality}, which can be interpreted as an analog to the statement of Herman \& Takesaki. The ``tradeoff'' for not relying on the KMS condition is the assumed weak cluster point property for the combined dynamics $(\gamma^\psi_{-t} \gamma^\varphi_t (X))_{t\in\R}$, which is $(\gamma_t^{(\frameu,\framew)} (X))_{t\in\R}$ in the proof of Theorem \ref{thm:primary-disjoint}. However, as shown in the latter proof, this property is always satisfied for inertial time translations in a local translation-covariant quantum field theory due to spacelike asymptotic abelianness of the dynamics (Lemma \ref{lem:spacelike-abelian}). It might be worth investigating if the assumptions in Lemma \ref{lem:equality} can be relaxed or modified (e.g.\ by combining with the results in \cite{Herman-Takesaki1970}) to cover dynamics other than inertial time translations.

\section{(Non-)Thermalization of Unruh-DeWitt detectors}
\label{sec:detectors}

In this section we want to relate the disjointness result to asymptotic equilibrium properties of Unruh-DeWitt detectors \cite{Unruh1976,DeWitt}, i.e.\ non-relativistic two-level quantum systems, that interact with a quantum field while being in inertial or uniformly accelerating motion. Inspired by the approach in \cite{DeB-M} we cast the model of a detector coupled to and moving inertially relative to a thermal bath in the framework of open quantum systems. In Section \ref{sec:th-rte} the notion of ``thermalization'' for a quantum dynamical system is introduced. The definition is motivated by the common property of ``return to equilibrium'' (RTE), which we briefly review. An overview of the general structure of the coupled detector-field system is presented in Section \ref{sec:general-setup}. Using previous results from the literature we provide a sketch for the proof of return to equilibrium for the stationary detector in the inertial heat bath of a massless scalar field on four-dimensional Minkowski spacetime in Section \ref{sec:inertial-rte}. This serves as the basis for a comparison with the situation of a detector that moves inertially relative to the heat bath (Section \ref{sec:moving-inertial}), for which we argue that disjointness (Theorem \ref{thm:KMS-disjoint}) prohibits thermalization for a whole folium of initial detector-field states. Finally, we discuss the impact of our argument on the thermal interpretation of the Unruh effect in Section \ref{sec:unruh}.

\subsection{Thermalization property and return to equilibrium}
\label{sec:th-rte}

Consider a quantum dynamical system $(\Acal,\alpha)$, given by a C$^\ast$-algebra $\Acal$ and a one-parameter group $\alpha=\{\alpha_t\}_{t\in\R}$ of $^\ast$-automorphisms on it. Starting from a state representing the initial (time $t=0$) configuration of $(\Acal,\alpha)$, does the system eventually reach thermal equilibrium if one waits long enough? The following definition formalizes this property (see, e.g., \cite{JP2,BFS2000}, and further references in the subsequent discussion).

\begin{definition}[Thermalization and return to equilibrium]
	\label{def:th-rte-general}
	Let $\omega$ be a KMS state on a quantum dynamical system $(\Acal,\alpha)$, and $\omega'$ a state on $\Acal$ referred to as the \textit{initial state}. (If $(\Acal,\alpha)$ is a W$^\ast$-dynamical system, assume that $\omega,\omega'$ are normal.) Then $(\Acal,\alpha)$ \textit{thermalizes} to $\omega$, or $(\Acal,\alpha,\omega ; \omega')$ satisfies the \textit{thermalization property}, if
	\begin{gather}
		\label{eq:th}
		\lim\limits_{t\to\infty} \omega'(\alpha_t (A)) = \omega (A) \quad \text{for all}\ A \in \Acal \, .
	\end{gather}
	The quantum dynamical system $(\Acal,\alpha)$ \textit{returns} to the equilibrium state $\omega$, or $(\Acal,\alpha,\omega)$ satisfies the \textit{RTE (return to equilibrium) property}, if $(\Acal,\alpha,\omega ; \omega')$ satisfies the thermalization property for all initial states $\omega' \in \Fol(\omega)$. If the convergence in Eq.\ \eqref{eq:th} is replaced by ergodic mean convergence (a continuous generalization of convergence in the Ces\`{a}ro mean),
	\begin{gather*}
		\lim\limits_{t\to\infty} \frac{1}{t} \int_{0}^{t} \omega'(\alpha_{t'} (A)) \diff t' = \omega(A) \quad \text{for all}\ A \in \Acal \, ,
	\end{gather*}
	then one refers to thermalization and RTE \textit{in the ergodic mean}. 
\end{definition}
Thermalization and RTE are properties that are stronger than their ergodic mean counterparts. The subsequent discussion mainly focuses on these strong notions, although the interpretation and references also apply to the weaker notions in the ergodic mean.\medskip

The thermalization property for $(\Acal,\alpha,\omega ; \omega')$ means that if the initial state of the dynamical system $(\Acal,\alpha)$ is prescribed by $\omega'$, the expectation values of observables under the dynamics $\alpha$ approach the expectation values in the thermal equilibrium state $\omega$ at asymptotically late times. The considered systems encompass observable algebras that can be proper C$^\ast$-algebras as well as, more specifically, von Neumann algebras. We put no restrictions on the allowed family of initial states; in principle, $\omega'$ could be ``far from equilibrium'' (see the following sections). Whether thermalization to some KMS state is even conceivable in such situations depends on the given quantum dynamical system. In different approaches and formulations, thermalization problems have been considered for closed, finite-dimensional quantum systems, for which we refer to the reviews \cite{Eisert-Friesdorf-Gogolin2015,Gogolin-Eisert2016} with their comprehensive lists of references.\medskip 

In this work we are interested in a particular class of open quantum systems \cite{Breuer-Petruccione,Davies-open,Rivas-Huelga,Attal2006} that arise from coupling an Unruh-DeWitt detector (a ``small'', finite-dimensional quantum system) to a relativistic quantum field (an infinitely extended ``reservoir'') at positive temperature. This falls into the category of so-called Pauli-Fierz models (see \cite{Derezinski-Gerard,Derezinski-Jaksic2001,Derezinski-Jaksic-RTE,Moeller2014} and references therein, and \cite{Pauli-Fierz1938} for the historical origin). The coupling is introduced by a suitable perturbation of the dynamics of the uncoupled (combined, tensor product) system. One would expect that, at least for sufficiently weak and analytically ``regular'' coupling, the coupled system thermalizes to an equilibrium state at the ``reservoir'' temperature, as long as the initial state is sufficiently close to thermal equilibrium. This asymptotic dynamical stability is, in fact, a paradigmatic property of open quantum systems (cf.\ the ``Quasitheorem'' in the introductory section of \cite{Derezinski-Jaksic-RTE}), put on mathematical grounds by the RTE property: All states that are perturbations of (i.e.\ normal to) a KMS state (cf.\ below Definition \ref{def:folium}) will return to that equilibrium at late times. We do not attempt to review the extensive literature on the topic of RTE. A non-exhaustive list of works on RTE and proofs of the validity in various models and under different conditions is \cite{JP1,JP2,BFS2000,Derezinski-Jaksic-RTE,Merkli-Positive-Commutators,Froehlich-Merkli-Another-RTE,Merkli-condensate,DeB-M,Koenenberg2011,Koenenberg-Merkli-Song-2014,Merkli-Song}, with early works \cite{Lanford-Robinson1972,Robinson1973,Hume-Robinson1986} (on RTE and related problems) and overviews \cite{JP-thermal-relaxation,Derezinski-RTE-report2007} (see also \cite{Bratteli-Robinson2}).

Let us add some further remarks. The evolution of small quantum systems coupled to thermal reservoirs can be studied (under certain conditions) in terms of a Lindblad master equation for the reduced density matrix of the small system (see, e.g., \cite{Breuer-Petruccione}). Here we want to consider the thermalization of the full detector-field system under a sufficiently weak coupling. Some techniques and recent developments regarding the treatment of the (reduced) dynamical evolution of open quantum systems in different coupling strength regimes are discussed in \cite{Trushechkin2022}. We mention that, in contrast to Definition \ref{def:th-rte-general}, the term ``thermalization'' has been used in \cite{Merkli-Sigal-Berman} to describe the asymptotically thermal behavior of the small system alone, with the reservoir being traced out. Recently the RTE property has also been studied in the context of perturbative (interacting) quantum field theory \cite{Drago-Faldino-Pinamonti2018,Galanda-Pinamonti-Sangaletti2023}. \medskip 

Starting from Definition \ref{def:th-rte-general}, consider a quantum dynamical system $(\Acal,\alpha)$ and a KMS state $\omega$ (assumed to be normal if $(\Acal,\alpha)$ is a W$^\ast$-dynamical system). Let $(\Mcal_\omega , \gamma)$ be the induced W$^\ast$-dynamical system obtained from the GNS representation $(\pi_\omega , \Hcal_\omega , \Omega_\omega)$ of $\omega$, as sketched in Proposition \ref{prop:weak-Liouvillean}. There is a close relation between the RTE property for $(\Acal,\alpha,\omega)$ and spectral properties of the Liouvillean $L$ (Proposition \ref{prop:weak-Liouvillean}) generating the dynamics by $\gamma_t (X) = \e^{iLt} X \e^{-iLt}$ for $X\in\Mcal_\omega$, $t\in\R$. As shown in \cite[Thm.\ 4.2]{JP2}, RTE is equivalent to mixing (Definition \ref{def:mixing}), and Proposition \ref{prop:cond-mixing} therefore gives a criterion for RTE. As a reminder we recall it here, including the characterization of RTE in the ergodic mean. A short proof for the characterization of the RTE property is discussed in \cite[Lemma 3.6]{DeB-M}, and applications can be found in the mentioned references.
\begin{proposition}[Quantum Koopmanism for RTE {\cite{JP2,JP-thermal-relaxation,BFS2000,Pillet-Attal2006}}]
	\label{prop:quantum-koopman}
	$(\Acal,\alpha,\omega)$ satisfies the RTE property in the ergodic mean if and only if $\ker L=\C\Omega_\omega$, i.e.\ the Liouvillean $L$ has a simple eigenvalue at $0$ corresponding to the GNS vector $\Omega_\omega$ of $\omega$. If the spectrum of $L$ is absolutely continuous apart from a simple eigenvalue at $0$, then $(\Acal,\alpha,\omega)$ satisfies the RTE property.
\end{proposition}
In particular, this yields a spectral characterization for the (ergodic mean) RTE property of a W$^\ast$-dynamical system $(\Mcal,\gamma)$ with a normal KMS state $\langle\Omega,(\slot)\Omega\rangle$ with cyclic and separating vector $\Omega$.\medskip

The definition of (ergodic mean) RTE, the equivalence between RTE and the mixing property, and the spectral characterization in Proposition \ref{prop:quantum-koopman} naturally extend to any invariant state that induces a W$^\ast$-dynamical system with a Liouvillean and a separating GNS vector \cite{JP2} (see the assumptions of Proposition \ref{prop:cond-mixing}). Since this section concerns KMS states, RTE will be viewed as ``return to \textit{thermal} equilibrium'', thereby justifying our restriction to KMS states in Definition \ref{def:th-rte-general}.\medskip

\subsection{General setup}
\label{sec:general-setup}

In the following sections we want to look at the asymptotic state behavior (thermalization and RTE) of an Unruh-DeWitt detector that interacts with a quantum field in different states. For this we need to define the corresponding quantum dynamical system, which should represent the combined system consisting of a finite-dimensional quantum system and a quantum field, with dynamics describing the time evolution of the combined system under some interaction coupling. For the sake of simplicity we will consider a two-level quantum system coupled to the free massless scalar field. The general setup, which is motivated by \cite{DeB-M}, is as follows:

\paragraph{Unruh-DeWitt detector system $(\Acal_D , \alpha_D)$} This will be the ``probe'' in our setup. The C$^\ast$-algebra is $\Acal_D = \Bcal(\C^2)$, and $\alpha_D = \{\alpha_{D,t}\}_{t\in\R}$ is the one-parameter group of $^\ast$-automorphisms 
\begin{gather*}
	\alpha_{D,t} = \e^{iH_D t} (\slot) \e^{-iH_D t} \, , \quad H_D := \diag(E,0)
\end{gather*}
for the Hamiltonian $H_D$ on $\C^2$ with $E>0$. This describes a two-level system (qubit) with ground state at zero energy and excited state at the detector excitation energy $E>0$. The parameter $t\in\R$ stands for the proper time of the detector, whose physical interpretation --- the stationary, timelike trajectory that the detector travels along --- is determined by the time evolution of the quantum field with which the detector system is combined. In our case, $t$ coincides with the time parameter in some inertial reference frame (see the next paragraph). Hence the data of this model exclusively encompasses the detector's internal degrees of freedom (cf.\ \cite[Sec.\ §3.1]{Takagi1986}, and the comment below \cite[Eq.\ (13)]{DeB-M}).

\paragraph{Quantum field system $(\Acal_F , \alpha_F)$} This will be the ``heat bath'' or ``reservoir'' in our setup. The C$^\ast$-algebra $\Acal_F$ is the quasi-local algebra of the free massless scalar quantum field on four-dimensional Minkowski spacetime $\M=\R^{1,3}$, obtained as the Weyl algebra over the classical phase space of Klein-Gordon theory, as discussed in Appendix \ref{appendix:scalar-field}. For a uniform notation we attach the index ``$F$'' to objects related to the field. The dynamics $\alpha_F = \{\alpha_{F,t}\}_{t\in\R}$ describing the field's inertial time evolution relative to some inertial reference frame can be defined by $^\ast$-automorphisms $\alpha_{F,t}$ on $\Acal_F$ as in Eq.\ \eqref{eq:dynamics-KG}, where the coordinates are chosen such that the time direction of the reference frame is $(1,\vec{0})$.

\paragraph{Coupled detector-field system} The C$^\ast$-algebra describing the combination of detector and quantum field is the tensor product $\Acal_D \otimes \Acal_F$. The dynamics of the uncoupled system, with detector and quantum field evolving independently without interaction between them, is given by $\alpha_D \otimes \alpha_F = \{\alpha_{D,t} \otimes \alpha_{F,t}\}_{t\in\R}$. The dynamics of the coupled system on $\Acal_D \otimes \Acal_F$ is introduced, in principle, by ``perturbing'' the uncoupled dynamics $\alpha_D \otimes \alpha_F$, parametrized by a coupling parameter $\lambda\in\R$ such that $\alpha_D \otimes \alpha_F$ is recovered when $\lambda\to 0$. The resulting coupled dynamics describes the time evolution of the detector-field system under the influence of an interaction between detector and quantum field. As we will discuss in the next section, the construction of the coupled dynamics can be performed rigorously in the GNS representation of a quasi-free reference state for a monopole interaction, i.e.\ a linear coupling of the detector's monopole moment to a quantum field operator smeared with some coupling test function.

\subsection{RTE of a detector at rest in a heat bath}
\label{sec:inertial-rte}

Suppose the Unruh-DeWitt detector interacts with the quantum field and follows the worldline $\R\ni t\mapsto t\frameu$, the time direction of an inertial reference frame specified by a future-directed timelike unit vector $\frameu$. Let the quantum field be in the quasi-free $\beta$-KMS state $\omega_{F,\beta}$ on $\Acal_F$ with respect to the time translation relative to $\frameu$, so the field is in thermal equilibrium relative to the rest frame of the detector. The definition of $\omega_{F,\beta}$ is given in Appendix \ref{appendix:scalar-field} and Appendix \ref{appendix:JP-glued-rep}, and as done there and also earlier we choose Minkowski coordinates on $\M\cong\R^4$ such that $\frameu=(1,\vec{0})$. One expects that, for some sufficiently ``benign'' coupling between the detector and the field, the coupled system asymptotically approaches a $\beta$-KMS state in the sense of Definition \ref{def:th-rte-general}. For the mathematical analysis it is beneficial to consider the W$^\ast$-dynamical systems induced by the GNS representation of certain reference states. The formal constructions are motivated by \cite{DeB-M,Froehlich-Merkli-Another-RTE} and follow the typical setup of Pauli-Fierz models (see also \cite{Merkli-LSO-2007} for an overview, and further references in the discussion below Proposition \ref{prop:inertial-RTE}). \medskip 

For $\beta>0$ and the detector Hamiltonian $H_D = \diag(E,0)$, $E>0$, let 
\begin{gather*}
	\omega_{D,\beta} := \frac{\tr(\e^{-\beta H_D} (\slot))}{\tr \e^{-\beta H_D}}
\end{gather*}
be the corresponding Gibbs equilibrium state at temperature $\beta^{-1}$, which is the unique $(\alpha_D,\beta)$-KMS state on $\Acal_D = \Bcal(\C^2)$. The GNS representation of $\omega_{D,\beta}$ can be given as $(\pi_D,\Hcal_D,\Omega_{D,\beta})$, where 
\begin{gather}
	\label{eq:detector-GNS}
	\Hcal_D = \C^2 \otimes \C^2 \, , \quad \pi_D = (\slot) \otimes \mathds{1}_2 \, , \quad
	\Omega_{D,\beta} = \frac{\e^{-\beta E/2} v_+ \otimes v_+ + v_- \otimes v_-}{\sqrt{1+\e^{-\beta E}}}
\end{gather}
for the identity matrix $\mathds{1}_2 := \diag(1,1)$ on $\C^2$ and the standard basis $v_+ = (1,0)^{\mathsf{T}}$, $v_- = (0,1)^{\mathsf{T}}$ of $\C^2$ (the eigenbasis of $H_D$).

For the GNS representation of $\omega_{F,\beta}$, $\beta>0$, the unique $(\alpha_F,\beta)$-KMS state on $\Acal_F$ defined in Appendix \ref{appendix:scalar-field}, we use the \JP\ glued representation discussed in Appendix \ref{appendix:JP-glued-rep}: The GNS representation $(\pi_{F,\beta} , \Fcal(\Hcal_F) , \Omega_F)$ is realized as Weyl operators on the symmetric Fock space over the one-particle Hilbert space $\Hcal_F := L^2 (\R \times S^2 , \diff s \diff\varOmega)$, with Fock vacuum $\Omega_F$ (Lemma \ref{lem:glued-GNS}). The von Neumann algebra induced by the GNS representation of the $\beta$-KMS state $\omega_{D,\beta} \otimes \omega_{F,\beta}$ on the uncoupled system $(\Acal_D \otimes \Acal_F , \alpha_D \otimes \alpha_F)$ hence is given by
\begin{gather}
	\label{eq:von-Neumann-detector-field}
	\Mcal_\beta := \pi_D (\Acal_D)'' \otimes \pi_{F,\beta} (\Acal_F)'' \subset \Bcal(\Hcal)
\end{gather}
for
\begin{gather*}
	\Hcal := \Hcal_D \otimes \Fcal(\Hcal_F) = \C^2 \otimes \C^2 \otimes \Fcal\left(L^2 (\R \times S^2 , \diff s \diff\varOmega)\right) \, .
\end{gather*}
The uncoupled dynamics $\alpha_D \otimes \alpha_F$ is represented on $\Mcal_\beta$ by a one-parameter group $\gamma_0$ of $^\ast$-automorphisms 
\begin{gather*}
	\gamma_{0,t} := \e^{iL_0 t} (\slot) \e^{-iL_0 t}
\end{gather*}
for $t\in\R$, generated by the sum $L_0$ of the Liouvilleans for the detector and field:
\begin{gather*}
	L_0 := L_D \otimes \mathds{1}_{\Fcal(\Hcal_F)} + \mathds{1}_{\Hcal_D} \otimes L_F \, , \quad L_D := H_D \otimes \mathds{1}_2 - \mathds{1}_2 \otimes H_D \, , \quad L_F = \dGamma(s)
\end{gather*}
The operator $\dGamma(s)$ is the second quantization of the self-adjoint operator of multiplication with the function $(s,\varOmega)\mapsto s$ on $\Hcal_F$ (see Lemma \ref{lem:glued-GNS}). The operator $L_0$ is self-adjoint on $\Hcal_D \otimes \dom(\dGamma(s))$ (see again Lemma \ref{lem:glued-GNS}) and satisfies
\begin{gather*}
	L_0 \Omega_{0,\beta} = 0 \, , \quad \Omega_{0,\beta} := \Omega_{D,\beta} \otimes \Omega_F \, ,
\end{gather*}
because $L_D \Omega_{D,\beta} = 0$ and $L_F \Omega_F = 0$. The W$^\ast$-dynamical system $(\Mcal_\beta , \gamma_0)$ with Liouvillean $L_0$ describes the uncoupled detector-field system in the GNS representation $(\pi_D \otimes \pi_{F,\beta} , \Hcal , \Omega_{0,\beta})$ of the $(\alpha_D \otimes \alpha_F , \beta)$-KMS state $\omega_{D,\beta} \otimes \omega_{F,\beta}$ (cf.\ Proposition \ref{prop:weak-Liouvillean}).\medskip

The dynamics of the coupled system is implemented by adding an interaction term to the Liouvillean $L_0$ of the uncoupled dynamics that encodes the specifics of the coupling between detector and quantum field. Here we consider a monopole coupling, which has been frequently used to study Unruh-DeWitt detectors in quantum field theory \cite{Birrell-Davies,Unruh-Wald1984,Takagi1986,DeB-M,Louko-Satz,Waiting-for-Unruh} and the interaction between toy atom systems and radiation fields \cite{Leggett1987,JP1,JP2,JP-thermal-relaxation,Merkli-Positive-Commutators,Froehlich-Merkli-Another-RTE,Froehlich-Merkli-Thermal-Ionization,Merkli-Sigal-Berman} (cf.\ the spin-boson model discussed later in this section). It is a linear interaction of the form
\begin{gather}
	\label{eq:monopole-coupling-I}
	I_\beta := G\otimes \mathds{1}_2 \otimes \bigl(a(h_\beta)+a^\ast (h_\beta)\bigr)
\end{gather}
for a self-adjoint operator $G\in\Bcal(\C^2)$ representing the monopole moment operator of the detector, and creation and annihilation operators $a^\ast (\slot) , a(\slot)$ on the Fock space $\Fcal(\Hcal_F)$ associated to the ``form factor'' $h_\beta \in \Hcal_F$ defined in Eq.\ \eqref{eq:hbeta} for a test function $h\in\Cinf(\M;\R)$, so $a(h_\beta)+a^\ast (h_\beta)$ is the field operator on $\Fcal(\Hcal_F)$ associated to $h_\beta$. The operator $I_\beta$ is unbounded but affiliated with the von Neumann algebra $\Mcal_\beta$. The operator $G$ is chosen to be a multiple of the first Pauli matrix (with respect to the eigenbasis of $H_D$). We remark that RTE cannot occur for interactions with diagonal $G$ due to the suppression of energy exchange (see \cite{Merkli-Sigal-Berman} for details and references). Let us also define
\begin{gather}
	\label{eq:monopole-coupling}
	V_\beta := I_\beta - JI_\beta J = G\otimes \mathds{1}_2 \otimes \bigl(a(h_\beta)+a^\ast (h_\beta)\bigr) - \mathds{1}_2 \otimes \overline{G} \otimes \bigl(a(\e^{-\beta s/2} h_\beta)+a^\ast (\e^{-\beta s/2} h_\beta)\bigr) \, ,
\end{gather}
where the equality follows from Eq.\ \eqref{eq:j-hbeta} and $\overline{G}\in\Bcal(\C^2)$ is obtained by complex conjugating the components of $G$ in the standard basis. The operator
\begin{gather}
	\label{eq:modular-conj}
	J=J_D \otimes J_F
\end{gather}
is the modular conjugation associated to $(\Mcal_\beta , \Omega_{0,\beta})$ (see Proposition \ref{prop:modular}), where $J_D$ is the modular conjugation associated to $(\pi_D (\Acal_D)'',\Omega_{D,\beta})$ given by $J_D (v_1 \otimes v_2) = \overline{v_2} \otimes \overline{v_1}$ on $\Hcal_D$ (and continuously and anti-linearly extended, with the bar denoting the complex conjugation of components in the standard basis of $\C^2$), and $J_F$ is the modular conjugation on $\Fcal(\Hcal_F)$ associated to $(\pi_{F,\beta} (\Acal_F)'',\Omega_F)$ defined in Eq.\ \eqref{eq:modular-conj-field-new}.\medskip

Let $\lambda\in\R$, which will represent the coupling parameter for the interaction between detector and field. Using the Glimm-Jaffe-Nelson commutator theorem \cite[Thm.\ X.37]{Reed-Simon-II} one proves that $L_0 + \lambda I_\beta$ and $L_0 + \lambda V_\beta$ are essentially self-adjoint on $\dom(L_0)\cap\dom(I_\beta)$ and $\dom(L_0)\cap\dom(I_\beta)\cap\dom(JI_\beta J)$, respectively (e.g., along the lines of \cite[Prop.\ 2.1, Lemma 3.2]{JP1}). General results on perturbations of Liouvilleans, which are presented in \cite[Thm.\ 3.3 \& 3.5]{DJP2003}, show that the self-adjoint closures $\widetilde{L}_{\lambda,\beta} := \overline{L_0 + \lambda I_\beta}$ and
\begin{gather}
	\label{eq:Liouv-coupled}
	L_{\lambda,\beta} := \overline{L_0 + \lambda V_\beta}
\end{gather}
generate a W$^\ast$-dynamics $\gamma_\lambda = \{\gamma_{\lambda,t}\}_{t\in\R}$ on $\Mcal_\beta$ via
\begin{gather}
	\label{eq:gamma-coupled}
	\gamma_{\lambda,t} := \e^{iL_{\lambda,\beta} t} (\slot) \e^{-iL_{\lambda,\beta} t} = \e^{i\widetilde{L}_{\lambda,\beta} t} (\slot) \e^{-i\widetilde{L}_{\lambda,\beta} t} \, .
\end{gather}
The pair $(\Mcal_\beta,\gamma_\lambda)$ is the W$^\ast$-dynamical system describing the coupled detector-field system, and one recovers the uncoupled system with Liouvillean $L_0$ for $\lambda=0$. Although $\widetilde{L}_{\lambda,\beta}$ and $L_{\lambda,\beta}$ generate the same dynamics, we will use $L_{\lambda,\beta}$ from now on, which is the (so-called ``standard'' \cite{DJP2003}) Liouvillean of $(\Mcal_\beta,\gamma_\lambda)$. This is closely related to the construction of a KMS state for the coupled dynamics $\gamma_\lambda$, as we will explain next. \medskip

One can show that $\Omega_{0,\beta} \in \dom(\e^{-\beta (L_0 + \lambda I_\beta) /2})$ for all $\lambda\in\R$ by a Dyson series expansion of the exponential and standard Fock space estimates (adapting \cite[Prop.\ 3.4]{DeB-M}, see also \cite[Sec.\ III. E]{BFS2000}, \cite[Eq.\ (95) ff.]{Merkli-LSO-2007}). Araki's perturbation theory of KMS states \cite{Araki1973,Bratteli-Robinson2,FNV,BFS2000}, which is developed for unbounded perturbations in \cite{DJP2003}, then allows to construct, for any $\lambda\in\R$ and $\beta>0$, a unique $(\gamma_\lambda , \beta)$-KMS state on $\Mcal_\beta$ given by
\begin{gather}
	\label{eq:coupled-KMS}
	\varphi_{\lambda,\beta} := \langle \Omega_{\lambda,\beta} , (\slot) \, \Omega_{\lambda,\beta} \rangle \, , \quad \Omega_{\lambda,\beta} := \frac{\e^{-\beta (L_0 + \lambda I_\beta) /2} \,  \Omega_{0,\beta}}{\|\e^{-\beta (L_0 + \lambda I_\beta) /2} \, \Omega_{0,\beta} \|} \in \Hcal \, .
\end{gather}
The vector $\Omega_{\lambda,\beta}$ is cyclic and separating for $\Mcal_\beta$ and approaches (in the norm) the GNS vector $\Omega_{0,\beta}$ of the $\beta$-KMS state $\omega_{D,\beta} \otimes \omega_{F,\beta}$ of the uncoupled system as $\lambda\to 0$ \cite[Thm.\ 5.5]{DJP2003}. In this sense, the KMS state for the coupled system is a perturbation of the uncoupled one, and these states are ``close'' to each other, as both are represented by normal states on $\Mcal_\beta$. 

The operator $L_{\lambda,\beta}$ generating $\gamma_\lambda$ (Eq.\ \eqref{eq:gamma-coupled}) has the property (see \cite[Sec.\ 5.6]{DJP2003} or \cite[Eq.\ (3.30) ff.]{Merkli-Quantum-Markovian})
\begin{gather*}
	L_{\lambda,\beta} \Omega_{\lambda,\beta} = 0 \, ,
\end{gather*}
which uniquely characterizes $L_{\lambda,\beta}$ as the Liouvillean of $(\Mcal_\beta,\gamma_\lambda)$, viewed as the W$^\ast$-dynamical system induced by $\varphi_{\lambda,\beta}$ (see Eq.\ \eqref{eq:Liouvillean} in Proposition \ref{prop:weak-Liouvillean}). We can therefore apply ``quantum Koopmanism'' (Proposition \ref{prop:quantum-koopman}) in order to study the return of the coupled system to the equilibrium state $\varphi_{\lambda,\beta}$, as discussed below. This is the motivation behind considering $L_{\lambda,\beta}$ instead of $\widetilde{L}_{\lambda,\beta} = \overline{L_0 + \lambda I_\beta}$, i.e.\ adding to $I_\beta$ the term $-JI_\beta J$ that is affiliated with the commutant ${\Mcal_\beta}'$ (cf.\ Eqs.\ \eqref{eq:monopole-coupling-I} \& \eqref{eq:monopole-coupling}). Moreover, the operator $L_{\lambda,\beta}$ is the natural generator of $\gamma_\lambda$ from the perspective of the KMS condition in terms of the modular objects associated to $(\Mcal_\beta,\Omega_{\lambda,\beta})$ (see Propositions \ref{prop:modular} \& \ref{prop:modular-KMS}, and cf.\ \cite[Eqs.\ (3.47)--(3.48)]{Kay-Wald} and \cite[Thm.\ 2.13]{DJP2003}): The modular operator is $\e^{-\beta L_{\lambda,\beta}}$ (see \cite[Thm.\ 5.5]{DJP2003}), and the modular conjugation is given by $J$ defined in Eq.\ \eqref{eq:modular-conj} and satisfies $JL_{\lambda,\beta} J = -L_{\lambda,\beta}$. (That $J$ from Eq.\ \eqref{eq:modular-conj} is the modular conjugation associated to $(\Mcal_\beta,\Omega_{\lambda,\beta})$ for every $\lambda\in\R$ follows from the fact that, by \cite[Thm.\ 5.5]{DJP2003}, the vector $\Omega_{\lambda,\beta}$ lies in the natural cone associated to $(\Mcal_\beta,\Omega_{0,\beta})$ as defined in Proposition \ref{prop:cone}, and the uniqueness of standard representations induced by faithful states with cyclic vector representatives shown in \cite[Prop.\ 2.5.30 (2)]{Bratteli-Robinson1}; see also \cite[Thm.\ 2.5]{DJP2003}.) \medskip

If the coupling parameter is small enough and the interaction satisfies sufficient regularity conditions, the coupled system will exhibit RTE. Here we just sketch the general structure of the statement and its proof and refer to some literature in which the main technical parts have been worked out in formally similar models. A systematic review of the arguments for our model may be given in future work.
\begin{proposition}[RTE for detector at rest in a heat bath]
	\label{prop:inertial-RTE}
	For every $\beta>0$, there exists $\lambda_0 (\beta)>0$ such that $(\Mcal_\beta,\gamma_\lambda,\varphi_{\lambda,\beta})$ satisfies the RTE property for $0<|\lambda|<\lambda_0 (\beta)$.
\end{proposition}
The proof makes use of ``quantum Koopmanism'' (Proposition \ref{prop:quantum-koopman}). For sufficiently small $|\lambda|>0$ one shows that $\ker L_{\lambda,\beta} = \C\Omega_{\lambda,\beta}$ and that the rest of the spectrum of $L_{\lambda,\beta}$ on $\{\Omega_{\lambda,\beta}\}^\perp$ is absolutely continuous. The spectrum of $L_D = H_D \otimes \mathds{1}_2 - \mathds{1}_2 \otimes H_D$ (for $H_D = \diag(E,0)$) consists of the eigenvalues $0,\pm E$, where $0$ is doubly degenerate, and $L_F = \dGamma(s)$ has a simple zero eigenvalue embedded in an otherwise absolutely continuous spectrum on $\R$. One therefore proves that the degeneracy of the zero eigenvalue of the Liouvillean $L_0$ of the uncoupled system is not stable under addition of the interaction term $\lambda V_\beta$. (The instability of the eigenvalues $\pm E$ under the interaction then follows from a general result: If the Liouvillean of a KMS state has a simple zero eigenvalue, then there do not exist any non-zero eigenvalues \cite{Jaksic-Pillet2001}.)\medskip 

The representation introduced in this section and in the appendices enables the identification of our setup with the representation of models that appear in the literature. In fact, our setup is analogous to, e.g., \cite{Froehlich-Merkli-Another-RTE} or \cite[Secs.\ 3.2 \& 3.3]{Merkli-Quantum-Markovian} and \cite[Secs.\ 2.1 \& 2.2]{Merkli2021-2}, where the ``purification'' representation of an $N$-level system coupled to a thermal gas of massless bosons is summarized. In the spin-boson model, which corresponds to $N=2$, a spin-$\frac{1}{2}$ or qubit system (which could be seen as an idealization of an atom confined in a way that there are effectively only two discrete energy levels) interacts with a reservoir of massless bosons, often described by a scalar field (which could be a proxy for a phonon or photon field, disregarding polarization degrees of freedom \cite{Froehlich-Merkli-Another-RTE,Froehlich-Merkli-Thermal-Ionization}); see, e.g., \cite{Leggett1987,FNV,Huebner-Spohn-decay,Huebner-Spohn,JP1,JP2,JP-thermal-relaxation,BFS2000,Froehlich-Merkli-Another-RTE,Merkli2021-2} (and the references therein). The massless scalar field in our setup is given in a \JP\ glued representation of Araki-Woods type (see Appendix \ref{appendix:JP-glued-rep}), similar to the thermal reservoir in spin-boson models, for which the required spectral properties of the coupled Liouvillean have been derived in \cite{JP2,BFS2000,Derezinski-Jaksic-RTE,Merkli-Positive-Commutators,Froehlich-Merkli-Another-RTE} using Mourre theory, positive commutator, or spectral deformation techniques. We refer to these works for the full, technically involved proof of the spectral properties leading to Proposition \ref{prop:inertial-RTE}, including a derivation of the required ``Fermi Golden Rule'' \cite{JP1,BFS2000,Derezinski-Jaksic-RTE,Derezinski-Frueboes}, which controls the (partial) instability of eigenvalues under the interaction at order $\lambda^2$, and conditions on the ``form factor'' $h_\beta \in \Hcal_F$ (Eq.\ \eqref{eq:hbeta}) appearing in the interaction (Eq.\ \eqref{eq:monopole-coupling-I}). A similar derivation for the Unruh effect using spectral deformation theory can be found in \cite{DeB-M}, which is summarized in \cite[Sec.\ 2.4]{Merkli-LSO-2007}. 

We also want to mention the recent paper \cite{Tjoa-Gray2024}, which considers the relation between the spin-boson model and a gapless Unruh-DeWitt detector coupled to a scalar field and discusses the construction of interacting ground states for this model under sufficient regularity conditions on the coupling function.\medskip

Proposition \ref{prop:inertial-RTE} means that the coupled system of an Unruh-DeWitt detector and a heat bath relative to its rest frame thermalizes, as long as the initial state is sufficiently close to equilibrium. In particular, the reduced state of the detector asymptotically approaches $\varphi_{\lambda,\beta} ((\slot)\otimes\mathds{1})$, which at lowest order in the coupling parameter $\lambda$ is the Gibbs equilibrium state (with respect to the free detector dynamics generated by $L_D$) at temperature $\beta^{-1}$, i.e.\
\begin{gather*}
	\lim\limits_{t\to\infty} \varphi'(\gamma_{\lambda,t}((A\otimes\mathds{1}_2)\otimes\mathds{1})) = \langle \Omega_{D,\beta} , (A\otimes\mathds{1}_2) \Omega_{D,\beta} \rangle + \Ocal(\lambda) \equiv \frac{\tr(\e^{-\beta H_D} (A))}{\tr \e^{-\beta H_D}} + \Ocal(\lambda)
\end{gather*}
for all normal (density operator) states $\varphi'$ on $\Mcal_\beta$ and detector observables $A\otimes\mathds{1}_2 \in\pi_D (\Acal_D)'' = \Bcal(\C^2)\otimes\mathds{1}_2$, as $\Omega_{\lambda,\beta} = \Omega_{D,\beta} \otimes \Omega_F + \Ocal(\lambda)$ by the perturbation theory of KMS states \cite{DJP2003}.\medskip

The admissible initial states in the RTE property lie in the folium of $\varphi_{\lambda,\beta}$, which consists of the normal states on $\Mcal_\beta$ and therefore are extensions of states in the folium of the uncoupled KMS state $\omega_{D,\beta} \otimes \omega_{F,\beta}$ on $\Acal_D \otimes \Acal_F$ (see Eq.\ \eqref{eq:von-Neumann-detector-field}). The initial state reflecting the physical setting described at the beginning of this section is the product state given by an arbitrary detector density matrix state\footnote{Every density matrix state on $\Acal_D$ is represented by a cyclic vector in $\Hcal_D$ and therefore all detector states lie in the same folium. For if $\rho = \sum_{j\in\{1,2\}} \varrho_j |\psi_j\rangle\langle \psi_j|$ is a density matrix on $\C^2$, diagonalized in a basis $\{\psi_1,\psi_2\}$ of eigenvectors with eigenvalues $\varrho_1,\varrho_2 \geq 0$, the GNS representation of $\tr(\rho(\slot))$ is $(\pi_D,\Hcal_D,\Omega_\rho)$ with $\Omega_\rho = \sum_{j\in\{1,2\}} \sqrt{\varrho_j} \psi_j \otimes \overline{\psi_j}$, where the overline denotes the complex conjugation of vector components in the eigenbasis of $H_D$. For $\rho_\beta=(\tr \e^{-\beta H_D})^{-1} \e^{-\beta H_D}$ we obtain $\Omega_{\rho_\beta} = \Omega_{D,\beta}$ (Eq.\ \eqref{eq:detector-GNS}).} and the $\beta$-KMS state $\langle \Omega_F , (\slot) \Omega_F \rangle$ of the quantum field (representing $\omega_{F,\beta}$ on $\pi_{F,\beta} (\Acal_F)''$). We emphasize that the initial state does not have to be of product form, but could also be a normal state on $\Mcal_\beta$ that is correlated (entangled) across detector and field (cf.\ comments in \cite{Merkli-Quantum-Markovian}). This ``folium invariance'' of the asymptotic state behavior, which is manifest in the RTE property, is a strong ergodic feature of the system.

\subsection{Detector moving inertially relative to a heat bath}
\label{sec:moving-inertial}

As before, assume that the worldline of the Unruh-DeWitt detector is given by $\R\ni t\mapsto t\frameu$ in coordinates where the time direction is $\frameu=(1,\vec{0})$. The detector and quantum field dynamical systems $(\Acal_D , \alpha_D)$ and $(\Acal_F , \alpha_F)$ are defined as in the previous section, with the dynamics $\alpha_D \otimes \alpha_F = \{\alpha_{D,t} \otimes \alpha_{F,t} \}_{t\in\R}$ implementing the time evolution of the combined detector-field system in the time coordinate $t$ of the detector's rest frame $\frameu$.

Let $\Uplambda\in\Lorentz$ be a Lorentz boost mediating between $\frameu$ and a different inertial frame $\framew$ with constant non-zero velocity relative to $\frameu$, such that $\framew=\Uplambda\frameu$ (cf.\ the beginning of Section \ref{sec:disjoint-primary}). The boost $\Uplambda$ is represented by a $^\ast$-automorphism $\updelta^{\Uplambda}$ on $\Acal_F$. For $\beta'>0$, consider
\begin{gather}
	\label{eq:lorentz-kms-relation}
	\omega_{F,\beta'}^{(\framew)} = \omega_{F,\beta'} \circ \updelta^{\Uplambda^{-1}} \, ,
\end{gather}
which is the unique quasi-free $\beta'$-KMS state on $\Acal_F$ with respect to the dynamics $\alpha_F^{(\framew)} = \{\alpha_{F,t}^{(\framew)}\}_{t\in\R}$ relative to $\framew$ given by $\alpha_{F,t}^{(\framew)} = \updelta^{\Uplambda} \circ \alpha_{F,t} \circ \updelta^{\Uplambda^{-1}}$. For any KMS parameters $\beta,\beta'>0$, the two KMS states $\omega_{F,\beta} , \omega_{F,\beta'}^{(\framew)}$ are distinct \cite[Prop.\ 3.1]{Sewell2008} (see Section \ref{sec:disjoint-KMS} for a discussion). We continue to denote states and dynamics relative to the detector's rest frame $\frameu$ without any index, while quantities relative to $\framew$ are decorated accordingly. The von Neumann algebra induced by $\omega_{D,\beta} \otimes \omega_{F,\beta}$ is again denoted $\Mcal_\beta$ (see Eq.\ \eqref{eq:von-Neumann-detector-field}).

\begin{theorem}
	\label{thm:far-from-eq}
	Let $\omega_D$ be any detector state on $\Acal_D$, and let $\beta'>0$ be fixed. Any state in $\Fol(\omega_D \otimes \omega_{F,\beta'}^{(\framew)})$ is \emph{far from equilibrium} relative to the rest frame $\frameu$ of the detector, that is, it is not normal to any $(\alpha_D \otimes \alpha_F , \beta)$-KMS state $\omega_{D,\beta} \otimes \omega_{F,\beta}$ on $\Acal_D \otimes \Acal_F$ for $\beta>0$.
\end{theorem}
\begin{proof}
	From Theorem \ref{thm:KMS-disjoint} we can conclude that
	\begin{gather}
		\label{eq:not-in-folium}
		\Fol(\omega_{F,\beta'}^{(\framew)})\cap \bigcup\limits_{\beta>0} \Fol(\omega_{F,\beta}) = \emptyset
	\end{gather}
	(in particular, $\updelta^{\Uplambda^{-1}}$ in Eq.\ \eqref{eq:lorentz-kms-relation} cannot be unitarily implemented in the GNS representation of $\omega_{F,\beta'}$, see Corollary \ref{cor:lorentz-breaking}). Hence, for any $\beta>0$, there is no normal state representing any state in $\Fol(\omega_D \otimes \omega_{F,\beta'}^{(\framew)})$ on $\Mcal_\beta$. 
\end{proof}

Theorem \ref{thm:far-from-eq} suggests that the detector will not asymptotically approach a thermal equilibrium state (in proper time of the rest frame $\frameu$) if it is coupled to the massless scalar field prepared in the KMS state $\omega_{F,\beta'}^{(\framew)}$ (viewed as thermal equilibrium state relative to the frame $\framew$ moving inertially relative to $\frameu$) or any perturbation of it. More concretely, let $\omega'\in\Fol(\omega_D \otimes \omega_{F,\beta'}^{(\framew)})$ be an initial state of the combined detector-field system (for any detector state $\omega_D$), and suppose that the detector couples to the quantum field via some interaction. The corresponding dynamics should describe the time evolution relative to $\frameu$, as the coupled system is supposed to evolve along the worldline of the detector and we ask about thermality relative to the detector's rest frame. Since the initial state is not represented as a normal state on $\Mcal_\beta$ (for any $\beta>0$), we cannot directly adopt the setup from Section \ref{sec:inertial-rte}. It is therefore suggested to treat the problem on a representation-independent C$^\ast$-algebraic level.

Assume that, for $\lambda\in\R$, there exists a one-parameter group of $^\ast$-automorphisms $\alpha_\lambda = \{\alpha_{\lambda,t} \}_{t\in\R}$ describing the coupled dynamics on $\Acal_D \otimes \Acal_F$ relative to $\frameu$, such that 
\begin{itemize}
	\item[(i)] for all $t\in\R$, $\alpha_{\lambda,t}$ weakly approaches $\alpha_{D,t} \otimes \alpha_{F,t}$ as $\lambda\to 0$,
	\item[(ii)] there exists, for any $\beta>0$, a unique $(\alpha_\lambda,\beta)$-KMS state $\omega_{\lambda,\beta}$ on $\Acal_D \otimes \Acal_F$ depending continuously on $\lambda\in\R$ such that, in the weak sense, $\lim_{\lambda\to 0} \omega_{\lambda,\beta} = \omega_{D,\beta} \otimes \omega_{F,\beta}$, and $\Fol(\omega_{\lambda,\beta})=\Fol(\omega_{D,\beta} \otimes \omega_{F,\beta})$.
\end{itemize}
Assumptions (i) and (ii) replicate the typical properties of the coupled dynamics and corresponding KMS states based on the general framework of perturbation theory of dynamics and KMS states \cite{DJP2003}. In particular, the hypothetical KMS state $\omega_{\lambda,\beta}$ should be a perturbation of the uncoupled KMS state $\omega_{D,\beta} \otimes \omega_{F,\beta}$ relative to the detector's rest frame $\frameu$. \medskip 

The coupled system with initial state $\omega'\in\Fol(\omega_D \otimes \omega_{F,\beta'}^{(\framew)})$ approaches a thermal equilibrium state with respect to the rest frame $\frameu$ of the detector, if $(\Acal_D \otimes \Acal_F , \alpha_\lambda , \omega_{\lambda,\beta} ; \omega')$ satisfies the thermalization property for some $\beta>0$ and sufficiently small $|\lambda|>0$. But Eq.\ \eqref{eq:not-in-folium} tells us that all quantum field states in $\Fol(\omega_{F,\beta'}^{(\framew)})$ are as far as possible from any thermal equilibrium state relative to $\frameu$, and thus $\omega'\notin\Fol(\omega_{\lambda,\beta})$ for any $\beta>0$. Under these conditions the coupled detector-field system is expected not to satisfy the thermalization property (which, in turn, would imply the asymptotic thermality of the detector). As $\omega'$ and $\omega_{\lambda,\beta}$ are disjoint, so are $\omega' \circ \alpha_{\lambda,t}$ and $\omega_{\lambda,\beta}$ for all $t\in\R$ \cite[Lemma 2]{Hepp1972}. Hence the system would have to leave its initial folium at asymptotically late times in order to settle in a limit state that would be deemed thermal in the reference frame of the detector.\medskip

It is important to note that this discussion provides evidence, but not a definite proof of non-thermalization for a whole folium of initial states. Under special conditions, such as a strong or finely tuned interaction, the coupled detector-field system could still end up in a thermal equilibrium state or at least some invariant state (cf.\ the final paragraph of Section \ref{sec:th-rte}) relative to the detector's rest frame $\frameu$. From the examples of Hepp's study of coherence in quantum measurement processes in \cite[Sec.\ 4]{Hepp1972} one can deduce the possibility that a system prepared in a state $\omega'$ evolves, in the sense of Eq.\ \eqref{eq:th}, into a state outside $\Fol(\omega')$ (see \cite[Sec.\ 6.6]{Landsman2007}). The problem of finding conditions for such a special equilibration process in our case remains open and will not be treated here (see Section \ref{sec:conclusions} for an outlook). Our result shows that the situation is not one of \textit{return} to equilibrium and is not covered by quantum Koopmanism (Proposition \ref{prop:quantum-koopman}).\medskip 

Finally, let us note that a favorable choice for coupled dynamics $\alpha_\lambda$ satisfying, at least formally, assumptions (i) and (ii) would be a one-parameter group of $^\ast$-automorphisms that is generated, for any $\beta>0$, by the Liouvillean $L_{\lambda,\beta}$ (Eq.\ \eqref{eq:Liouv-coupled}) with monopole coupling $V_\beta$ (Eq.\ \eqref{eq:monopole-coupling}) in the GNS representation of $\omega_{D,\beta} \otimes \omega_{F,\beta}$, so the coupled dynamics is represented by $\gamma_\lambda = \{\gamma_{\lambda,t} \}_{t\in\R}$ defined in Eq.\ \eqref{eq:gamma-coupled} on the von Neumann algebra $\Mcal_\beta$ (Eq.\ \eqref{eq:von-Neumann-detector-field}). Consider $\alpha_\lambda$ inducing $\gamma_\lambda$ in this sense. Then one can assume that $\omega_{\lambda,\beta}$ in assumption (ii) has the GNS representation $(\pi_D \otimes \pi_{F,\beta} , \Hcal , \Omega_{\lambda,\beta})$, so it is represented on $\Mcal_\beta$ by the coupled equilibrium state $\varphi_{\lambda,\beta}$ defined in Eq.\ \eqref{eq:coupled-KMS}. Still, the RTE property of $(\Mcal_\beta,\gamma_\lambda,\varphi_{\lambda,\beta})$, which holds for every $\beta>0$ and sufficiently small $|\lambda|>0$ by Proposition \ref{prop:inertial-RTE}, cannot imply thermalization in the present situation due to the above argument.

We however remark that since the free field dynamics is not strongly continuous (see the paragraph below Eq.\ \eqref{eq:dynamics-KG}), the common Dyson series construction of a one-parameter group $\alpha_\lambda$ of $^\ast$-automorphisms on $\Acal_D \otimes \Acal_F$ inducing $\gamma_\lambda$ on $\Mcal_\beta$ ($\beta>0$) can only be performed in a regularized form. This has been shown in \cite{Froehlich-Merkli-Thermal-Ionization} for the interaction between an atom (with a single electron and a static nucleus) and black body radiation. Instead, one considers a ``time-averaged'' C$^\ast$-algebra for the field on which the dynamics is represented by a strongly continuous one-parameter group (cf.\ \cite{Fewster-Verch2003} for a similar construction). On the tensor product C$^\ast$-algebra of the combined detector-field system one can define a regularized interaction operator, from which the coupled dynamics is obtained via a Dyson series. In the GNS representation of the uncoupled $\beta$-KMS state the dynamics is represented by a one-parameter group of $^\ast$-automorphisms on $\Mcal_\beta$ that weakly approach $\gamma_{\lambda,t} = \e^{iL_{\lambda,\beta} t} (\slot) \e^{-iL_{\lambda,\beta} t}$ as the regularization is removed; see \cite[Sec.\ 2.1]{Froehlich-Merkli-Thermal-Ionization} for details.

\subsection{Comparison with the Unruh effect}
\label{sec:unruh}

Theorem \ref{thm:far-from-eq} allows to highlight the difference to the behavior of an Unruh-DeWitt detector that is uniformly accelerating through Minkowski spacetime $\M\cong\R^4$ and is coupled to the Minkowski vacuum state of the free scalar field. Such a detector moves along a hyperbolic worldline in the (right) Rindler wedge $\WR := \{(x_0,x_1,x_2,x_3)\in\M\,:\,|x_0|<x_1\}$, which is a submanifold of $\M$ equipped with the induced metric and causal structure (see, e.g., \cite{Rindler1977}). If the detector is coupled to the Minkowski vacuum and follows the orbit of a Lorentz boost describing its uniformly accelerated motion with proper acceleration $\acc>0$, then it will converge (in proper time) to a thermal equilibrium state with Unruh temperature $\acc/2\pi$. This is the Unruh effect \cite{Unruh1976} (see Section \ref{sec:intro} for more comments and references), proven rigorously in \cite{DeB-M} in terms of RTE of the coupled detector-field system.\medskip

Hence the detector appears to respond as if it was coupled to a heat bath at temperature $\acc/2\pi$. Does this mean that there is, from the viewpoint of the accelerated observer, an Unruh ``heat bath'' or a ``thermal gas'' of Rindler quanta (relative to Fulling-Rindler quantization of the scalar field on the Rindler wedge \cite{Fulling1973,Wald-book,Crispino-Higuchi-Matsas})? And is the detector in thermal equilibrium with this ``heat bath'' due to exchange of heat, so that the detector temperature is the bath temperature in the sense of the zeroth law of thermodynamics? On the observable algebra associated to $\WR$, the Minkowski vacuum state is a $2\pi$-KMS state with respect to time translations given by the Killing flow of Lorentz boosts (respectively $2\pi/\acc$-KMS in proper time parametrization) \cite{Fulling1973,Unruh1976,Bisognano-Wichmann,Bisognano-Wichmann2,Sewell-BW1980,Sewell1982}. Although it might be tempting to view this KMS parameter as the inverse temperature of the Minkowski vacuum relative to an accelerating reference frame, the physical interpretation needs to be distinguished from inertial KMS states, i.e.\ KMS states with respect to ordinary time evolution relative to an inertial reference frame. 

The interpretation of the latter as global thermal equilibrium states of inertial heat baths at positive temperature is adopted from KMS states of general quantum systems \cite{Haag1996,Bratteli-Robinson2}. They represent systems that are in thermal equilibrium in the sense of the zeroth law \cite{Kossakowski1977} (describing equilibration and coexistence under thermal contact, see also \cite{Roos1972}) and the second law \cite{Pusz-Woronowicz} (KMS states are ``passive'', meaning that energy extraction via cyclic processes is prohibited, see also \cite{Kuckert2002}). In that case, the inverse of the KMS parameter specifies an empirical and absolute temperature scale of the system. 

On the other hand, as mentioned in Section \ref{sec:intro}, this correspondence fails in non-inertial situations. For the Unruh effect this has been demonstrated in \cite{Buchholz-Solveen} based on a concept of local thermality for non-equilibrium states \cite{Buchholz-Ojima-Roos} (which has been developed for curved spacetimes in the framework of locally covariant quantum field theory \cite{Buchholz-Schlemmer2007,Solveen2010,Solveen2012}, has connections to quantum energy inequalities \cite{Schlemmer-Verch2008} and cosmology \cite{Verch-local-thermal-eq}, and is related to a local version of the KMS condition \cite{Gransee-Pinamonti-Verch-LKMS}). In the formalism of \cite{Buchholz-Ojima-Roos}, states are compared with global thermal equilibrium states via the expectation values of a class of ``local thermal observables'' at any spacetime point. As a result, the Minkowski vacuum turns out to have zero temperature relative to every (inertial or non-inertial) observer \cite{Buchholz-Solveen}, and studies of macroscopic properties as well as a comparison with Tolman's law have further corroborated the absence of a heat bath or thermal gas of ``Rindler quanta'' in the Minkowski vacuum \cite{Buchholz-Verch2015,Buchholz-Verch2016}. The thermalization in the Unruh effect can be attributed to quantum excitations that are created from the vacuum by the coupling and are picked up by the detector, which responds in an ergodic fashion; the picture is that work from the accelerating force is effectively converted into heat in a dissipative process, similar to friction \cite{Buchholz-Verch2015,Buchholz-Verch2016}. Consequently, Unruh-DeWitt detectors are not adequate thermometers (in the sense of the zeroth law) under the influence of external forces. In the words of \cite{Buchholz-Verch2015}, ``instead of indicating the temperature of the vacuum, the probe indicates its own temperature at asymptotic times''. \medskip

Our analysis highlights yet another aspect of the difference to the inertial case. The proper orthochronous Lorentz group $\Lorentz$ is unitarily implemented in the Minkowski vacuum state, which is invariant under Lorentz boosts (see Eq.\ \eqref{eq:vacuum-poincare-inv}), in contrast to inertial KMS states (cf.\ Corollary \ref{cor:lorentz-breaking}). In the Unruh effect, the uniformly accelerated detector is stationary relative to the Killing flow induced by the Lorentz boost isometries of the Rindler wedge, and the Minkowski vacuum state of the field is a $2\pi$-KMS state with respect to boosts and thus is invariant under the same flow. Hence the Unruh effect is observed in any Lorentz frame, and the accelerated detector does not measure any Doppler shifts in the vacuum, as noted in \cite{Candelas-Deutsch-Sciama} (see the quotation in Section \ref{sec:intro}). On the other hand, a detector that moves inertially relative to the thermal bath of a free massless scalar field interacts with an inertial KMS state that is far from equilibrium with respect to the time evolution along the trajectory of the detector (Theorem \ref{thm:far-from-eq}). In this sense, the breakdown of Lorentz symmetry can be observed in the asymptotic behavior of detectors, which distinguishes the alleged Unruh ``heat bath'' from a proper inertial heat bath.

\section{Conclusions and outlook}
\label{sec:conclusions}

In this paper, we showed that in a local, translation-covariant quantum field theory any two distinct states that are invariant relative to different inertial reference frames are disjoint, if they have separating GNS vectors, are primary, and satisfy the mixing property (Theorem \ref{thm:primary-disjoint}). This particularly applies to the inertial KMS states of the free scalar field (Theorem \ref{thm:KMS-disjoint}). The proof uses the algebraic framework of quantum field theory and quantum statistical mechanics. From the disjointness result we obtained a proof of the Lorentz symmetry breaking in primary, mixing states that are not Lorentz boost invariant (Corollary \ref{cor:lorentz-breaking}). The mentioned statements strengthen previous results by Herman \& Takesaki \cite{Herman-Takesaki1970}, Sewell \cite{Sewell2008,Sewell-rep2009}, and Ojima \cite{Ojima1986}.\medskip

A motivation for our analysis is the clarification of the behavior of detector systems in inertial and accelerated motion. The disjointness result supports the conclusion that an Unruh-DeWitt detector moving with constant velocity relative to a free massless scalar field in thermal equilibrium does not thermalize. The initial field state is far from equilibrium relative to the detector's rest frame, which indicates that the state of the coupled system will not approach a KMS state at asymptotically late times. Our discussion illustrates that the thermalization in the Unruh effect, where a uniformly accelerated detector reaches thermal equilibrium upon coupling to the Minkowski vacuum, has to be distinguished from the interaction with a relativistic thermal gas in an inertial frame (as concluded before in \cite{Buchholz-Solveen,Buchholz-Verch2015,Buchholz-Verch2016}). The difference is that Lorentz boosts are not unitarily implementable in the GNS representation of KMS states with respect to inertial time translations (as a consequence of disjointness), while the Minkowski vacuum state is invariant under Lorentz boosts and thus a uniformly accelerated detector thermalizes to the Unruh temperature in every boosted Lorentz frame (cf.\ the statement in \cite{Candelas-Deutsch-Sciama} quoted in Section \ref{sec:intro}).\medskip

The discussion of non-thermalization of inertially moving detector systems in the present work complements earlier results from the literature. At the level of perturbation theory Costa \& Matsas \cite{Costa-Matsas-background1995,Costa-Matsas1995} showed that an inertially moving detector does not attain thermal transition rates from the heat bath of a massless scalar field due to the Doppler shift of quanta from the field. The transition rate is a function of the relative velocity and the detector's energy gap, and corresponds to radiation with a spectrum that is not of Planck (black body) form. It reduces to a Planck spectrum for zero velocity, and vanishes when the velocity approaches the speed of light. A well-defined temperature, obtained from a transition rate corresponding to a Planck spectrum, only exists for a thermal bath relative to its inertial rest frame; there is no Lorentz transformation for temperature from the frame of the bath to the rest frame of the detector \cite{Costa-Matsas1995,Landsberg-Matsas1996,Landsberg-Matsas2004} (see also \cite{Sewell2010}). Regarding equilibrium states, Sewell showed in \cite{Sewell2008,Sewell-rep2009} that KMS states in one inertial reference frame are non-KMS with respect to an inertial frame moving relative to it with constant non-zero velocity. For an Unruh-DeWitt detector coupled to and moving relative to an inertial thermal reservoir it is reasonable to expect that it does not converge to a KMS state at late times: The KMS condition characterizes systems that are thermal in the sense of the zeroth law \cite{Kossakowski1977}, but the detector is coupled to a reservoir that is not in a KMS state relative to the detector's rest frame \cite{Sewell2008}.\medskip

These results suggest that the moving detector does not reach a thermal equilibrium state at late times. However, as emphasized in Section \ref{sec:intro}, it is important to note that recording excitation spectra is a physical procedure that differs from observing the thermalization in the state of the detector or the coupled detector-field system. The system could be driven into thermal equilibrium, even though the reservoir is not thermal from the viewpoint of the detector and the detector transition rate does not correspond to a Planck spectrum. On a related note, it is difficult to exclude thermalization on the basis of thermodynamic properties alone; after all, the asymptotic temperature of a thermalizing detector is not necessarily a pointer for the local temperature of the quantum field it couples to, as can be observed in non-inertial situations such as the Unruh effect (see Section \ref{sec:unruh}). Moreover, we can expect that the asymptotic behavior of the full, coupled system will depend on the details of the interaction between detector and quantum field.

For the moving detector our Theorem \ref{thm:far-from-eq}, which is based on the disjointness of KMS states relative to different inertial frames, provides strong indication for the non-thermalization under generic conditions. However, as we noted in Section \ref{sec:moving-inertial}, thermalization of the coupled system under specially fine-tuned, coupling-dependent conditions cannot be excluded. In light of the above comments, the question arises to which degree deviations from a thermal transition rate are ``visible'' in the asymptotic state behavior. The evolution of the reduced density matrix of the Unruh-DeWitt detector (i.e.\ taking the partial trace over the field part), expressed in an open quantum systems approach by a master equation, has been shown in \cite{Papadatos-Anastopoulos2020} to reflect the non-Planckian response derived by Costa \& Matsas \cite{Costa-Matsas1995}. However, the relation to non-thermalization of the full detector-field system in the context of our work is not obvious, especially when the invariance of the asymptotic state behavior under the choice of the (not necessarily factorizing) initial state from its folium, as suggested by Theorem \ref{thm:far-from-eq}, is to be retained. We save these investigations for future work.\medskip 

We also propose a possible connection to non-equilibrium steady states (NESS) \cite{Ruelle-NESS} (see also \cite{Jaksic-Pillet-NESS-math,Jaksic-Pillet-NESS,Merkli-Mueck-Sigal-NESS} and references therein). For a detector moving inertially relative to an inertial heat bath of a massless scalar field, the master equation for the reduced density matrix reveals the effect of a continuum of heat baths interacting with the detector in its rest frame, corresponding to temperatures in a bounded, velocity-dependent interval \cite{Papadatos-Anastopoulos2020}. This is reminiscent of the ``quench'' configuration of a NESS, arising from the thermal contact of two semi-infinite heat baths \cite{DLSB,Hack-Verch}. In upcoming work we will examine the similarities and differences between such a NESS and the state of the thermal quantum field from the perspective of the moving detector's rest frame based on the transition rate of the detector. \medskip

\bigskip

\mysepline

\paragraph{Acknowledgments}

The authors would like to express their gratitude to Detlev Buchholz, Christopher J.\ Fewster, and Marco Merkli for illuminating discussions, and to an anonymous referee for helpful comments on a previous version of the paper that led to improvements of the presentation. A.G.P.\ thanks the IMPRS and MPI for Mathematics in the Sciences, Leipzig, where most parts of this work have been carried out, for support and hospitality. R.V.\ gratefully acknowledges funding by the CY Initiative of Excellence (grant ``Investissements d'Avenir'' ANR-16-IDEX-0008) during his stay at the CY Advanced Studies.

\bigskip
\bigskip

\mysepline


\newpage

\appendix

\section{More on quantum dynamics and operator algebras}
\label{appendix:qds-add}

In this appendix we compile some additional material on quantum statistical mechanics and dynamical systems, continuing Section \ref{sec:prelim-qds} and using the notation introduced there.

\begin{proposition}[Tomita-Takesaki modular theory {\cite{Takesaki1970,Bratteli-Robinson1,Summers-TT2006}}]
	\label{prop:modular}
	Let $\Mcal\subseteq\Bcal(\Hcal)$ be a von Neumann algebra on a Hilbert space $\Hcal$ with cyclic and separating unit vector $\Omega\in\Hcal$. The \emph{Tomita-Takesaki modular objects} associated to the pair $(\Mcal,\Omega)$ are given by the \emph{modular conjugation} $J$, which is an anti-unitary involution on $\Hcal$, and the \emph{modular operator} $\varDelta$, which is positive and self-adjoint on $\Hcal$, uniquely determined by the property
	\begin{gather*}
		J\varDelta^{1/2} X \Omega = X^\ast \Omega \quad \text{for all}\ X\in\Mcal \, .
	\end{gather*}
	The modular objects satisfy $J=J^\ast = J^{-1}$, $J\Omega = \varDelta\Omega = \Omega$, $J\Mcal J = \Mcal'$, and $\varDelta^{it} \Mcal \varDelta^{-it} = \Mcal$ for all $t\in\R$. The \emph{modular automorphism group} is the one-parameter group of $^\ast$-automorphisms $\varDelta^{it} (\slot) \varDelta^{-it}$, $t\in\R$, on $\Mcal$. The faithful vector state $\langle \Omega, (\slot) \Omega \rangle$ is a KMS state on $\Mcal$ with KMS parameter $-1$ with respect to modular automorphisms.
\end{proposition} 
The significance of modular theory in the context of quantum field theory is discussed in \cite{Borchers-Tomita}.

If $\omega$ is a KMS state, its GNS vector is cyclic and separating for the induced von Neumann algebra (Proposition \ref{prop:weak-Liouvillean}), so one can associate Tomita-Takesaki modular objects to it. The following relation between the modular group and the dynamics of the KMS state is shown in \cite[Thm.\ 13.2 \& Thm.\ 13.3]{Takesaki1970} (see also \cite[Thm.\ 5.3.10]{Bratteli-Robinson2} and \cite[Thm.\ 2.2]{Summers-TT2006}). We use the terminology of Proposition \ref{prop:weak-Liouvillean}.
\begin{proposition}[Modular automorphism group for KMS states]
	\label{prop:modular-KMS}
	Let $\omega$ be a $\beta$-KMS state on a quantum dynamical system $(\Acal,\alpha)$ with GNS representation $(\pi_\omega , \Hcal_\omega , \Omega_\omega)$, inducing the W$^\ast$-dynamical system $(\Mcal_\omega , \gamma)$ with Liouvillean $L$. Let $\varDelta$ be the modular operator associated to $(\Mcal_\omega , \Omega_\omega)$. Then $\ln\varDelta = -\beta L$, i.e.\ the modular automorphism group coincides with the dynamics $\gamma$ up to rescaling:
	\begin{gather*}
		\varDelta^{it} (\slot) \varDelta^{-it} = \gamma_{-\beta t} \quad \text{for all}\ t\in\R \, .
	\end{gather*}
\end{proposition}

We state a useful characterization of the collection of invariant states on a W$^\ast$-dynamical system in terms of the Liouvillean (Proposition \ref{prop:weak-Liouvillean}), following \cite[Sec.\ 2.10]{DJP2003}.
\begin{proposition}[Characterization of normal invariant states]
	\label{prop:cone}
	Let $\omega$ be a state on a C$^\ast$-algebra $\Acal$ with GNS representation $(\pi_\omega , \Hcal_\omega , \Omega_\omega)$, such that $\Omega_\omega$ is separating for $\Mcal_\omega = \pi_\omega (\Acal)''$. Let $J$ be the modular conjugation associated to $(\Mcal_\omega , \Omega_\omega)$. The \textit{natural cone} associated to $(\Mcal_\omega , \Omega_\omega)$ is defined as
	\begin{gather*}
		\Hcal_\omega^+ := \overline{\{XJX\Omega_\omega : X \in\Mcal_\omega \}} \subset \Hcal_\omega \, .
	\end{gather*}
	For every normal state $\varphi$ on $\Mcal_\omega$ there is a unique unit vector $\hat{\Omega}_\varphi \in \Hcal_\omega^+$ such that $\varphi = \langle \hat{\Omega}_\varphi , (\slot) \hat{\Omega}_\varphi \rangle$. The vector $\hat{\Omega}_\varphi$ is cyclic for $\Mcal_\omega$ if and only if it is separating. If $\omega$ is an invariant state on a quantum dynamical system $(\Acal,\alpha)$ and induces a W$^\ast$-dynamical system $(\Mcal_\omega,\gamma)$ with Liouvillean $L$, then a normal state $\varphi$ on $\Mcal_\omega$ is $\gamma$-invariant if and only if $\hat{\Omega}_\varphi \in \ker L$. This provides a bijection between $\ker L \cap \Hcal_\omega^+$ and the collection of normal, $\gamma$-invariant states on $\Mcal_\omega$.
\end{proposition}
The result is part of the rich theory of ``standard forms'' of von Neumann algebras, discussed in \cite[Sec.\ 2.5.4]{Bratteli-Robinson1} and \cite{DJP2003} (see also the original literature cited therein). Notice that $\Omega_\omega \in \Hcal_\omega^+$ and $L\Omega_\omega=0$ by Proposition \ref{prop:weak-Liouvillean}, hence $0$ is a simple (non-degenerate) eigenvalue of $L$ if and only if the normal state $\langle \Omega_\omega , (\slot) \Omega_\omega \rangle$ extending $\omega$ to $\Mcal_\omega$ is the unique normal, $\gamma$-invariant state on $\Mcal_\omega$.\medskip

We present a proof of Proposition \ref{prop:cond-primary-KMS} that is based on the concepts introduced in Section \ref{sec:prelim-qds} and this appendix as well as the equivalence between the uniqueness of KMS states on a von Neumann algebra and factoriality shown in \cite[Prop.\ 5.3.29]{Bratteli-Robinson2}. Another proof can be found in \cite[Thm.\ 1]{kay-purification} for ``double KMS states'' (which could be rewritten to apply to our setup).
\begin{proposition}[Criterion for primarity of KMS states]
	\label{prop:cond-primary-KMS-app}
	Let $\omega$ be a KMS state on $(\Acal,\alpha)$ with GNS representation $(\pi_\omega , \Hcal_\omega , \Omega_\omega)$ and Liouvillean $L$ generating the dynamics $\gamma$ of the induced W$^\ast$-dynamical system $(\Mcal_\omega,\gamma)$. If $L$ has a simple eigenvalue at $0$, i.e.\ $\ker L=\C\Omega_\omega$, then $\omega$ is primary.
\end{proposition}
\begin{proof}
	Assume that $\omega$ is an $(\alpha,\beta)$-KMS state for $\beta>0$, and let $\Hcal_\omega^+$ be the natural cone associated to $(\Mcal_\omega , \Omega_\omega)$. Due to the bijection between $\ker L \cap \Hcal_\omega^+$ and the collection of normal, $\gamma$-invariant states (Proposition \ref{prop:cone}), the assumption $\ker L = \C\Omega_\omega$ implies that $\langle\Omega_\omega , (\slot)\Omega_\omega \rangle$ is the unique normal, $\gamma$-invariant state on $\Mcal_\omega$. That state is a $(\gamma,\beta)$-KMS state by Proposition \ref{prop:weak-Liouvillean}, thus a $(-1)$-KMS state with respect to the rescaled dynamics $\{\gamma_{-\beta t}\}_{t\in\R}$. Tomita-Takesaki theory (Proposition \ref{prop:modular-KMS}) tells us that the latter coincides with the modular automorphism group. Hence $\langle\Omega_\omega , (\slot)\Omega_\omega \rangle$ is the unique $(-1)$-KMS state on $\Mcal_\omega$ with respect to modular automorphisms. By \cite[Prop.\ 5.3.29]{Bratteli-Robinson2} this uniqueness is equivalent to the factor property of $\Mcal_\omega$.
\end{proof}

The following lemma shows the claim at the end of Section \ref{sec:prelim-mixing-primary} that the mixing property (Def.\ \eqref{def:mixing}) is equivalent to weak asymptotic abelianness (defined in Eq.\ \eqref{eq:weak-asymp-abelian}) for states that are primary and have a separating GNS vector. This and similar statements have appeared before in \cite{Narnhofer-Thirring_Galilei1991,Narnhofer-Thirring-Wiklicky,Thirring1992} (see also \cite[p.\ 403]{Bratteli-Robinson1}).
\begin{lemma}[Mixing versus weak asymptotic abelianness]
	\label{lem:weak-asymp-ab}
	Let $\omega$ be an $\alpha$-invariant state on a quantum dynamical system $(\Acal,\alpha)$ with GNS representation $(\pi_\omega , \Hcal_\omega , \Omega_\omega)$ and induced von Neumann algebra $\Mcal_\omega = \pi_\omega (\Acal)''$. If $\Omega_\omega$ is separating for $\Mcal_\omega$ and $(\Acal,\alpha,\omega)$ satisfies the mixing property, then it is weakly asymptotically abelian, i.e.\ $\lim_{t\to\infty} [\pi_\omega (\alpha_t (A)),\pi_\omega (B)] = 0$ for all $A,B\in\Acal$ in weak operator topology. If $\omega$ is primary, then weak asymptotic abelianness implies mixing.
\end{lemma}

\begin{proof}
	For the first statement, Lemma \ref{lem:mixing-limit} shows that $\lim_{t\to\infty} \pi_\omega (\alpha_t (A)) = \omega(A)\mathds{1}$ in weak operator topology for all $A\in\Acal$, which implies weak asymptotic abelianness. For the second one notices that, for every $A\in\Acal$, the net $(\pi_\omega (\alpha_t (A)))_{t\in\R}$ has weak cluster points in the center of $\Mcal_\omega$ due to weak asymptotic abelianness, which is $\C\mathds{1}$ by primarity. Using arguments similar to the ones in the proof of Lemma \ref{lem:equality}, the $\alpha$-invariance of $\omega$ shows that $\lim_{t\to\infty} \pi_\omega (\alpha_t (A)) = \omega(A)\mathds{1}$ in weak operator topology for every $A\in\Acal$, and the mixing property for $\omega$ follows (cf.\ \cite{Narnhofer-Thirring_Galilei1991}).
\end{proof}

We close this section by noting that the mixing property by itself (without the assumption in Lemma \ref{lem:weak-asymp-ab} that the GNS vector is separating) merely implies $\lim_{t\to\infty} \omega([\alpha_t (A),B])=0$ for all $A,B\in\Acal$, which is slightly weaker than weak asymptotic abelianness.

\section{Inertial KMS states of the free scalar field}
\label{appendix:scalar-field}

In this appendix we summarize the theory of the free scalar field of mass $m\geq 0$ on four-dimensional Minkowski spacetime $\M=\R^{1,3}$ using Weyl quantization and define the quasi-free KMS states of the theory. As these are established constructions and results, we will only present an overview. For details on the quantization of the free scalar field on globally hyperbolic spacetimes we refer to \cite{Dimock1980,Kay-Wald,Wald-book,Fewster-Rejzner-AQFT} (see also the summary in \cite[Sec.\ 2]{Verch1993}). The construction of quasi-free KMS states follows the formalism of one-particle structures \cite{Kay1978,kay-double-wedge,kay-purification,kay-uniqueness,Dimock-Kay1987,Kay-Wald}. This is a more abstract and elaborate approach compared to the specification of the thermal two-point function on a Fock space, but the implementation of the dynamics in the GNS representation is rather simple and concise on a formal level, which makes it well-suited for the purposes of this work.

\subsection{Quantum field}
\label{appendix:quantum-field-KG}

Let $P:=\square+m^2$ be the Klein-Gordon operator for mass $m\geq 0$, where $\square$ is the d'Alembert (wave) operator on Minkowski spacetime $\M$, and let $\Cinf(\M;\R)$ be the space of real-valued, smooth, compactly supported functions on $\M$. Let 
\begin{gather*}
	\class{\slot}:\Cinf(\M;\R) \to \Cinf(\M;\R)/P\Cinf(\M;\R) =: \Csymp
\end{gather*}
be the canonical quotient map. Let $E_P : \Cinf(\M ; \R)\to C^{\infty} (\M ; \R)$ be the advanced-minus-retarded Green operator associated to $P$ (also called Pauli-Jordan or causal propagator) \cite{Khavkine-Moretti,Baer-Ginoux-Pfaeffle}. It has the property that $\ker E_P = P\Cinf(\M;\R)$, and the map $\Csymp \to C^{\infty} (\M ; \R)$ induced by $E_P$ is a well-defined $\R$-linear isomorphism onto the space of real-valued smooth solutions of the Klein-Gordon equation (i.e.\ elements of $\ker P \cap C^\infty (\M;\R)$) with compact support on Cauchy surfaces. The bilinear form $(\class{f},\class{g})\mapsto \mathcal{E}(\class{f},\class{g}):=\int_{\M} f E_P g$ defines a real-valued symplectic form on $\Csymp$, and the symplectic space $(\Csymp,\mathcal{E})$ represents the classical phase space of the theory.

The C$^\ast$-algebra $\Acal_F$ is defined as the Weyl (CCR) algebra over the symplectic space $(\Csymp,\mathcal{E})$. It is generated by elements $W(\class{f})$, $f\in\Cinf(\M;\R)$, such that $W(-\class{f})=W(\class{f})^\ast$ and $W(\class{f})W(\class{g}) = \e^{-\frac{i}{2}\mathcal{E}(\class{f},\class{g})} \, W(\class{f+g})$ for all $f,g\in\Cinf(\M;\R)$, which determines $\Acal_F$ uniquely up to $^\ast$-isomorphism \cite{Bratteli-Robinson2,Baer-Ginoux-Pfaeffle,Petz-CCR,Dimock1980}. The local algebras $\Acal_F (O)$ for regions (open, bounded, causally convex subsets) $O\subset\M$ are the C$^\ast$-subalgebras of $\Acal_F$ that are generated by elements $W(\class{f})$ for test functions $f\in\Cinf(O;\R)$ compactly supported in $O$. The net $\{\Acal_F (O)\}_{O\subset\M}$ of C$^\ast$-algebras satisfies the axioms (A)--(C') from Section \ref{sec:prelim-aqft} \cite{Dimock1980} (see also \cite{Fewster-Rejzner-AQFT} and \cite[Sec.\ I.2]{Borchers1996}). (The axioms correspond to Axiom 1, 3, 5 for the free scalar field on a globally hyperbolic spacetime in \cite{Dimock1980}, with the covariance Axiom 5 implying Poincar\'{e} covariance for the field on Minkowski spacetime.) The algebra $\Acal_F$ is the associated quasi-local algebra (see \cite[Prop.\ 5.2.10]{Bratteli-Robinson2}).\medskip

The time translations in an inertial reference frame $\frameu$ are implemented as a $^\ast$-automorphism group $\alpha^{(\frameu)}$ on $\Acal_F$, such that $\alpha_t^{(\frameu)} (\Acal_F (O)) = \Acal_F (O+t\frameu)$ for all $t\in\R$ and regions $O\subset\M$. Without loss of generality, as we only consider a single fixed reference frame in this appendix, we will employ the isomorphism $\M\cong\R^4$ and coordinates $(x_0 , x_1 , x_2 , x_3)\equiv(t , \xu)$ on $\R^4$ such that $\frameu=(1,\vec{0})$. In the chosen coordinates the Klein-Gordon operator is expressed in the familiar form $P=\partial_t^2 - \Delta + m^2$, where $\Delta$ is the Laplace operator on $\R^3$. The dynamics $\alpha^{(\frameu)} \equiv \alpha_F=\{\alpha_{F,t}\}_{t\in\R}$ is given by (cf.\ \cite[Thm.\ 5.2.8]{Bratteli-Robinson2})
\begin{gather}
	\label{eq:dynamics-KG}
	\alpha_{F,t} (W(\class{f})) = W(\class{f\circ T_{-t}})
\end{gather}
on generators of $\Acal_F$ for the time translation $T_t : \M\to\M$, $(t',\xu)\mapsto(t+t',\xu)$ (action of the Killing flow along $\partial_t$). The pair $(\Acal_F , \alpha_F)$ is the quantum dynamical system (Definition \ref{def:dynamical-system}) of the free scalar field with respect to inertial time translations along $\frameu$. Note that $t\mapsto\alpha_{F,t}$ is not strongly continuous since the C$^\ast$-norm distance of distinct Weyl algebra elements is always $2$ (see \cite[Prop.\ 4.2.3]{Baer-Ginoux-Pfaeffle}).

\subsection{KMS states}

The KMS states on $\Acal_F$ with respect to inertial time translations $\alpha_F$ are quasi-free, i.e.\ they are completely determined by their two-point function, with all odd $n$-point functions vanishing (see, e.g., \cite{Petz-CCR,Kay-Wald,Khavkine-Moretti}). Given a symmetric $\R$-bilinear form $\upmu:\Csymp \times \Csymp \to \R$ on $\Csymp$ satisfying the inequality $|\mathcal{E} (\class{f},\class{g})|^2 \leq 4\upmu(\class{f} , \class{f}) \upmu(\class{g} , \class{g})$ for all $f,g\in \Cinf(\M;\R)$, one can define a quasi-free state by $\omega (W(\class{f})) := \e^{-\frac{1}{2} \upmu(\class{f} , \class{f})}$ (extended linearly and continuously to all of $\Acal_F$), with two-point function given by $(f , g) \mapsto \upmu(\class{f} , \class{g}) + \frac{i}{2} \mathcal{E} (\class{f},\class{g})$ \cite{Kay-Wald}. Quasi-free states are equivalently characterized by their so-called one-particle structure on some Hilbert space, that is, a pair $(\kf,\Hcal)$ consisting of a symplectic map $\kf:(\Csymp,\mathcal{E}) \to (\Hcal,2\Im\langle \slot,\slot \rangle_{\Hcal})$ to a ``one-particle Hilbert space'' $\Hcal$ such that $\ran\kf+i\ran\kf\subset\Hcal$ is dense. Given a one-particle structure $(\kf,\Hcal)$, the bilinear form $\upmu$ defined by $\upmu(\class{f} , \class{g}) = \Re \langle \kf \class{f} , \kf \class{g} \rangle_{\Hcal}$ satisfies the above inequality and induces the quasi-free state $\omega$; conversely, given a bilinear form $\upmu$ satisfying this inequality, there exists a one-particle structure $(\kf,\Hcal)$, unique up to unitary equivalence, such that $\upmu$ arises this way \cite{kay-purification,kay-uniqueness,Kay-Wald} (see also \cite{Wald-book} for a general discussion). 

For the construction of the inertial KMS states on $\Acal_F$ the necessary ingredients in terms of one-particle structures can be found in \cite[Sec.\ 3.2]{Kay-Wald}. The explicit construction is presented in \cite[Secs.\ 7.1 \& A4]{Dimock-Kay1987}. Here we only summarize the most important points.\medskip

Let $\Sigma := \{0\}\times\R^3 \cong \R^3$ be the $t=0$ Cauchy surface in $\M$. For $f\in\Cinf(\M;\R)$, let
\begin{gather*}
	\utilde{f}_1 := E_P f\restr_{\Sigma} \, \in \Cinf(\Sigma;\R) \, , \quad \utilde{f}_2 := \partial_t E_P f\restr_{\Sigma} \, \in \Cinf(\Sigma;\R)
\end{gather*}
be the pair of Cauchy data on $\Sigma$ of the solution $E_P f$ to the Klein-Gordon equation corresponding to $f$ under the isomorphism of $\Csymp$ with the space of real-valued smooth solutions mentioned in the previous section. Due to the well-posedness and uniqueness of the Cauchy problem for the Klein-Gordon equation, solutions and Cauchy data can be identified (see \cite[Thm.\ 3.2.11]{Baer-Ginoux-Pfaeffle}), and thus $\Csymp\cong\Cinf(\Sigma;\R) \oplus \Cinf(\Sigma;\R)$. 

Let
\begin{gather*}
	\Gcal := L^2(\Sigma)\equiv L^2(\R^3,\Diff 3\xu) \, .
\end{gather*}
The operator $-\Delta+m^2$ is essentially self-adjoint on $\Cinf(\Sigma;\C)\subset\Gcal$ (see, e.g., \cite{Chernoff1973}), and we denote its self-adjoint closure by the same symbol. The $\R$-linear symplectic map $\gf : (\Csymp,\mathcal{E}) \to (\Gcal,2\Im\langle \slot,\slot \rangle_{\Gcal})$ defined by
\begin{gather}
	\label{eq:ground-one-particle}
	\gf\class{f} := \frac{1}{\sqrt{2}} \left( \ham_0^{1/2} \utilde{f}_1 + i\ham_0^{-1/2} \utilde{f}_2 \right) \, , \quad \ham_0 := (-\Delta+m^2)^{1/2}
\end{gather}
is part of the so-called ground one-particle structure $(\gf,\Gcal,\e^{i\ham_0 t})$ for the Minkowski vacuum state with one-particle Hamiltonian $\ham_0$. We will not go into the details here and instead refer to \cite[Secs.\ 1 \& 2]{kay-double-wedge} (in particular Eq.\ (2.8) ff.) and \cite[Secs.\ 7.1 \& A4]{Dimock-Kay1987}. On generators of $\Acal_F$, the Minkowski vacuum state $\omega_{F,\infty}$ is uniquely determined by $\omega_{F,\infty} (W(\class{f})) := \e^{-\frac{1}{2}\|\gf\class{f} \|^2_{\Gcal}}$ for $f\in\Cinf(\M;\R)$. The map $\gf$ has the properties (see \cite[Sec.\ A4]{Dimock-Kay1987})
\begin{gather}
	\label{eq:range-ground}
	\overline{\ran(\gf)}=\Gcal \, , \quad \ran(\gf) \subset \dom(\ham_0^{-1/2}) \, ,
\end{gather}
where the first is equivalent to the statement that $\omega_{F,\infty}$ is a pure state by \cite[Lemma A.2]{Kay-Wald}, and the second property, called ``regularity'' in \cite{kay-double-wedge,kay-uniqueness}, is essential to the construction of thermal states.\medskip 

The definition of the KMS states is based on a standard construction presented in \cite[\S 3.1]{kay-uniqueness} and summarized in \cite[p.\ 102 \& Lemma 6.2]{Kay-Wald}. For $\beta>0$, the $\R$-linear symplectic map $\kfbeta : (\Csymp,\mathcal{E}) \to (\Kcal,2\Im\langle \slot,\slot \rangle_{\Kcal})$ on the one-particle Hilbert space 
\begin{gather}
	\label{eq:KMS-Hilbert}
	\Kcal:=\Gcal\oplus\Gcal
\end{gather}
is defined by 
\begin{gather}
	\label{eq:kfbeta}
	\kfbeta\class{f} := \left(\sqrt{1+\mu_\beta (\ham_0)} \gf\class{f}\right) \oplus \left(\sqrt{\mu_\beta (\ham_0)} \Cc \gf\class{f}\right) \equiv \frac{1}{\sqrt{1-\e^{-\beta\ham_0}}} \gf\class{f} \oplus \frac{1}{\sqrt{\e^{\beta\ham_0} -1}} \Cc\gf\class{f}
\end{gather}
using spectral calculus, where $\mu_\beta : (0,\infty) \to (0,\infty)$, $x \mapsto (\e^{\beta x} -1)^{-1}$, and the anti-unitary map $\Cc:\Gcal\to\Gcal$ is the complex conjugation. The one-particle Hamiltonian on $\Kcal$ is given by the self-adjoint operator (on the usual domain)
\begin{gather}
	\label{eq:one-particle-Hamiltonian-KMS}
	\ham=\ham_0 \oplus -\ham_0 \, .
\end{gather}
It implements the time translations on $\Kcal$, that is,
\begin{gather}
	\label{eq:dyn-one-particle}
	\e^{i\ham t} \circ \kfbeta = \kfbeta \circ (T_t)_\ast
\end{gather}
on $\Csymp$, where $\e^{i\ham t} = \e^{i\ham_0 t} \oplus \e^{-i\ham_0 t}$ as an operator on $\Kcal$, and $(T_t)_\ast (\class{f}) := \class{f\circ T_{-t}}$ for $f\in\Cinf(\M;\R)$ and $T_t : \M\to\M$ the time translation as defined below Eq.\ \eqref{eq:dynamics-KG}. The triple $(\kfbeta,\Kcal,\e^{i\ham t})$ is the KMS one-particle structure (see \cite{kay-uniqueness} for the general definition) of the unique quasi-free $(\alpha_F,\beta)$-KMS state $\omega_{F,\beta}$ on $\Acal_F$, uniquely determined on generators by
\begin{gather}
	\label{eq:field-KMS-state}
	\omega_{F,\beta} (W(\class{f})) := \e^{-\frac{1}{2}\|\kfbeta\class{f} \|^2_{\Kcal}} \equiv \e^{-\frac{1}{2}\langle\gf\class{f} , \coth(\beta\ham_0 /2)\gf\class{f} \rangle_{\Gcal}}
\end{gather}
for $f\in\Cinf(\M;\R)$. The definition requires the regularity condition in \eqref{eq:range-ground}, see \cite[§A2]{kay-uniqueness}. The final expression follows from the definition of $\kfbeta$ and $1+2\mu_\beta(\ham_0)=\coth(\beta\ham_0 /2)$ (cf.\ \cite[Eq.\ (3.2)]{kay-uniqueness}, \cite[Eq.\ (A4.3)]{Dimock-Kay1987}) and reveals the similarity to the generating function for the canonical state of the thermal (infinitely extended) free Bose gas in the Araki-Woods representation \cite{Araki-Woods} (see also \cite{Cannon1973,JP1,JP-thermal-relaxation,Merkli-ideal-2006}, and the comment on top of p.\ 1021 in \cite{kay-uniqueness}), where $\mu_\beta (\nu) = (\e^{\beta \nu} -1)^{-1}$ is the Planckian momentum density distribution of the gas at temperature $\beta^{-1}$ with single boson energy $\nu$.\medskip 

The GNS representation of $\omega_{F,\beta}$ can be expressed via Weyl operators (exponentials of field operators) on a Fock space (see \cite[p.\ 76 \& Lemma A.2]{Kay-Wald}): It is given by $(\pi_\beta , \Fcal(\Kcal), \Omega_F)$, where $\Fcal(\Kcal):=\C\oplus\bigoplus_{n=1}^\infty \Kcal^{\odot n}$ is the symmetric (bosonic) Fock space over the one-particle space $\Kcal$ with symmetric tensor product $\odot$ and Fock vacuum vector $\Omega_F := 1\oplus 0\oplus 0\oplus\ldots$, and 
\begin{gather}
	\label{eq:WF}
	\pi_\beta (W(\class{f})) := \e^{i (a(\kfbeta\class{f})+a^\ast (\kfbeta\class{f}))} =: \WF(\kfbeta\class{f})
\end{gather}
for $f\in\Cinf(\M;\R)$, in which the self-adjoint closure is implied in the exponent, and $a(\slot)$ and $a^\ast (\slot)$ denote the annihilation and creation operators on $\Fcal(\Kcal)$, respectively. (For a comprehensive treatment of the theory of Fock spaces we refer to \cite{Arai}.) 

Using Eqs.\ \eqref{eq:dynamics-KG}, \eqref{eq:dyn-one-particle}, \eqref{eq:WF} one finds that, on Weyl algebra generators of $\Acal_F$,
\begin{multline*}
	\pi_\beta (\alpha_{F,t} (W(\class{f}))) = \pi_\beta (W(\class{f\circ T_{-t}})) = \WF(\kfbeta\class{f\circ T_{-t}}) = \WF(\e^{i\ham t} \kfbeta\class{f}) = \\ = \Gamma(\e^{i\ham t}) \WF(\kfbeta\class{f}) \Gamma(\e^{i\ham t})^{-1} = \e^{i\dGamma(\ham) t} \pi_\beta (W(\class{f})) \e^{-i\dGamma(\ham) t} \, ,
\end{multline*}
where $\Gamma(\e^{i\ham t}) = \bigoplus_{n=0}^\infty (\e^{i\ham t})^{\odot n} = \e^{i\dGamma(\ham) t}$ for the second quantization map $\dGamma$, and the step from the first to the second line follows from \cite[Thm.\ 5.30]{Arai}. This extends to $\Acal_F$ by linearity and continuity, and it holds $\dGamma(\ham)\Omega_F = 0$. Using Proposition \ref{prop:weak-Liouvillean}, the Liouvillean $L$ generating the W$^\ast$-dynamics on the von Neumann algebra $\pi_\beta (\Acal_F)''$ induced by $\omega_{F,\beta}$ is therefore given by the second quantization 
\begin{gather}
	\label{eq:Liouv-Fock}
	L=\dGamma(\ham)
\end{gather}
of the one-particle Hamiltonian $\ham$.

The ``exponential law'' for Fock spaces, $\Fcal(\Kcal)\cong\Fcal(\Gcal) \otimes \Fcal(\Gcal)$, allows to express the Fock space over $\Kcal = \Gcal \oplus \Gcal$ as a tensor product of Fock spaces over $\Gcal$ (see \cite[Thm.\ 5.38]{Arai}). The tensor product Fock space is frequently used for the GNS Hilbert space of the thermal massless Bose gas in the Araki-Woods representation, see, e.g., \cite{JP1,Merkli-Sigal-Berman,Merkli-Quantum-Markovian} and the overview in \cite[Sec.\ 2.2]{Merkli-LSO-2007}. We will not make use of this representation here.

\subsection{Comments on the mixing property}
\label{appendix:comments-mixing}

In Lemma \ref{lem:KMS-scalar-primary-mixing} in the main text we discussed the mixing property for $(\Acal_F , \alpha_F , \omega)$ for any inertial KMS state $\omega=\omega_{F,\beta}$ using the spectral characterization of Proposition \ref{prop:cond-mixing}. Let us illustrate the connection to spectral properties of the Liouvillean using the explicit definition of the mixing property (Definition \ref{def:mixing}) and the one-particle structure formalism (see also \cite[Prop.\ 4.1]{Buchholz-Verch2015} and \cite[Sec.\ 4.5]{Merkli-ideal-2006} for similar discussions). Recall from the proof of Lemma \ref{lem:KMS-scalar-primary-mixing} that $L=\dGamma(\ham)$ (Eq.\ \eqref{eq:Liouv-Fock}) has spectrum $\R$ that is absolutely continuous with a simple embedded eigenvalue at $0$ corresponding to $\Omega_F$. For the mixing property, it suffices to prove that $\lim_{t\to\infty} \omega(W(\class{g})\alpha_{F,t} W(\class{f})) = \omega(W(\class{f}))\omega(W(\class{g}))$ for all $f,g\in\Cinf(\M;\R)$; this then extends by linearity to the span of all Weyl generators, which is dense in $\Acal_F$, yielding the desired mixing property. By the Weyl relations for $\WF$ (Eq.\ \eqref{eq:WF}) and the definition of $\omega$ (Eq.\ \eqref{eq:field-KMS-state}),
\begin{multline*}
	\omega(W(\class{g})\alpha_{F,t} W(\class{f})) = \langle \Omega_F , \WF(\kfbeta\class{g}) \WF(\kfbeta\class{f\circ T_{-t}}) \Omega_F \rangle = \\ = \e^{-i\Im\langle \kfbeta\class{g} , \kfbeta\class{f\circ T_{-t}} \rangle_{\Kcal}} \langle \Omega_F , \WF(\kfbeta\class{g+f\circ T_{-t}}) \Omega_F \rangle = \\ = \e^{-\langle \kfbeta\class{g} , \kfbeta\class{f\circ T_{-t}} \rangle_{\Kcal}} \cdot \omega(W(\class{f}))\omega(W(\class{g})) \, ,
\end{multline*}
where in the last step it is used that $\langle \Omega_F , \WF(\kfbeta\class{g+f\circ T_{-t}}) \Omega_F \rangle = \e^{-\frac{1}{2}\|\kfbeta\class{g+f\circ T_{-t}} \|^2_{\Kcal}}$ and $\|\kfbeta\class{g+f\circ T_{-t}} \|^2_{\Kcal} = \|\kfbeta\class{g}\|^2_{\Kcal} + \|\kfbeta\class{f}\|^2_{\Kcal} + 2\Re\langle \kfbeta\class{g} , \kfbeta\class{f\circ T_{-t}} \rangle_{\Kcal}$, since $\|\kfbeta\class{f\circ T_{-t}}\|^2_{\Kcal} = \|\e^{i\ham t} \kfbeta\class{f}\|^2_{\Kcal} = \|\kfbeta\class{f}\|^2_{\Kcal}$ by Eq.\ \eqref{eq:dyn-one-particle}. 

Hence we have to show that $\lim_{t\to\infty} \Wcal_2 (g,f\circ T_{-t}) = 0$, where $\Wcal_2 (f,g) := \langle\kfbeta\class{f},\kfbeta\class{g} \rangle_{\Kcal} = \langle \Omega_F , Q(f)Q(g) \Omega_F \rangle$ defines the smeared two-point (Wightman) function of $\omega$ for $f,g\in\Cinf(\M;\R)$ and self-adjoint field operator $Q(g):=a(\kfbeta\class{g})+a^\ast (\kfbeta\class{g})$ (the closure of this expression is implied) on the Fock space $\Fcal(\Kcal)$. The desired decay of the two-point function at timelike separations is a manifestation of Huygens' principle \cite{Buchholz-Verch2015}. The absolute continuity of the spectrum of $L=\dGamma(\ham)$ on $\{\Omega_F\}^\perp$ implies that $\lim_{t\to\infty} \e^{iLt} = |\Omega_F \rangle\langle \Omega_F |$ in the weak operator sense (as a consequence of the Riemann-Lebesgue lemma), and therefore
\begin{gather*}
	\lim\limits_{t\to\infty} \Wcal_2 (g,f\circ T_{-t}) = \lim\limits_{t\to\infty} \langle \Omega_F , Q(g) \e^{iLt} Q(f) \Omega_F \rangle = \langle \Omega_F , Q(g) \Omega_F \rangle \cdot \langle \Omega_F , Q(f) \Omega_F \rangle = 0
\end{gather*}
since $\langle \Omega_F , Q(f) \Omega_F \rangle = \langle \Omega_F , a^\ast (\kfbeta\class{f}) \Omega_F \rangle = 0$.

\section{\JP\ glued representation}
\label{appendix:JP-glued-rep}

Consider the free scalar field of mass $m=0$, using the setup of Appendix \ref{appendix:scalar-field}. We present a unitary transformation (isometric isomorphism) of the field representation in which the Liouvillean $L=\dGamma(\ham)$ (Eq.\ \eqref{eq:Liouv-Fock}) is given by the second quantization of a simple multiplication operator. The main part of the construction is due to \cite{JP1} and is therefore known as ``\JP\ gluing''. We first unitarily transform the KMS one-particle structure $(\kfbeta,\Kcal,\e^{i\ham t})$ (Appendix \ref{appendix:scalar-field}), from which a unitarily equivalent GNS representation of $\omega_{F,\beta}$ is obtained.\medskip

Let $\hat{\Delta}:=-\Delta$, so we have $\ham_0 = \hat{\Delta}^{1/2}$ in Eq.\ \eqref{eq:ground-one-particle} and $\ham=\hat{\Delta}^{1/2} \oplus -\hat{\Delta}^{1/2}$ in Eq.\ \eqref{eq:one-particle-Hamiltonian-KMS}. The operator $\hat{\Delta}$ is positive, unbounded, densely defined and self-adjoint on $\Gcal = L^2(\R^3,\Diff 3\xu)$ with $\spec\hat{\Delta} = [0,\infty)$. Its spectral representation is provided by the Fourier transform
\begin{gather*}
	\Ffrak : L^2 (\R^3,\Diff3\xu) \to L^2 (\R^3,\Diff3\ku) \, , \quad
	\Ffrak[g](\ku) := \frac{1}{(2\pi)^{3/2}} \int\limits_{\R^3} \e^{-i\xu\ku} g(\xu) \Diff3\xu \, ,
\end{gather*}
which is a unitary operator by Plancherel's theorem. For any Borel subset $X\subseteq\R^3$ (equipped with the Lebesgue measure) and measurable function $f:X \to \R$ we define the self-adjoint multiplication operator $M_f$ by $M_f \, g := f\cdot g$ on the dense domain $\dom(M_f) := \{ g\in L^2 (X) \, : \, f\cdot g \in L^2 (X) \}$. To shorten the notation we will write the image of $f$ in the index. For example, for $f:\R^3 \to \R$, $\ku\mapsto\|\ku\|^2$ we write $M_{\|\ku\|^2}$ to denote the multiplication operator $M_f$ on $L^2 (\R^3,\Diff3\ku)$. It holds $\Ffrak[\hat{\Delta} g](\ku)=\|\ku\|^2 \Ffrak[g](\ku)$ for all $g\in\dom(\Delta)\subset \Gcal$ and $\dom(M_{\|\ku\|^2})=\Ffrak\dom(\Delta)$. By spectral calculus,
\begin{gather*}
	\Ffrak\ham_0 \Ffrak^{-1} = M_{\|\ku\|}
\end{gather*}
on the natural domain. Define
\begin{gather*}
	\Ftfrak :=\Ffrak\oplus\Cc\Ffrak\Cc : \Kcal = \Gcal \oplus \Gcal \to \Ffrak\Gcal \oplus \Ffrak\Gcal \, , \quad g \oplus g' \mapsto \Ffrak [g] \oplus \overline{\Ffrak [\overline{g'}]} \, ,
\end{gather*}
where the map $\Cc$ is the natural complex conjugation operator, which is an anti-unitary involution on the respective $L^2$-spaces whose action on $L^2$-functions is denoted by an overline. The maps $\Cc\Ffrak\Cc$ and $\Ftfrak$ are unitary, and we immediately obtain that
\begin{gather*}
	\Ftfrak \ham \Ftfrak^{-1} = \Ffrak \hat{\Delta}^{1/2} \Ffrak^{-1} \oplus -\Cc\Ffrak\Cc \hat{\Delta}^{1/2} \Cc \Ffrak^{-1} \Cc = M_{\|\ku\|} \oplus -M_{\|\ku\|} \, ,
\end{gather*}
where we used that $[\Cc,\hat{\Delta}]=0$ and $[\Cc,M_{\|\ku\|}]=0$. 

Let
\begin{gather*}
	\Hcal_F := L^2 (\R \times S^2 , \diff s \diff\varOmega) \, ,
\end{gather*}
where $\varOmega$ is the standard Riemannian metric on the unit 2-sphere $S^2 \subset \R^3$. There is an isomorphism $\Ffrak\Gcal = L^2 (\R^3 , \Diff3\ku) \to L^2 ((0,\infty) \times S^2 , q^2 \diff q \diff\varOmega)$ induced by introducing spherical coordinates on $\R^3$, denoted $f\mapsto f^\circ$. In particular, $(M_{\|\ku\|} f)^\circ (q,\varOmega) = qf^\circ (q,\varOmega)$ almost everywhere for $f\in\dom(M_{\|\ku\|})\subset L^2 (\R^3 , \Diff3\ku)$. For $\zeta\in\R$, define (almost everywhere)
\begin{gather*}
	\JPmap_{\zeta} : \Ffrak\Gcal \oplus \Ffrak\Gcal \to \Hcal_F \, , \quad 
	\JPmap_{\zeta} [f\oplus g](s,\varOmega) := s \left\{ \begin{array}{lr}
		f^\circ (s,\varOmega) & \textup{for }s> 0\\
		\e^{i\zeta} g^\circ (-s,\varOmega) & \textup{for }s<0
	\end{array} \right.  \, ,
\end{gather*}
called the \JP\ gluing map \cite{JP1,Derezinski-Jaksic-RTE} (see also \cite[Appendix A]{Merkli-Sigal-Berman}). The map $\JPmap_{\zeta}$ is unitary, and it holds
\begin{gather*}
	\JPmap_{\zeta} [M_{\|\ku\|} f \oplus -M_{\|\ku\|} g](s,\varOmega) =  s \left\{ \begin{array}{lr}
		s f^\circ (s,\varOmega) & \textup{for }s> 0\\
		-\e^{i\zeta} (-s) g^\circ (-s,\varOmega) & \textup{for }s<0
	\end{array} \right. = (M_s\JPmap_{\zeta} [f\oplus g])(s,\varOmega)
\end{gather*}
almost everywhere, where $M_s$ is the operator of multiplication with the function $(s,\varOmega)\mapsto s$ on $\Hcal_F$. We therefore obtain a unitary map
\begin{gather*}
	\Tfrak_{\zeta} := \JPmap_{\zeta} \circ \Ftfrak: \Kcal \to \Hcal_F
\end{gather*}
such that $\ham=\hat{\Delta}^{1/2} \oplus -\hat{\Delta}^{1/2}$ (Eq.\ \eqref{eq:one-particle-Hamiltonian-KMS}) is represented by the multiplication operator
\begin{gather*}
	\Tfrak_{\zeta} \ham\Tfrak_{\zeta}^{-1} = M_s
\end{gather*}
on $\Tfrak_{\zeta} \dom(\ham)=\dom(M_s)\subset\Hcal_F$.\medskip

The unitary map $\Tfrak_{\zeta}$ can be used to unitarily transform the one-particle structure $(\kfbeta,\Kcal,\e^{i\ham t})$ of $\omega_{F,\beta}$ (defined in Eqs.\ \eqref{eq:KMS-Hilbert}, \eqref{eq:kfbeta}, \eqref{eq:one-particle-Hamiltonian-KMS}). For all $\beta>0$ and $\zeta\in\R$, the triple
\begin{gather*}
	(\Kfbeta , \, \Hcal_F, \, \e^{i M_s t}) \, , \quad \Kfbeta := \Tfrak_{\zeta} \circ \kfbeta : (\Csymp,\mathcal{E}) \to \Hcal_F
\end{gather*}
is a $\beta$-KMS one-particle structure such that $\omega_{F,\beta} (W(\class{f})) = \e^{-\frac{1}{2}\|\Kfbeta\class{f} \|^2_{\Hcal_F}}$ for $f\in\Cinf(\M;\R)$ (cf.\ Eq.\ \eqref{eq:field-KMS-state}). The KMS one-particle structures describe the same KMS state $\omega_{F,\beta}$ since $\|\Kfbeta\class{f} \|_{\Hcal_F} = \|\kfbeta\class{f} \|_{\Kcal}$ for all $f\in\Cinf(\M;\R)$ by unitarity of $\Tfrak_{\zeta}$ (cf.\ the uniqueness result \cite[Thm.\ 1b]{kay-uniqueness}). The definitions, functional calculus, and the fact that $\Cc$ commutes with $\hat{\Delta}$ and $(\slot)^\circ$ imply that
\begin{multline}
	\label{eq:hbeta}
	f_\beta (s,\varOmega) := (\Kfbeta\class{f})(s,\varOmega) = s \left\{ \begin{array}{lr}
		\sqrt{1+\mu_\beta (s)} \Ffrak[\gf\class{f}]^\circ (s,\varOmega) & \textup{for }s> 0\\
		e^{i\zeta} \sqrt{\mu_\beta (-s)} \overline{\Ffrak[\gf\class{f}]^\circ}(-s,\varOmega) & \textup{for }s<0
	\end{array} \right. \, = \\ = \sqrt{\frac{s}{1-\e^{-\beta s}}} |s|^{1/2} \left\{ \begin{array}{lr}
		\Ffrak[\gf\class{f}]^\circ (s,\varOmega) & \textup{for }s> 0\\
		-\e^{i\zeta} \overline{\Ffrak[\gf\class{f}]^\circ}(-s,\varOmega) & \textup{for }s<0
	\end{array} \right.
\end{multline}
almost everywhere for all $f\in\Cinf(\M ; \R)$, where $\gf$ is defined in Eq.\ \eqref{eq:ground-one-particle}, and $\mu_\beta$ is given below Eq.\ \eqref{eq:kfbeta}.\medskip

From this unitarily transformed one-particle structure, the GNS representation of $\omega_{F,\beta}$ is now obtained as in Eq.\ \eqref{eq:WF} (using \cite[p.\ 76 \& Lemma A.2]{Kay-Wald}): 
\begin{lemma}
	\label{lem:glued-GNS}
	The GNS representation of $\omega_{F,\beta}$ is given by $(\pi_{F,\beta} , \Fcal(\Hcal_F), \Omega_F)$ for the symmetric Fock space $\Fcal(\Hcal_F)$ with Fock vacuum vector\footnote{Here we slightly abuse notation by using the same symbol as in Appendix \ref{appendix:scalar-field}.} $\Omega_F$, and $\pi_{F,\beta}$ is defined analogous to Eq.\ \eqref{eq:WF} (replacing $\kfbeta$ with $\Kfbeta$). By construction, this GNS representation is unitarily equivalent to $(\pi_\beta , \Fcal(\Kcal), \Omega_F)$. The Liouvillean $L_F$ generating the dynamics $\alpha_F$ (Eq.\ \eqref{eq:dynamics-KG}) in this representation is unitarily equivalent to Eq.\ \eqref{eq:Liouv-Fock}, and is given by the second-quantized multiplication operator
	\begin{gather*}
		L_F = \dGamma(M_s) \, ,
	\end{gather*}
	which is commonly denoted $\dGamma(s)$ and is self-adjoint on the natural domain $\dom(\dGamma(s))$ (see \cite[Thm.\ 5.2]{Arai}).
\end{lemma}

The same result is also obtained if we first transform the one-particle structure $(\gf,\Gcal,\e^{i\ham_0 t})$ for the Minkowski vacuum state to $(\Ffrak\circ\gf,\Ffrak\Gcal,\e^{iM_{\|\ku\|} t})$ using $\Ffrak$, then perform the construction of the KMS one-particle structure for this unitarily transformed ground one-particle structure following \textit{mutatis mutandis} Eq.\ \eqref{eq:KMS-Hilbert} ff., and finally apply the gluing map $\JPmap_{\zeta}$. Comparing with, e.g., \cite{Froehlich-Merkli-Another-RTE} or the summaries in \cite[Secs.\ 3.2 \& 3.3]{Merkli-Quantum-Markovian} or \cite[Secs.\ 2.1 \& 2.2]{Merkli2021-2}, we see that the GNS representation of the quasi-free KMS state $\omega_{F,\beta}$ of the free massless scalar field formally coincides with the \JP\ glued Araki-Woods representation \cite{Araki-Woods,JP1} of a thermal massless Bose gas, which is a common choice for the reservoir in the study of RTE (see Section \ref{sec:th-rte}). As a complex-valued function, $f_\beta$ appears in the interaction term of the generator of dynamics for the coupled system (see Eq.\ \eqref{eq:monopole-coupling}). The phase parameter $\zeta\in\R$ is chosen such that the two parts of $f_\beta$ fit together at $s=0$ in a continuous, sufficiently regular way \cite{Froehlich-Merkli-Another-RTE,Merkli-Quantum-Markovian}.\medskip 

Since the Fock vacuum vector $\Omega_F$ represents a KMS state, there exists a Tomita-Takesaki modular conjugation $J_F$ on $\Fcal(\Hcal_F)$ associated to the pair $(\pi_{F,\beta} (\Acal_F)'' , \Omega_F)$, and the modular operator is given by $\e^{-\beta L_F}$ (Propositions \ref{prop:modular} \& \ref{prop:modular-KMS}). The operator $J_F$ is uniquely determined by the relations $J_F \e^{iL_F t} J_F = \e^{iL_F t}$ and
\begin{gather*}
	J_F \e^{-\beta L_F /2} \pi_{F,\beta} (W(\class{f})) \Omega_F = \pi_{F,\beta} (W(\class{f}))^\ast \Omega_F
\end{gather*}
for all $f\in\Cinf(\M;\R)$, which extends to $\pi_{F,\beta} (\Acal_F)''$ and implies, together with the fact that $\pi_{F,\beta} (\Acal_F)'' \Omega_F$ is a core for $\e^{-\beta L_F /2}$, that such $J_F$ indeed coincides with the Tomita-Takesaki modular conjugation (see \cite[Eqs.\ (3.47)--(3.48)]{Kay-Wald} and the proof of \cite[Thm.\ 1]{kay-purification}). From the one-particle formalism \cite{kay-double-wedge,kay-purification,kay-uniqueness,Kay-Wald} one obtains
\begin{gather}
	\label{eq:modular-conj-field-new}
	J_F = \Gamma(j_F) \, ,
\end{gather}
where $j_F$ is the anti-unitary involution on $\Hcal_F$ given almost everywhere by
\begin{gather*}
	j_F (f)(s,\varOmega) = -\e^{i\zeta} \overline{f}(-s,\varOmega)
\end{gather*}
(recall that the GNS representation $\pi_{F,\beta}$ depends on the parameter $\zeta\in\R$ appearing in the definition of $\Tfrak_{\zeta}$). A direct calculation using Eq.\ \eqref{eq:hbeta} and the identity $\mu_\beta (s) = \e^{-\beta s} (1+\mu_\beta (s))$, $s\in\R\setminus\{0\}$, shows that, almost everywhere,
\begin{gather}
	(j_F f_\beta)(s,\varOmega) = \e^{-\beta s/2} f_\beta (s,\varOmega) \, .
	\label{eq:j-hbeta}
\end{gather}


\bigskip

\mysepline

\newpage

\renewcommand{\refname}{\Large References}
\begin{footnotesize}
\providecommand{\etalchar}[1]{$^{#1}$}
\providecommand{\doi}[1]{\url{https://doi.org/#1}}

\end{footnotesize}


\begin{thebibliography}{TMCA22}
	
	\bibitem[Ara18]{Arai}
	\bgroup\scshape{}A.~Arai\egroup{}, \emph{{Analysis on Fock Spaces and
			Mathematical Theory of Quantum Fields}}, {World} {Scientific}, 2018.
	\doi{10.1142/10367}.
	
	\bibitem[AW63]{Araki-Woods}
	\bgroup\scshape{}H.~Araki\egroup{} and \bgroup\scshape{}E.~J. Woods\egroup{},
	{Representations of the Canonical Commutation Relations Describing a
		Nonrelativistic Infinite Free Bose Gas},  \emph{Journal of Mathematical
		Physics} \textbf{4} (1963), 637--662. \doi{10.1063/1.1704002}.
	
	\bibitem[Ara73]{Araki1973}
	\bgroup\scshape{}H.~Araki\egroup{}, Relative {Hamiltonian} for faithful normal
	states of a von {Neumann} algebra,  \emph{Publications of the RIMS Kyoto
		University} \textbf{9} (1973), 165--209. \doi{10.2977/prims/1195192744}.
	
	\bibitem[AA68]{Arnold-Avez}
	\bgroup\scshape{}V.~I. Arnold\egroup{} and \bgroup\scshape{}A.~Avez\egroup{},
	\emph{Ergodic problems of classical mechanics}, \emph{Mathematical physics
		monograph series} \textbf{9}, W. A. Benjamin, 1968.
	
	\bibitem[ABGG21]{inversion-of-statistics-2021}
	\bgroup\scshape{}J.~Arrechea\egroup{},
	\bgroup\scshape{}C.~Barcel{\'{o}}\egroup{}, \bgroup\scshape{}L.~J.
	Garay\egroup{}, and \bgroup\scshape{}G.~{Garc\'{i}a-Moreno}\egroup{},
	Inversion of statistics and thermalization in the {Unruh} effect,
	\emph{Physical Review D} \textbf{104} (2021), 065004.
	\doi{10.1103/physrevd.104.065004}.
	
	\bibitem[AJP06]{Attal2006}
	\bgroup\scshape{}S.~Attal\egroup{}, \bgroup\scshape{}A.~Joye\egroup{}, and
	\bgroup\scshape{}C.-A. Pillet\egroup{} (eds.), \emph{{Open Quantum Systems
			I}}, \emph{Lecture Notes in Mathematics} \textbf{1880}, Springer Berlin
	Heidelberg, 2006. \doi{10.1007/b128449}.
	
	\bibitem[BFS00]{BFS2000}
	\bgroup\scshape{}V.~Bach\egroup{}, \bgroup\scshape{}J.~Fr\"{o}hlich\egroup{},
	and \bgroup\scshape{}I.~M. Sigal\egroup{}, Return to equilibrium,
	\emph{Journal of Mathematical Physics} \textbf{41} (2000), 3985--4060.
	\doi{10.1063/1.533334}.
	
	\bibitem[Bac99]{Bachelot}
	\bgroup\scshape{}A.~Bachelot\egroup{}, The {Hawking} effect,  \emph{Annales de
		l'Institut Henri Poincar\'{e} Physique th\'{e}orique} \textbf{70} (1999),
	41--99. Available at \url{http://www.numdam.org/item/AIHPA_1999__70_1_41_0}.
	
	\bibitem[BGP07]{Baer-Ginoux-Pfaeffle}
	\bgroup\scshape{}C.~B\"{a}r\egroup{}, \bgroup\scshape{}N.~Ginoux\egroup{}, and
	\bgroup\scshape{}F.~Pf\"{a}ffle\egroup{}, \emph{{Wave Equations on Lorentzian
			Manifolds and Quantization}}, European Mathematical Society Publishing House,
	2007. \doi{10.4171/037}.
	
	\bibitem[BFK06]{Berkovitz2006}
	\bgroup\scshape{}J.~Berkovitz\egroup{}, \bgroup\scshape{}R.~Frigg\egroup{}, and
	\bgroup\scshape{}F.~Kronz\egroup{}, The ergodic hierarchy, randomness and
	{Hamiltonian} chaos,  \emph{Studies in History and Philosophy of Science Part
		B: Studies in History and Philosophy of Modern Physics} \textbf{37} (2006),
	661--691. \doi{10.1016/j.shpsb.2006.02.003}.
	
	\bibitem[BEG{\etalchar{+}}20]{Biermann-et-al2020}
	\bgroup\scshape{}S.~Biermann\egroup{}, \bgroup\scshape{}S.~Erne\egroup{},
	\bgroup\scshape{}C.~Gooding\egroup{}, \bgroup\scshape{}J.~Louko\egroup{},
	\bgroup\scshape{}J.~Schmiedmayer\egroup{}, \bgroup\scshape{}W.~G.
	Unruh\egroup{}, and \bgroup\scshape{}S.~Weinfurtner\egroup{}, {Unruh} and
	analogue {Unruh} temperatures for circular motion in $3+1$ and $2+1$
	dimensions,  \emph{Physical Review D} \textbf{102} (2020), 085006.
	\doi{10.1103/physrevd.102.085006}.
	
	\bibitem[BD82]{Birrell-Davies}
	\bgroup\scshape{}N.~D. Birrell\egroup{} and \bgroup\scshape{}P.~C.~W.
	Davies\egroup{}, \emph{{Quantum Fields in Curved Space}}, Cambridge
	University Press, 1982. \doi{10.1017/cbo9780511622632}.
	
	\bibitem[BW75]{Bisognano-Wichmann}
	\bgroup\scshape{}J.~J. Bisognano\egroup{} and \bgroup\scshape{}E.~H.
	Wichmann\egroup{}, On the duality condition for a {Hermitian} scalar field,
	\emph{Journal of Mathematical Physics} \textbf{16} (1975), 985--1007.
	\doi{10.1063/1.522605}.
	
	\bibitem[BW76]{Bisognano-Wichmann2}
	\bgroup\scshape{}J.~J. Bisognano\egroup{} and \bgroup\scshape{}E.~H.
	Wichmann\egroup{}, On the duality condition for quantum fields,
	\emph{Journal of Mathematical Physics} \textbf{17} (1976), 303--321.
	\doi{10.1063/1.522898}.
	
	\bibitem[Bor00]{Borchers-Tomita}
	\bgroup\scshape{}H.-J. Borchers\egroup{}, On revolutionizing quantum field
	theory with {Tomita's} modular theory,  \emph{Journal of Mathematical
		Physics} \textbf{41} (2000), 3604--3673. \doi{10.1063/1.533323}.
	
	\bibitem[Bor96]{Borchers1996}
	\bgroup\scshape{}H.-J. Borchers\egroup{}, \emph{{Translation Group and Particle
			Representations in Quantum Field Theory}}, \emph{Lecture Notes in Physics
		Monographs} \textbf{40}, Springer Berlin Heidelberg, 1996.
	\doi{10.1007/978-3-540-49954-1}.
	
	\bibitem[BKR78]{Bratteli-Kishimoto-Robinson}
	\bgroup\scshape{}O.~Bratteli\egroup{}, \bgroup\scshape{}A.~Kishimoto\egroup{},
	and \bgroup\scshape{}D.~W. Robinson\egroup{}, Stability properties and the
	{KMS} condition,  \emph{Communications in Mathematical Physics} \textbf{61}
	(1978), 209--238. \doi{10.1007/bf01940765}.
	
	\bibitem[BR87]{Bratteli-Robinson1}
	\bgroup\scshape{}O.~Bratteli\egroup{} and \bgroup\scshape{}D.~W.
	Robinson\egroup{}, \emph{{Operator Algebras and Quantum Statistical Mechanics
			1}}, 2nd ed., Springer Berlin Heidelberg, 1987.
	\doi{10.1007/978-3-662-02520-8}.
	
	\bibitem[BR97]{Bratteli-Robinson2}
	\bgroup\scshape{}O.~Bratteli\egroup{} and \bgroup\scshape{}D.~W.
	Robinson\egroup{}, \emph{{Operator Algebras and Quantum Statistical Mechanics
			2}}, 2nd ed., Springer Berlin Heidelberg, 1997.
	\doi{10.1007/978-3-662-03444-6}.
	
	\bibitem[BP07]{Breuer-Petruccione}
	\bgroup\scshape{}H.-P. Breuer\egroup{} and
	\bgroup\scshape{}F.~Petruccione\egroup{}, \emph{{The Theory of Open Quantum
			Systems}}, Oxford University Press, 2007.
	\doi{10.1093/acprof:oso/9780199213900.001.0001}.
	
	\bibitem[BB94]{Bros-Buchholz}
	\bgroup\scshape{}J.~Bros\egroup{} and \bgroup\scshape{}D.~Buchholz\egroup{},
	Towards a relativistic {KMS}-condition,  \emph{Nuclear Physics B}
	\textbf{429} (1994), 291--318. \doi{10.1016/0550-3213(94)00298-3}.
	
	\bibitem[BB96]{Bros-Buchholz1996}
	\bgroup\scshape{}J.~Bros\egroup{} and \bgroup\scshape{}D.~Buchholz\egroup{},
	Axiomatic analyticity properties and representations of particles in thermal
	quantum field theory,  \emph{Annales de l'Institut Henri Poincar\'{e}
		Physique th\'{e}orique} \textbf{64} (1996), 495--521. Available at
	\url{http://www.numdam.org/item/AIHPA_1996__64_4_495_0}.
	
	\bibitem[BFV03]{BFV2003}
	\bgroup\scshape{}R.~Brunetti\egroup{},
	\bgroup\scshape{}K.~Fredenhagen\egroup{}, and
	\bgroup\scshape{}R.~Verch\egroup{}, {The Generally Covariant Locality
		Principle -- A New Paradigm for Local Quantum Field Theory},
	\emph{Communications in Mathematical Physics} \textbf{237} (2003), 31--68.
	\doi{10.1007/s00220-003-0815-7}.
	
	\bibitem[Buc00]{Buchholz-supersymmetry}
	\bgroup\scshape{}D.~Buchholz\egroup{}, {On the Implementation of
		Supersymmetry},  in \emph{Theoretical Physics Fin de Si{\`{e}}cle}
	(\bgroup\scshape{}A.~Borowiec\egroup{},
	\bgroup\scshape{}W.~Ceg{\l}a\egroup{},
	\bgroup\scshape{}B.~Jancewicz\egroup{}, and
	\bgroup\scshape{}W.~Karwowski\egroup{}, eds.), \emph{Lecture Notes in
		Physics} \textbf{539}, Springer Berlin Heidelberg, 2000, pp.~211--220.
	\doi{10.1007/3-540-46700-9_14}.
	
	\bibitem[BO97]{Buchholz-Ojima-supersymmetry}
	\bgroup\scshape{}D.~Buchholz\egroup{} and \bgroup\scshape{}I.~Ojima\egroup{},
	Spontaneous collapse of supersymmetry,  \emph{Nuclear Physics B} \textbf{498}
	(1997), 228--242. \doi{10.1016/s0550-3213(97)00274-5}.
	
	\bibitem[BOR02]{Buchholz-Ojima-Roos}
	\bgroup\scshape{}D.~Buchholz\egroup{}, \bgroup\scshape{}I.~Ojima\egroup{}, and
	\bgroup\scshape{}H.~Roos\egroup{}, {Thermodynamic Properties of
		Non-equilibrium States in Quantum Field Theory},  \emph{Annals of Physics}
	\textbf{297} (2002), 219--242. \doi{10.1006/aphy.2002.6222}.
	
	\bibitem[BS07]{Buchholz-Schlemmer2007}
	\bgroup\scshape{}D.~Buchholz\egroup{} and
	\bgroup\scshape{}J.~Schlemmer\egroup{}, Local temperature in curved
	spacetime,  \emph{Classical and Quantum Gravity} \textbf{24} (2007),
	F25--F31. \doi{10.1088/0264-9381/24/7/f01}.
	
	\bibitem[BS13]{Buchholz-Solveen}
	\bgroup\scshape{}D.~Buchholz\egroup{} and \bgroup\scshape{}C.~Solveen\egroup{},
	Unruh effect and the concept of temperature,  \emph{Classical and Quantum
		Gravity} \textbf{30} (2013), 085011. \doi{10.1088/0264-9381/30/8/085011}.
	
	\bibitem[BV15]{Buchholz-Verch2015}
	\bgroup\scshape{}D.~Buchholz\egroup{} and \bgroup\scshape{}R.~Verch\egroup{},
	Macroscopic aspects of the {Unruh} effect,  \emph{Classical and Quantum
		Gravity} \textbf{32} (2015), 245004. \doi{10.1088/0264-9381/32/24/245004}.
	
	\bibitem[BV16]{Buchholz-Verch2016}
	\bgroup\scshape{}D.~Buchholz\egroup{} and \bgroup\scshape{}R.~Verch\egroup{},
	{Unruh} versus {Tolman}: on the heat of acceleration,  \emph{General
		Relativity and Gravitation} \textbf{48} (2016), 32.
	\doi{10.1007/s10714-016-2029-2}.
	
	\bibitem[BPPL24]{Bunney-Parry-Perche-Louko2024}
	\bgroup\scshape{}C.~R.~D. Bunney\egroup{}, \bgroup\scshape{}L.~Parry\egroup{},
	\bgroup\scshape{}T.~R. Perche\egroup{}, and
	\bgroup\scshape{}J.~Louko\egroup{}, Ambient temperature versus ambient
	acceleration in the circular motion {Unruh} effect,  \emph{Physical Review D}
	\textbf{109} (2024), 065001. \doi{10.1103/physrevd.109.065001}.
	
	\bibitem[CH71]{Callen-Horwitz}
	\bgroup\scshape{}H.~Callen\egroup{} and \bgroup\scshape{}G.~Horwitz\egroup{},
	{Relativistic Thermodynamics},  \emph{American Journal of Physics}
	\textbf{39} (1971), 938--947. \doi{10.1119/1.1986330}.
	
	\bibitem[Can73]{Cannon1973}
	\bgroup\scshape{}J.~T. Cannon\egroup{}, Infinite volume limits of the canonical
	free {Bose} gas states on the {Weyl} algebra,  \emph{Communications in
		Mathematical Physics} \textbf{29} (1973), 89--104. \doi{10.1007/bf01645656}.
	
	\bibitem[CS69]{Cavalleri-Salgarelli}
	\bgroup\scshape{}G.~Cavalleri\egroup{} and
	\bgroup\scshape{}G.~Salgarelli\egroup{}, Revision of the relativistic
	dynamics with variable rest mass and application to relativistic
	thermodynamics,  \emph{Il Nuovo Cimento A} \textbf{62} (1969), 722--754.
	\doi{10.1007/bf02819595}.
	
	\bibitem[Che73]{Chernoff1973}
	\bgroup\scshape{}P.~R. Chernoff\egroup{}, Essential self-adjointness of powers
	of generators of hyperbolic equations,  \emph{Journal of Functional Analysis}
	\textbf{12} (1973), 401--414. \doi{10.1016/0022-1236(73)90003-7}.
	
	\bibitem[Con73]{Connes1973}
	\bgroup\scshape{}A.~Connes\egroup{}, Une classification des facteurs de type
	{III},  \emph{Annales scientifiques de l'{\'{E}}cole normale
		sup{\'{e}}rieure} \textbf{6} (1973), 133--252. \doi{10.24033/asens.1247}.
	
	\bibitem[CM95a]{Costa-Matsas-background1995}
	\bgroup\scshape{}S.~S. Costa\egroup{} and \bgroup\scshape{}G.~E.~A.
	Matsas\egroup{}, Background thermal contributions in testing the {Unruh}
	effect,  \emph{Physical Review D} \textbf{52} (1995), 3466--3471.
	\doi{10.1103/physrevd.52.3466}.
	
	\bibitem[CM95b]{Costa-Matsas1995}
	\bgroup\scshape{}S.~S. Costa\egroup{} and \bgroup\scshape{}G.~E.~A.
	Matsas\egroup{}, Temperature and relativity,  \emph{Physics Letters A}
	\textbf{209} (1995), 155--159. \doi{10.1016/0375-9601(95)00843-7}.
	
	\bibitem[CHM08]{Crispino-Higuchi-Matsas}
	\bgroup\scshape{}L.~C.~B. Crispino\egroup{},
	\bgroup\scshape{}A.~Higuchi\egroup{}, and \bgroup\scshape{}G.~E.~A.
	Matsas\egroup{}, The {Unruh} effect and its applications,  \emph{Reviews of
		Modern Physics} \textbf{80} (2008), 787--838.
	\doi{10.1103/revmodphys.80.787}.
	
	\bibitem[DH06]{DAntoni-Hollands}
	\bgroup\scshape{}C.~D'Antoni\egroup{} and
	\bgroup\scshape{}S.~Hollands\egroup{}, Nuclearity, local quasiequivalence and
	split property for {Dirac} quantum fields in curved spacetime,
	\emph{Communications in Mathematical Physics} \textbf{261} (2006), 133--159.
	\doi{10.1007/s00220-005-1398-2}.
	
	\bibitem[Dap11]{Dappiaggi2011}
	\bgroup\scshape{}C.~Dappiaggi\egroup{}, Remarks on the {Reeh-Schlieder}
	property for higher spin free fields on curved spacetimes,  \emph{Reviews in
		Mathematical Physics} \textbf{23} (2011), 1035--1062.
	\doi{10.1142/s0129055x11004515}.
	
	\bibitem[Dav76]{Davies-open}
	\bgroup\scshape{}E.~B. Davies\egroup{}, \emph{{Quantum Theory of Open
			Systems}}, Academic Press London, 1976.
	
	\bibitem[Dav75]{Davies1975}
	\bgroup\scshape{}P.~C.~W. Davies\egroup{}, Scalar production in {Schwarzschild}
	and {Rindler} metrics,  \emph{Journal of Physics A: Mathematical and General}
	\textbf{8} (1975), 609--616. \doi{10.1088/0305-4470/8/4/022}.
	
	\bibitem[DM06]{DeB-M}
	\bgroup\scshape{}S.~{De Bi{\`{e}}vre}\egroup{} and
	\bgroup\scshape{}M.~Merkli\egroup{}, The {Unruh} effect revisited,
	\emph{Classical and Quantum Gravity} \textbf{23} (2006), 6525--6541.
	\doi{10.1088/0264-9381/23/22/026}.
	
	\bibitem[DJP03]{DJP2003}
	\bgroup\scshape{}J.~Derezi{\'{n}}ski\egroup{},
	\bgroup\scshape{}V.~Jak{\v{s}}i{\'{c}}\egroup{}, and \bgroup\scshape{}C.-A.
	Pillet\egroup{}, {Perturbation Theory of {W$^{\ast}$}-Dynamics,
		{Liouvilleans} and {KMS}-States},  \emph{Reviews in Mathematical Physics}
	\textbf{15} (2003), 447--489. \doi{10.1142/s0129055x03001679}.
	
	\bibitem[Der07]{Derezinski-RTE-report2007}
	\bgroup\scshape{}J.~Derezi{\'{n}}ski\egroup{}, Return to equilibrium for small
	quantum systems interacting with environment,  \emph{Reports on Mathematical
		Physics} \textbf{59} (2007), 317--330. \doi{10.1016/s0034-4877(07)80068-1}.
	
	\bibitem[DF06]{Derezinski-Frueboes}
	\bgroup\scshape{}J.~Derezi{\'{n}}ski\egroup{} and
	\bgroup\scshape{}R.~Fr\"{u}boes\egroup{}, {Fermi Golden Rule and Open Quantum
		Systems},  in \emph{Open Quantum Systems III}
	(\bgroup\scshape{}S.~Attal\egroup{}, \bgroup\scshape{}A.~Joye\egroup{}, and
	\bgroup\scshape{}C.-A. Pillet\egroup{}, eds.), \emph{Lecture Notes in
		Mathematics} \textbf{1882}, Springer Berlin Heidelberg, 2006, pp.~67--116.
	\doi{10.1007/3-540-33967-1_2}.
	
	\bibitem[DG99]{Derezinski-Gerard}
	\bgroup\scshape{}J.~Derezi{\'{n}}ski\egroup{} and
	\bgroup\scshape{}C.~G\'{e}rard\egroup{}, {Asymptotic Completeness in Quantum
		Field Theory: Massive Pauli-Fierz Hamiltonians},  \emph{Reviews in
		Mathematical Physics} \textbf{11} (1999), 383--450.
	\doi{10.1142/s0129055x99000155}.
	
	\bibitem[DJ01]{Derezinski-Jaksic2001}
	\bgroup\scshape{}J.~Derezi{\'{n}}ski\egroup{} and
	\bgroup\scshape{}V.~Jak{\v{s}}i{\'{c}}\egroup{}, {Spectral Theory of
		Pauli-Fierz Operators},  \emph{Journal of Functional Analysis} \textbf{180}
	(2001), 243--327. \doi{10.1006/jfan.2000.3681}.
	
	\bibitem[DJ03]{Derezinski-Jaksic-RTE}
	\bgroup\scshape{}J.~Derezi{\'{n}}ski\egroup{} and
	\bgroup\scshape{}V.~Jak{\v{s}}i{\'{c}}\egroup{}, {Return to Equilibrium for
		Pauli-Fierz Systems},  \emph{Annales Henri Poincar{\'{e}}} \textbf{4} (2003),
	739--793. \doi{10.1007/s00023-003-0146-4}.
	
	\bibitem[DeW79]{DeWitt}
	\bgroup\scshape{}B.~S. DeWitt\egroup{}, Quantum gravity: the new synthesis,  in
	\emph{General Relativity: An Einstein Centenary Survey}
	(\bgroup\scshape{}S.~W. Hawking\egroup{} and
	\bgroup\scshape{}W.~Israel\egroup{}, eds.), Cambridge University Press, 1979,
	pp.~680--745.
	
	\bibitem[Dim80]{Dimock1980}
	\bgroup\scshape{}J.~Dimock\egroup{}, Algebras of local observables on a
	manifold,  \emph{Communications in Mathematical Physics} \textbf{77} (1980),
	219--228. \doi{10.1007/bf01269921}.
	
	\bibitem[Dim82]{Dimock1982}
	\bgroup\scshape{}J.~Dimock\egroup{}, Dirac quantum fields on a manifold,
	\emph{Transactions of the American Mathematical Society} \textbf{269} (1982),
	133--147. \doi{10.1090/s0002-9947-1982-0637032-8}.
	
	\bibitem[DK87]{Dimock-Kay1987}
	\bgroup\scshape{}J.~Dimock\egroup{} and \bgroup\scshape{}B.~S. Kay\egroup{},
	Classical and quantum scattering theory for linear scalar fields on the
	{Schwarzschild} metric {I},  \emph{Annals of Physics} \textbf{175} (1987),
	366--426. \doi{10.1016/0003-4916(87)90214-4}.
	
	\bibitem[DKKR67]{Doplicher-Kadison-Kastler-Robinson1967}
	\bgroup\scshape{}S.~Doplicher\egroup{}, \bgroup\scshape{}R.~V.
	Kadison\egroup{}, \bgroup\scshape{}D.~Kastler\egroup{}, and
	\bgroup\scshape{}D.~W. Robinson\egroup{}, Asymptotically abelian systems,
	\emph{Communications in Mathematical Physics} \textbf{6} (1967), 101--120.
	\doi{10.1007/bf01654127}.
	
	\bibitem[DKS69]{Doplicher-Kastler-Stormer1969}
	\bgroup\scshape{}S.~Doplicher\egroup{}, \bgroup\scshape{}D.~Kastler\egroup{},
	and \bgroup\scshape{}E.~St{\o}rmer\egroup{}, Invariant states and asymptotic
	abelianness,  \emph{Journal of Functional Analysis} \textbf{3} (1969),
	419--434. \doi{10.1016/0022-1236(69)90033-0}.
	
	\bibitem[DLSB15]{DLSB}
	\bgroup\scshape{}B.~Doyon\egroup{}, \bgroup\scshape{}A.~Lucas\egroup{},
	\bgroup\scshape{}K.~Schalm\egroup{}, and \bgroup\scshape{}M.~J.
	Bhaseen\egroup{}, Non-equilibrium steady states in the {Klein-Gordon} theory,
	\emph{Journal of Physics A: Mathematical and Theoretical} \textbf{48}
	(2015), 095002. \doi{10.1088/1751-8113/48/9/095002}.
	
	\bibitem[DFP18]{Drago-Faldino-Pinamonti2018}
	\bgroup\scshape{}N.~Drago\egroup{}, \bgroup\scshape{}F.~Faldino\egroup{}, and
	\bgroup\scshape{}N.~Pinamonti\egroup{}, {On the Stability of KMS States in
		Perturbative Algebraic Quantum Field Theories},  \emph{Communications in
		Mathematical Physics} \textbf{357} (2018), 267--293.
	\doi{10.1007/s00220-017-2975-x}.
	
	\bibitem[Ear11]{Earman2011}
	\bgroup\scshape{}J.~Earman\egroup{}, The {Unruh} effect for philosophers,
	\emph{Studies in History and Philosophy of Science Part B: Studies in History
		and Philosophy of Modern Physics} \textbf{42} (2011), 81--97.
	\doi{10.1016/j.shpsb.2011.04.001}.
	
	\bibitem[EFG15]{Eisert-Friesdorf-Gogolin2015}
	\bgroup\scshape{}J.~Eisert\egroup{}, \bgroup\scshape{}M.~Friesdorf\egroup{},
	and \bgroup\scshape{}C.~Gogolin\egroup{}, Quantum many-body systems out of
	equilibrium,  \emph{Nature Physics} \textbf{11} (2015), 124--130.
	\doi{10.1038/nphys3215}.
	
	\bibitem[Emc72]{Emch}
	\bgroup\scshape{}G.~G. Emch\egroup{}, \emph{{Algebraic Methods in Statistical
			Mechanics and Quantum Field Theory}}, Wiley \& Sons, 1972.
	
	\bibitem[EKV70]{Emch-Knops-Verboven1970}
	\bgroup\scshape{}G.~G. Emch\egroup{}, \bgroup\scshape{}H.~J.~F. Knops\egroup{},
	and \bgroup\scshape{}E.~J. Verboven\egroup{}, {Breaking of Euclidean Symmetry
		with an Application to the Theory of Crystallization},  \emph{Journal of
		Mathematical Physics} \textbf{11} (1970), 1655--1668.
	\doi{10.1063/1.1665307}.
	
	\bibitem[FNV88]{FNV}
	\bgroup\scshape{}M.~Fannes\egroup{}, \bgroup\scshape{}B.~Nachtergaele\egroup{},
	and \bgroup\scshape{}A.~Verbeure\egroup{}, The equilibrium states of the
	spin-boson model,  \emph{Communications in Mathematical Physics} \textbf{114}
	(1988), 537--548. \doi{10.1007/bf01229453}.
	
	\bibitem[FPM17]{Farias-Pinto-Moya}
	\bgroup\scshape{}C.~Far{\'i}as\egroup{}, \bgroup\scshape{}V.~A. Pinto\egroup{},
	and \bgroup\scshape{}P.~S. Moya\egroup{}, What is the temperature of a moving
	body?,  \emph{Scientific Reports} \textbf{7} (2017), 17657.
	\doi{10.1038/s41598-017-17526-4}.
	
	\bibitem[FJL16]{Waiting-for-Unruh}
	\bgroup\scshape{}C.~J. Fewster\egroup{}, \bgroup\scshape{}B.~A.
	{Ju{\'{a}}rez-Aubry}\egroup{}, and \bgroup\scshape{}J.~Louko\egroup{},
	{Waiting for Unruh},  \emph{Classical and Quantum Gravity} \textbf{33}
	(2016), 165003. \doi{10.1088/0264-9381/33/16/165003}.
	
	\bibitem[FR20]{Fewster-Rejzner-AQFT}
	\bgroup\scshape{}C.~J. Fewster\egroup{} and
	\bgroup\scshape{}K.~Rejzner\egroup{}, {Algebraic Quantum Field Theory -- An
		Introduction},  in \emph{Progress and Visions in Quantum Theory in View of
		Gravity} (\bgroup\scshape{}F.~Finster\egroup{},
	\bgroup\scshape{}D.~Giulini\egroup{}, \bgroup\scshape{}J.~Kleiner\egroup{},
	and \bgroup\scshape{}J.~Tolksdorf\egroup{}, eds.), Birkh\"{a}user Basel,
	2020, pp.~1--61. \doi{10.1007/978-3-030-38941-3_1}.
	
	\bibitem[FV03]{Fewster-Verch2003}
	\bgroup\scshape{}C.~J. Fewster\egroup{} and \bgroup\scshape{}R.~Verch\egroup{},
	{Stability of Quantum Systems at Three Scales: Passivity, Quantum Weak Energy
		Inequalities and the Microlocal Spectrum Condition},  \emph{Communications in
		Mathematical Physics} \textbf{240} (2003), 329--375.
	\doi{10.1007/s00220-003-0884-7}.
	
	\bibitem[FV15]{Fewster-Verch-AQFT2015}
	\bgroup\scshape{}C.~J. Fewster\egroup{} and \bgroup\scshape{}R.~Verch\egroup{},
	{Algebraic Quantum Field Theory in Curved Spacetimes},  in \emph{Advances in
		Algebraic Quantum Field Theory} (\bgroup\scshape{}R.~Brunetti\egroup{},
	\bgroup\scshape{}C.~Dappiaggi\egroup{},
	\bgroup\scshape{}K.~Fredenhagen\egroup{}, and
	\bgroup\scshape{}J.~Yngvason\egroup{}, eds.), Springer International
	Publishing, 2015, pp.~125--189. \doi{10.1007/978-3-319-21353-8_4}.
	
	\bibitem[FV20]{FV-measurement}
	\bgroup\scshape{}C.~J. Fewster\egroup{} and \bgroup\scshape{}R.~Verch\egroup{},
	{Quantum Fields and Local Measurements},  \emph{Communications in
		Mathematical Physics} \textbf{378} (2020), 851--889.
	\doi{10.1007/s00220-020-03800-6}.
	
	\bibitem[FV23]{FV-measurement-enc}
	\bgroup\scshape{}C.~J. Fewster\egroup{} and \bgroup\scshape{}R.~Verch\egroup{},
	{Measurement in Quantum Field Theory}, 2023, {to} appear in Encyclopedia of
	Mathematical Physics. \doi{10.48550/ARXIV.2304.13356}.
	
	\bibitem[Fre15]{Fredenhagen-Intro-AQFT}
	\bgroup\scshape{}K.~Fredenhagen\egroup{}, {An Introduction to Algebraic Quantum
		Field Theory},  in \emph{Advances in Algebraic Quantum Field Theory}
	(\bgroup\scshape{}R.~Brunetti\egroup{},
	\bgroup\scshape{}C.~Dappiaggi\egroup{},
	\bgroup\scshape{}K.~Fredenhagen\egroup{}, and
	\bgroup\scshape{}J.~Yngvason\egroup{}, eds.), Springer International
	Publishing, 2015, pp.~1--30. \doi{10.1007/978-3-319-21353-8_1}.
	
	\bibitem[FH90]{Fredenhagen-Haag1990}
	\bgroup\scshape{}K.~Fredenhagen\egroup{} and \bgroup\scshape{}R.~Haag\egroup{},
	On the derivation of {Hawking} radiation associated with the formation of a
	black hole,  \emph{Communications in Mathematical Physics} \textbf{127}
	(1990), 273--284. \doi{10.1007/bf02096757}.
	
	\bibitem[FM04a]{Froehlich-Merkli-Another-RTE}
	\bgroup\scshape{}J.~Fr\"{o}hlich\egroup{} and
	\bgroup\scshape{}M.~Merkli\egroup{}, {Another Return of ``Return to
		Equilibrium''},  \emph{Communications in Mathematical Physics} \textbf{251}
	(2004), 235--262. \doi{10.1007/s00220-004-1176-6}.
	
	\bibitem[FM04b]{Froehlich-Merkli-Thermal-Ionization}
	\bgroup\scshape{}J.~Fr\"{o}hlich\egroup{} and
	\bgroup\scshape{}M.~Merkli\egroup{}, {Thermal Ionization},
	\emph{Mathematical Physics, Analysis and Geometry} \textbf{7} (2004),
	239--287. \doi{10.1023/b:mpag.0000034613.13746.8a}.
	
	\bibitem[Ful73]{Fulling1973}
	\bgroup\scshape{}S.~A. Fulling\egroup{}, {Nonuniqueness of Canonical Field
		Quantization in Riemannian Space-Time},  \emph{Physical Review D} \textbf{7}
	(1973), 2850--2862. \doi{10.1103/physrevd.7.2850}.
	
	\bibitem[GPS23]{Galanda-Pinamonti-Sangaletti2023}
	\bgroup\scshape{}S.~Galanda\egroup{}, \bgroup\scshape{}N.~Pinamonti\egroup{},
	and \bgroup\scshape{}L.~Sangaletti\egroup{}, Secular growths and their
	relation to equilibrium states in perturbative {QFT}, 2023.
	\doi{10.48550/ARXIV.2312.00556}.
	
	\bibitem[GE16]{Gogolin-Eisert2016}
	\bgroup\scshape{}C.~Gogolin\egroup{} and \bgroup\scshape{}J.~Eisert\egroup{},
	Equilibration, thermalisation, and the emergence of statistical mechanics in
	closed quantum systems,  \emph{Reports on Progress in Physics} \textbf{79}
	(2016), 056001. \doi{10.1088/0034-4885/79/5/056001}.
	
	\bibitem[GJMT20]{Good-JA-Moustos-Temirkhan2020}
	\bgroup\scshape{}M.~Good\egroup{}, \bgroup\scshape{}B.~A.
	{Ju\'{a}rez-Aubry}\egroup{}, \bgroup\scshape{}D.~Moustos\egroup{}, and
	\bgroup\scshape{}M.~Temirkhan\egroup{}, Unruh-like effects: effective
	temperatures along stationary worldlines,  \emph{Journal of High Energy
		Physics} \textbf{2020} (2020), 59. \doi{10.1007/jhep06(2020)059}.
	
	\bibitem[GPV17]{Gransee-Pinamonti-Verch-LKMS}
	\bgroup\scshape{}M.~Gransee\egroup{}, \bgroup\scshape{}N.~Pinamonti\egroup{},
	and \bgroup\scshape{}R.~Verch\egroup{}, {KMS}-like properties of local
	equilibrium states in quantum field theory,  \emph{Journal of Geometry and
		Physics} \textbf{117} (2017), 15--35. \doi{10.1016/j.geomphys.2017.02.014}.
	
	\bibitem[HHW67]{HHW1967}
	\bgroup\scshape{}R.~Haag\egroup{}, \bgroup\scshape{}N.~M. Hugenholtz\egroup{},
	and \bgroup\scshape{}M.~Winnink\egroup{}, On the equilibrium states in
	quantum statistical mechanics,  \emph{Communications in Mathematical Physics}
	\textbf{5} (1967), 215--236. \doi{10.1007/bf01646342}.
	
	\bibitem[Haa96]{Haag1996}
	\bgroup\scshape{}R.~Haag\egroup{}, \emph{{Local Quantum Physics}}, Springer
	Berlin Heidelberg, 1996. \doi{10.1007/978-3-642-61458-3}.
	
	\bibitem[HK64]{Haag-Kastler}
	\bgroup\scshape{}R.~Haag\egroup{} and \bgroup\scshape{}D.~Kastler\egroup{}, {An
		Algebraic Approach to Quantum Field Theory},  \emph{Journal of Mathematical
		Physics} \textbf{5} (1964), 848--861. \doi{10.1063/1.1704187}.
	
	\bibitem[HKT74]{Haag-Kastler-Trych-Pohlmeyer}
	\bgroup\scshape{}R.~Haag\egroup{}, \bgroup\scshape{}D.~Kastler\egroup{}, and
	\bgroup\scshape{}E.~B. {Trych-Pohlmeyer}\egroup{}, Stability and equilibrium
	states,  \emph{Communications in Mathematical Physics} \textbf{38} (1974),
	173--193. \doi{10.1007/bf01651541}.
	
	\bibitem[HV18]{Hack-Verch}
	\bgroup\scshape{}T.-P. Hack\egroup{} and \bgroup\scshape{}R.~Verch\egroup{},
	Non-equilibrium steady states for the interacting {Klein-Gordon} field in 1+3
	dimensions, 2018. \doi{10.48550/ARXIV.1806.00504}.
	
	\bibitem[Haw74]{Hawking1974}
	\bgroup\scshape{}S.~W. Hawking\egroup{}, Black hole explosions?,
	\emph{{Nature}} \textbf{248} (1974), 30--31. \doi{10.1038/248030a0}.
	
	\bibitem[Haw75]{Hawking1975}
	\bgroup\scshape{}S.~W. Hawking\egroup{}, Particle creation by black holes,
	\emph{Communications in Mathematical Physics} \textbf{43} (1975), 199--220.
	\doi{10.1007/bf02345020}.
	
	\bibitem[Hep72]{Hepp1972}
	\bgroup\scshape{}K.~Hepp\egroup{}, Quantum theory of measurement and
	macroscopic observables,  \emph{Helvetica Physica Acta} \textbf{45} (1972),
	237--248. \doi{10.5169/seals-114381}.
	
	\bibitem[HT70]{Herman-Takesaki1970}
	\bgroup\scshape{}R.~H. Herman\egroup{} and
	\bgroup\scshape{}M.~Takesaki\egroup{}, States and automorphism groups of
	operator algebras,  \emph{Communications in Mathematical Physics} \textbf{19}
	(1970), 142--160. \doi{10.1007/bf01646631}.
	
	\bibitem[HW15]{Hollands-Wald-2015}
	\bgroup\scshape{}S.~Hollands\egroup{} and \bgroup\scshape{}R.~M. Wald\egroup{},
	Quantum fields in curved spacetime,  \emph{Physics Reports} \textbf{574}
	(2015), 1--35. \doi{10.1016/j.physrep.2015.02.001}.
	
	\bibitem[HS95a]{Huebner-Spohn-decay}
	\bgroup\scshape{}M.~H\"{u}bner\egroup{} and \bgroup\scshape{}H.~Spohn\egroup{},
	{Radiative Decay: Nonperturbative Approaches},  \emph{Reviews in Mathematical
		Physics} \textbf{07} (1995), 363--387. \doi{10.1142/s0129055x95000165}.
	
	\bibitem[HS95b]{Huebner-Spohn}
	\bgroup\scshape{}M.~H\"{u}bner\egroup{} and \bgroup\scshape{}H.~Spohn\egroup{},
	Spectral properties of the spin-boson {Hamiltonian},  \emph{Annales de
		l'Institut Henri Poincar\'{e} Physique th\'{e}orique} \textbf{62} (1995),
	289--323. Available at
	\url{http://www.numdam.org/item/AIHPA_1995__62_3_289_0}.
	
	\bibitem[Hug67]{Hugenholtz-factor}
	\bgroup\scshape{}N.~M. Hugenholtz\egroup{}, On the factor type of equilibrium
	states in quantum statistical mechanics,  \emph{Communications in
		Mathematical Physics} \textbf{6} (1967), 189--193. \doi{10.1007/bf01659975}.
	
	\bibitem[HR86]{Hume-Robinson1986}
	\bgroup\scshape{}L.~Hume\egroup{} and \bgroup\scshape{}D.~W. Robinson\egroup{},
	Return to equilibrium in the {$XY$} model,  \emph{Journal of Statistical
		Physics} \textbf{44} (1986), 829--848. \doi{10.1007/bf01011909}.
	
	\bibitem[JNW10]{Jaekel-Narnhofer-Wreszinski}
	\bgroup\scshape{}C.~D. J\"{a}kel\egroup{},
	\bgroup\scshape{}H.~Narnhofer\egroup{}, and \bgroup\scshape{}W.~F.
	Wreszinski\egroup{}, On the mixing property for a class of states of
	relativistic quantum fields,  \emph{Journal of Mathematical Physics}
	\textbf{51} (2010), 052703. \doi{10.1063/1.3372623}.
	
	\bibitem[JP96a]{JP1}
	\bgroup\scshape{}V.~Jak{\v{s}}i{\'{c}}\egroup{} and \bgroup\scshape{}C.-A.
	Pillet\egroup{}, On a model for quantum friction, {II}. {Fermi}'s golden rule
	and dynamics at positive temperature,  \emph{Communications in Mathematical
		Physics} \textbf{176} (1996), 619--644. \doi{10.1007/bf02099252}.
	
	\bibitem[JP96b]{JP2}
	\bgroup\scshape{}V.~Jak{\v{s}}i{\'{c}}\egroup{} and \bgroup\scshape{}C.-A.
	Pillet\egroup{}, On a model for quantum friction {III}. {Ergodic} properties
	of the spin-boson system,  \emph{Communications in Mathematical Physics}
	\textbf{178} (1996), 627--651. \doi{10.1007/bf02108818}.
	
	\bibitem[JP97]{JP-thermal-relaxation}
	\bgroup\scshape{}V.~Jak{\v{s}}i{\'{c}}\egroup{} and \bgroup\scshape{}C.-A.
	Pillet\egroup{}, Spectral theory of thermal relaxation,  \emph{Journal of
		Mathematical Physics} \textbf{38} (1997), 1757--1780. \doi{10.1063/1.531912}.
	
	\bibitem[JP01]{Jaksic-Pillet2001}
	\bgroup\scshape{}V.~Jak{\v{s}}i{\'{c}}\egroup{} and \bgroup\scshape{}C.-A.
	Pillet\egroup{}, {A Note on Eigenvalues of Liouvilleans},  \emph{Journal of
		Statistical Physics} \textbf{105} (2001), 937--941.
	\doi{10.1023/a:1013561529682}.
	
	\bibitem[JP02a]{Jaksic-Pillet-NESS-math}
	\bgroup\scshape{}V.~Jak{\v{s}}i{\'{c}}\egroup{} and \bgroup\scshape{}C.-A.
	Pillet\egroup{}, {Mathematical Theory of Non-Equilibrium Quantum Statistical
		Mechanics},  \emph{Journal of Statistical Physics} \textbf{108} (2002),
	787--829. \doi{10.1023/a:1019818909696}.
	
	\bibitem[JP02b]{Jaksic-Pillet-NESS}
	\bgroup\scshape{}V.~Jak{\v{s}}i{\'{c}}\egroup{} and \bgroup\scshape{}C.-A.
	Pillet\egroup{}, {Non-Equilibrium Steady States of Finite Quantum Systems
		Coupled to Thermal Reservoirs},  \emph{Communications in Mathematical
		Physics} \textbf{226} (2002), 131--162. \doi{10.1007/s002200200602}.
	
	\bibitem[JM19]{JuarezAubry2019}
	\bgroup\scshape{}B.~A. {Ju{\'{a}}rez-Aubry}\egroup{} and
	\bgroup\scshape{}D.~Moustos\egroup{}, Asymptotic states for stationary
	{Unruh}-{DeWitt} detectors,  \emph{Physical Review D} \textbf{100} (2019),
	025018. \doi{10.1103/physrevd.100.025018}.
	
	\bibitem[KR97]{Kadison-Ringrose-II}
	\bgroup\scshape{}R.~V. Kadison\egroup{} and \bgroup\scshape{}J.~R.
	Ringrose\egroup{}, \emph{{Fundamentals of the Theory of Operator Algebras,
			Volume II}}, American Mathematical Society, 1997. \doi{10.1090/gsm/016}.
	
	\bibitem[Kay78]{Kay1978}
	\bgroup\scshape{}B.~S. Kay\egroup{}, Linear spin-zero quantum fields in
	external gravitational and scalar fields {I}: {A} one particle structure for
	the stationary case,  \emph{Communications in Mathematical Physics}
	\textbf{62} (1978), 55--70. \doi{10.1007/bf01940330}.
	
	\bibitem[Kay85a]{kay-double-wedge}
	\bgroup\scshape{}B.~S. Kay\egroup{}, The double-wedge algebra for quantum
	fields on {Schwarzschild} and {Minkowski} spacetimes,  \emph{Communications
		in Mathematical Physics} \textbf{100} (1985), 57--81.
	\doi{10.1007/bf01212687}.
	
	\bibitem[Kay85b]{kay-purification}
	\bgroup\scshape{}B.~S. Kay\egroup{}, Purification of {KMS} states,
	\emph{Helvetica Physica Acta} \textbf{58} (1985), 1030--1040.
	\doi{10.5169/seals-115635}.
	
	\bibitem[Kay85c]{kay-uniqueness}
	\bgroup\scshape{}B.~S. Kay\egroup{}, A uniqueness result for quasi-free {KMS}
	states,  \emph{Helvetica Physica Acta} \textbf{58} (1985), 1017--1029.
	\doi{10.5169/seals-115634}.
	
	\bibitem[KW91]{Kay-Wald}
	\bgroup\scshape{}B.~S. Kay\egroup{} and \bgroup\scshape{}R.~M. Wald\egroup{},
	Theorems on the uniqueness and thermal properties of stationary, nonsingular,
	quasifree states on spacetimes with a bifurcate {Killing} horizon,
	\emph{Physics Reports} \textbf{207} (1991), 49--136.
	\doi{10.1016/0370-1573(91)90015-E}.
	
	\bibitem[KM15]{Khavkine-Moretti}
	\bgroup\scshape{}I.~Khavkine\egroup{} and \bgroup\scshape{}V.~Moretti\egroup{},
	{Algebraic QFT in Curved Spacetime and Quasifree Hadamard States: An
		Introduction},  in \emph{Advances in Algebraic Quantum Field Theory}
	(\bgroup\scshape{}R.~Brunetti\egroup{},
	\bgroup\scshape{}C.~Dappiaggi\egroup{},
	\bgroup\scshape{}K.~Fredenhagen\egroup{}, and
	\bgroup\scshape{}J.~Yngvason\egroup{}, eds.), Springer International
	Publishing, 2015, pp.~191--251. \doi{10.1007/978-3-319-21353-8_5}.
	
	\bibitem[KMS14]{Koenenberg-Merkli-Song-2014}
	\bgroup\scshape{}M.~K{\"o}nenberg\egroup{},
	\bgroup\scshape{}M.~Merkli\egroup{}, and \bgroup\scshape{}H.~Song\egroup{},
	{Ergodicity of the Spin-Boson Model for Arbitrary Coupling Strength},
	\emph{Communications in Mathematical Physics} \textbf{336} (2014), 261--285.
	\doi{10.1007/s00220-014-2242-3}.
	
	\bibitem[K{\"o}n11]{Koenenberg2011}
	\bgroup\scshape{}M.~K{\"o}nenberg\egroup{}, Return to equilibrium for an
	anharmonic oscillator coupled to a heat bath,  \emph{Journal of Mathematical
		Physics} \textbf{52} (2011), 022110. \doi{10.1063/1.3544476}.
	
	\bibitem[Koo31]{Koopman1931}
	\bgroup\scshape{}B.~O. Koopman\egroup{}, {Hamiltonian Systems and
		Transformations in Hilbert Space},  \emph{Proceedings of the National Academy
		of Sciences} \textbf{17} (1931), 315--318. \doi{10.1073/pnas.17.5.315}.
	
	\bibitem[KFGV77]{Kossakowski1977}
	\bgroup\scshape{}A.~Kossakowski\egroup{},
	\bgroup\scshape{}A.~Frigerio\egroup{}, \bgroup\scshape{}V.~Gorini\egroup{},
	and \bgroup\scshape{}M.~Verri\egroup{}, Quantum detailed balance and {KMS}
	condition,  \emph{Communications in Mathematical Physics} \textbf{57} (1977),
	97--110. \doi{10.1007/bf01625769}.
	
	\bibitem[Kub57]{Kubo1957}
	\bgroup\scshape{}R.~Kubo\egroup{}, {Statistical-mechanical theory of
		irreversible processes I. General theory and simple applications to magnetic
		and conduction problems},  \emph{Journal of the Physical Society of Japan}
	\textbf{12} (1957), 570--586. \doi{10.1143/jpsj.12.570}.
	
	\bibitem[Kuc02]{Kuckert2002}
	\bgroup\scshape{}B.~Kuckert\egroup{}, {Covariant Thermodynamics of Quantum
		Systems: Passivity, Semipassivity, and the Unruh Effect},  \emph{Annals of
		Physics} \textbf{295} (2002), 216--229. \doi{10.1006/aphy.2001.6220}.
	
	\bibitem[LM04]{Landsberg-Matsas2004}
	\bgroup\scshape{}P.~T. Landsberg\egroup{} and \bgroup\scshape{}G.~E.~A.
	Matsas\egroup{}, The impossibility of a universal relativistic temperature
	transformation,  \emph{Physica A: Statistical Mechanics and its Applications}
	\textbf{340} (2004), 92--94. \doi{10.1016/j.physa.2004.03.081}.
	
	\bibitem[LM96]{Landsberg-Matsas1996}
	\bgroup\scshape{}P.~T. Landsberg\egroup{} and \bgroup\scshape{}G.~E.~A.
	Matsas\egroup{}, Laying the ghost of the relativistic temperature
	transformation,  \emph{Physics Letters A} \textbf{223} (1996), 401--403.
	\doi{10.1016/s0375-9601(96)00791-8}.
	
	\bibitem[Lan07]{Landsman2007}
	\bgroup\scshape{}N.~P. Landsman\egroup{}, {Between Classical and Quantum},  in
	\emph{Handbook of the Philosophy of Science: Philosophy of Physics}
	(\bgroup\scshape{}J.~Butterfield\egroup{} and
	\bgroup\scshape{}J.~Earman\egroup{}, eds.), Elsevier, 2007, pp.~417--553.
	\doi{10.1016/b978-044451560-5/50008-7}.
	
	\bibitem[LR72]{Lanford-Robinson1972}
	\bgroup\scshape{}O.~E. {Lanford, III.}\egroup{} and \bgroup\scshape{}D.~W.
	Robinson\egroup{}, Approach to equilibrium of free quantum systems,
	\emph{Communications in Mathematical Physics} \textbf{24} (1972), 193--210.
	\doi{10.1007/bf01877712}.
	
	\bibitem[LCD{\etalchar{+}}87]{Leggett1987}
	\bgroup\scshape{}A.~J. Leggett\egroup{},
	\bgroup\scshape{}S.~Chakravarty\egroup{}, \bgroup\scshape{}A.~T.
	Dorsey\egroup{}, \bgroup\scshape{}M.~P.~A. Fisher\egroup{},
	\bgroup\scshape{}A.~Garg\egroup{}, and \bgroup\scshape{}W.~Zwerger\egroup{},
	Dynamics of the dissipative two-state system,  \emph{Reviews of Modern
		Physics} \textbf{59} (1987), 1--85. \doi{10.1103/revmodphys.59.1}.
	
	\bibitem[LPTM23]{Lima-Patterson-Tjoa-Mann2023}
	\bgroup\scshape{}C.~Lima\egroup{}, \bgroup\scshape{}E.~Patterson\egroup{},
	\bgroup\scshape{}E.~Tjoa\egroup{}, and \bgroup\scshape{}R.~B. Mann\egroup{},
	Unruh phenomena and thermalization for qudit detectors,  \emph{Physical
		Review D} \textbf{108} (2023), 105020. \doi{10.1103/physrevd.108.105020}.
	
	\bibitem[LS08]{Louko-Satz}
	\bgroup\scshape{}J.~Louko\egroup{} and \bgroup\scshape{}A.~Satz\egroup{},
	Transition rate of the {Unruh}-{DeWitt} detector in curved spacetime,
	\emph{Classical and Quantum Gravity} \textbf{25} (2008), 055012.
	\doi{10.1088/0264-9381/25/5/055012}.
	
	\bibitem[LT16]{Louko-Toussaint2016}
	\bgroup\scshape{}J.~Louko\egroup{} and \bgroup\scshape{}V.~Toussaint\egroup{},
	{Unruh-DeWitt} detector's response to fermions in flat spacetimes,
	\emph{Physical Review D} \textbf{94} (2016), 064027.
	\doi{10.1103/physrevd.94.064027}.
	
	\bibitem[MS59]{Martin-Schwinger1959}
	\bgroup\scshape{}P.~C. Martin\egroup{} and
	\bgroup\scshape{}J.~Schwinger\egroup{}, {Theory of many-particle systems. I},
	\emph{Physical Review} \textbf{115} (1959), 1342--1373.
	\doi{10.1103/physrev.115.1342}.
	
	\bibitem[MSB08]{Merkli-Sigal-Berman}
	\bgroup\scshape{}M.~Merkli\egroup{}, \bgroup\scshape{}I.~M. Sigal\egroup{}, and
	\bgroup\scshape{}G.~P. Berman\egroup{}, Resonance theory of decoherence and
	thermalization,  \emph{Annals of Physics} \textbf{323} (2008), 373--412.
	\doi{10.1016/j.aop.2007.04.013}.
	
	\bibitem[Mer01]{Merkli-Positive-Commutators}
	\bgroup\scshape{}M.~Merkli\egroup{}, {Positive Commutators in Non-Equilibrium
		Quantum Statistical Mechanics},  \emph{Communications in Mathematical
		Physics} \textbf{223} (2001), 327--362. \doi{10.1007/s002200100545}.
	
	\bibitem[Mer05]{Merkli-condensate}
	\bgroup\scshape{}M.~Merkli\egroup{}, {Stability of Equilibria with a
		Condensate},  \emph{Communications in Mathematical Physics} \textbf{257}
	(2005), 621--640. \doi{10.1007/s00220-005-1352-3}.
	
	\bibitem[Mer06]{Merkli-ideal-2006}
	\bgroup\scshape{}M.~Merkli\egroup{}, {The Ideal Quantum Gas},  in \emph{Open
		Quantum Systems I} (\bgroup\scshape{}S.~Attal\egroup{},
	\bgroup\scshape{}A.~Joye\egroup{}, and \bgroup\scshape{}C.-A.
	Pillet\egroup{}, eds.), \emph{Lecture Notes in Mathematics} \textbf{1880},
	Springer Berlin Heidelberg, 2006, pp.~183--233.
	\doi{10.1007/3-540-33922-1_5}.
	
	\bibitem[Mer07]{Merkli-LSO-2007}
	\bgroup\scshape{}M.~Merkli\egroup{}, Level shift operators for open quantum
	systems,  \emph{Journal of Mathematical Analysis and Applications}
	\textbf{327} (2007), 376--399. \doi{10.1016/j.jmaa.2006.04.030}.
	
	\bibitem[Mer20]{Merkli-Quantum-Markovian}
	\bgroup\scshape{}M.~Merkli\egroup{}, Quantum {Markovian} master
	equations{$\colon\!$} {Resonance} theory shows validity for all time scales,
	\emph{Annals of Physics} \textbf{412} (2020), 167996.
	\doi{10.1016/j.aop.2019.167996}.
	
	\bibitem[Mer22]{Merkli2021-2}
	\bgroup\scshape{}M.~Merkli\egroup{}, {Dynamics of Open Quantum Systems {II},
		Markovian Approximation},  \emph{Quantum} \textbf{6} (2022), 616.
	\doi{10.22331/q-2022-01-03-616}.
	
	\bibitem[MMS07]{Merkli-Mueck-Sigal-NESS}
	\bgroup\scshape{}M.~Merkli\egroup{}, \bgroup\scshape{}M.~M\"{u}ck\egroup{}, and
	\bgroup\scshape{}I.~M. Sigal\egroup{}, {Theory of Non-Equilibrium Stationary
		States as a Theory of Resonances},  \emph{Annales Henri Poincar\'{e}}
	\textbf{8} (2007), 1539--1593. \doi{10.1007/s00023-007-0346-4}.
	
	\bibitem[MS15]{Merkli-Song}
	\bgroup\scshape{}M.~Merkli\egroup{} and \bgroup\scshape{}H.~Song\egroup{},
	{Overlapping Resonances in Open Quantum Systems},  \emph{Annales Henri
		Poincar\'{e}} \textbf{16} (2015), 1397--1427.
	\doi{10.1007/s00023-014-0349-x}.
	
	\bibitem[M{\o}l14]{Moeller2014}
	\bgroup\scshape{}J.~S. M{\o}ller\egroup{}, {Fully coupled Pauli-Fierz systems
		at zero and positive temperature},  \emph{Journal of Mathematical Physics}
	\textbf{55} (2014), 075203. \doi{10.1063/1.4879239}.
	
	\bibitem[Mou18]{Moustos2018}
	\bgroup\scshape{}D.~Moustos\egroup{}, Asymptotic states of accelerated
	detectors and universality of the {Unruh} effect,  \emph{Physical Review D}
	\textbf{98} (2018), 065006. \doi{10.1103/physrevd.98.065006}.
	
	\bibitem[Mou22]{Moustos2022}
	\bgroup\scshape{}D.~Moustos\egroup{}, {Uniformly accelerated Brownian
		oscillator in (2+1)D: Temperature-dependent dissipation and frequency shift},
	\emph{Physics Letters B} \textbf{829} (2022), 137115.
	\doi{10.1016/j.physletb.2022.137115}.
	
	\bibitem[MA17]{Moustos-Anastopoulos2017}
	\bgroup\scshape{}D.~Moustos\egroup{} and
	\bgroup\scshape{}C.~Anastopoulos\egroup{}, {Non-Markovian time evolution of
		an accelerated qubit},  \emph{Physical Review D} \textbf{95} (2017), 025020.
	\doi{10.1103/physrevd.95.025020}.
	
	\bibitem[MH80]{Mueller-Herold}
	\bgroup\scshape{}U.~M\"{u}ller-Herold\egroup{}, {Disjointness of $\beta$-KMS
		states with different chemical potential},  \emph{Letters in Mathematical
		Physics} \textbf{4} (1980), 45--48. \doi{10.1007/bf00419804}.
	
	\bibitem[NT89]{Narnhofer-Thirring_mixing}
	\bgroup\scshape{}H.~Narnhofer\egroup{} and
	\bgroup\scshape{}W.~Thirring\egroup{}, Mixing properties of quantum systems,
	\emph{Journal of Statistical Physics} \textbf{57} (1989), 811--825.
	\doi{10.1007/bf01022834}.
	
	\bibitem[NTW88]{Narnhofer-Thirring-Wiklicky}
	\bgroup\scshape{}H.~Narnhofer\egroup{}, \bgroup\scshape{}W.~Thirring\egroup{},
	and \bgroup\scshape{}H.~Wiklicky\egroup{}, Transitivity and ergodicity of
	quantum systems,  \emph{Journal of Statistical Physics} \textbf{52} (1988),
	1097--1112. \doi{10.1007/bf01019741}.
	
	\bibitem[Nar77]{Narnhofer1977}
	\bgroup\scshape{}H.~Narnhofer\egroup{}, {Kommutative Automorphismen und
		Gleichgewichtszust\"{a}nde},  \emph{Acta Physica Austriaca} \textbf{47}
	(1977), 1--29.
	
	\bibitem[NT91]{Narnhofer-Thirring_Galilei1991}
	\bgroup\scshape{}H.~Narnhofer\egroup{} and
	\bgroup\scshape{}W.~Thirring\egroup{}, Galilei-invariant quantum field
	theories with pair interactions: a review,  \emph{International Journal of
		Modern Physics A} \textbf{06} (1991), 2937--2970.
	\doi{10.1142/s0217751x91001453}.
	
	\bibitem[Oji86]{Ojima1986}
	\bgroup\scshape{}I.~Ojima\egroup{}, Lorentz invariance vs. temperature in
	{QFT},  \emph{Letters in Mathematical Physics} \textbf{11} (1986), 73--80.
	\doi{10.1007/bf00417467}.
	
	\bibitem[Oji04]{Ojima2004}
	\bgroup\scshape{}I.~Ojima\egroup{}, Temperature as order parameter of broken
	scale invariance,  \emph{Publications of the Research Institute for
		Mathematical Sciences} \textbf{40} (2004), 731--756.
	\doi{10.2977/prims/1145475491}.
	
	\bibitem[Oji05]{Ojima2005}
	\bgroup\scshape{}I.~Ojima\egroup{}, {Micro-Macro Duality in Quantum Physics},
	in \emph{Stochastic Analysis: Classical and Quantum}
	(\bgroup\scshape{}T.~Hida\egroup{}, ed.), World Scientific, 2005,
	pp.~143--161. \doi{10.1142/9789812701541_0012}.
	
	\bibitem[PA20]{Papadatos-Anastopoulos2020}
	\bgroup\scshape{}N.~Papadatos\egroup{} and
	\bgroup\scshape{}C.~Anastopoulos\egroup{}, Relativistic quantum
	thermodynamics of moving systems,  \emph{Physical Review D} \textbf{102}
	(2020), 085005. \doi{10.1103/physrevd.102.085005}.
	
	\bibitem[PF24]{Papageorgiou-Fraser2024}
	\bgroup\scshape{}M.~Papageorgiou\egroup{} and
	\bgroup\scshape{}D.~Fraser\egroup{}, {Eliminating the ‘Impossible’:
		Recent Progress on Local Measurement Theory for Quantum Field Theory},
	\emph{Foundations of Physics} \textbf{54} (2024), 26.
	\doi{10.1007/s10701-024-00756-8}.
	
	\bibitem[PF38]{Pauli-Fierz1938}
	\bgroup\scshape{}W.~Pauli\egroup{} and \bgroup\scshape{}M.~Fierz\egroup{}, {Zur
		Theorie der Emission langwelliger Lichtquanten},  \emph{Il Nuovo Cimento}
	\textbf{15} (1938), 167--188. \doi{10.1007/bf02958939}.
	
	\bibitem[Per21]{Perche-thermalization}
	\bgroup\scshape{}T.~R. Perche\egroup{}, General features of the thermalization
	of particle detectors and the {Unruh} effect,  \emph{Physical Review D}
	\textbf{104} (2021), 065001. \doi{10.1103/physrevd.104.065001}.
	
	\bibitem[PPTM24]{Perche-PoloGomez-Torres-MM2024}
	\bgroup\scshape{}T.~R. Perche\egroup{},
	\bgroup\scshape{}J.~{Polo-G\'{o}mez}\egroup{}, \bgroup\scshape{}B.~{\relax de
		S. L}. {Torres}\egroup{}, and
	\bgroup\scshape{}E.~{Mart\'{i}n-Mart\'{i}nez}\egroup{}, Particle detectors
	from localized quantum field theories,  \emph{Physical Review D} \textbf{109}
	(2024), 045013. \doi{10.1103/physrevd.109.045013}.
	
	\bibitem[Pet90]{Petz-CCR}
	\bgroup\scshape{}D.~Petz\egroup{}, \emph{{An Invitation to the Algebra of
			Canonical Commutation Relations}}, \emph{Series A: Mathematical Physics}
	\textbf{2}, Leuven University Press, 1990.
	
	\bibitem[Pil06]{Pillet-Attal2006}
	\bgroup\scshape{}C.-A. Pillet\egroup{}, {Quantum Dynamical Systems},  in
	\emph{Open Quantum Systems I} (\bgroup\scshape{}S.~Attal\egroup{},
	\bgroup\scshape{}A.~Joye\egroup{}, and \bgroup\scshape{}C.-A.
	Pillet\egroup{}, eds.), \emph{Lecture Notes in Mathematics} \textbf{1880},
	Springer Berlin Heidelberg, 2006, pp.~107--182.
	\doi{10.1007/3-540-33922-1_4}.
	
	\bibitem[PGM22]{PoloGomez-Garay-MM2022}
	\bgroup\scshape{}J.~{Polo-G\'{o}mez}\egroup{}, \bgroup\scshape{}L.~J.
	Garay\egroup{}, and \bgroup\scshape{}E.~{Mart\'{i}n-Mart\'{i}nez}\egroup{}, A
	detector-based measurement theory for quantum field theory,  \emph{Physical
		Review D} \textbf{105} (2022), 065003. \doi{10.1103/physrevd.105.065003}.
	
	\bibitem[PW78]{Pusz-Woronowicz}
	\bgroup\scshape{}W.~Pusz\egroup{} and \bgroup\scshape{}S.~L.
	Woronowicz\egroup{}, Passive states and {KMS} states for general quantum
	systems,  \emph{Communications in Mathematical Physics} \textbf{58} (1978),
	273--290. \doi{10.1007/bf01614224}.
	
	\bibitem[Rad96]{Radzikowski1996}
	\bgroup\scshape{}M.~J. Radzikowski\egroup{}, Micro-local approach to the
	{Hadamard} condition in quantum field theory on curved space-time,
	\emph{Communications in Mathematical Physics} \textbf{179} (1996), 529--553.
	\doi{10.1007/bf02100096}.
	
	\bibitem[RS72]{Reed-Simon-I}
	\bgroup\scshape{}M.~Reed\egroup{} and \bgroup\scshape{}B.~Simon\egroup{},
	\emph{{Methods of Modern Mathematical Physics I: Functional Analysis}},
	Academic Press Inc., 1972.
	
	\bibitem[RS75]{Reed-Simon-II}
	\bgroup\scshape{}M.~Reed\egroup{} and \bgroup\scshape{}B.~Simon\egroup{},
	\emph{{Methods of Modern Mathematical Physics II: Fourier Analysis,
			Self-Adjointness}}, Academic Press Inc., 1975.
	
	\bibitem[RW85]{Requardt-Wreszinski}
	\bgroup\scshape{}M.~Requardt\egroup{} and \bgroup\scshape{}W.~F.
	Wreszinski\egroup{}, Temperature states, ground states and relativistic
	vacuum states in the context of symmetry breakdown,  \emph{Journal of Physics
		A: Mathematical and General} \textbf{18} (1985), 705--712.
	\doi{10.1088/0305-4470/18/4/017}.
	
	\bibitem[Rin77]{Rindler1977}
	\bgroup\scshape{}W.~Rindler\egroup{}, \emph{{Essential Relativity}}, Springer
	Berlin Heidelberg, 1977. \doi{10.1007/978-3-642-86650-0}.
	
	\bibitem[RH12]{Rivas-Huelga}
	\bgroup\scshape{}{\'{A}}.~Rivas\egroup{} and \bgroup\scshape{}S.~F.
	Huelga\egroup{}, \emph{{Open Quantum Systems: An Introduction}},
	\emph{SpringerBriefs in Physics}, Springer Berlin Heidelberg, 2012.
	\doi{10.1007/978-3-642-23354-8}.
	
	\bibitem[Rob73]{Robinson1973}
	\bgroup\scshape{}D.~W. Robinson\egroup{}, Return to equilibrium,
	\emph{Communications in Mathematical Physics} \textbf{31} (1973), 171--189.
	\doi{10.1007/bf01646264}.
	
	\bibitem[Roo72]{Roos1972}
	\bgroup\scshape{}H.~Roos\egroup{}, On quantum systems in thermal contact,
	\emph{Communications in Mathematical Physics} \textbf{26} (1972), 149--168.
	\doi{10.1007/bf01645701}.
	
	\bibitem[Rue00]{Ruelle-NESS}
	\bgroup\scshape{}D.~Ruelle\egroup{}, {Natural Nonequilibrium States in Quantum
		Statistical Mechanics},  \emph{Journal of Statistical Physics} \textbf{98}
	(2000), 57--75. \doi{10.1023/a:1018618704438}.
	
	\bibitem[SV00]{Sahlmann-Verch}
	\bgroup\scshape{}H.~Sahlmann\egroup{} and \bgroup\scshape{}R.~Verch\egroup{},
	{Passivity and Microlocal Spectrum Condition},  \emph{Communications in
		Mathematical Physics} \textbf{214} (2000), 705--731.
	\doi{10.1007/s002200000297}.
	
	\bibitem[Sak98]{Sakai1998}
	\bgroup\scshape{}S.~Sakai\egroup{}, \emph{{C$^\ast$-Algebras and
			W$^\ast$-Algebras}}, \emph{Classics in Mathematics} \textbf{60}, Springer
	Berlin Heidelberg, 1998. \doi{10.1007/978-3-642-61993-9}.
	
	\bibitem[San10]{Sanders2010}
	\bgroup\scshape{}K.~Sanders\egroup{}, The locally covariant {Dirac} field,
	\emph{Reviews in Mathematical Physics} \textbf{22} (2010), 381--430.
	\doi{10.1142/s0129055x10003990}.
	
	\bibitem[San13]{Sanders2013}
	\bgroup\scshape{}K.~Sanders\egroup{}, Thermal equilibrium states of a linear
	scalar quantum field in stationary spacetimes,  \emph{International Journal
		of Modern Physics A} \textbf{28} (2013), 1330010.
	\doi{10.1142/s0217751x1330010x}.
	
	\bibitem[Sch97]{Schechter}
	\bgroup\scshape{}E.~Schechter\egroup{}, \emph{{Handbook of Analysis and Its
			Foundations}}, Elsevier, Academic Press, 1997.
	\doi{10.1016/b978-0-12-622760-4.x5000-6}.
	
	\bibitem[SV08]{Schlemmer-Verch2008}
	\bgroup\scshape{}J.~Schlemmer\egroup{} and \bgroup\scshape{}R.~Verch\egroup{},
	{Local Thermal Equilibrium States and Quantum Energy Inequalities},
	\emph{Annales Henri Poincar{\'{e}}} \textbf{9} (2008), 945--978.
	\doi{10.1007/s00023-008-0380-x}.
	
	\bibitem[SCD81]{Candelas-Deutsch-Sciama}
	\bgroup\scshape{}D.~W. Sciama\egroup{}, \bgroup\scshape{}P.~Candelas\egroup{},
	and \bgroup\scshape{}D.~Deutsch\egroup{}, Quantum field theory, horizons and
	thermodynamics,  \emph{Advances in Physics} \textbf{30} (1981), 327--366.
	\doi{10.1080/00018738100101457}.
	
	\bibitem[Sew80]{Sewell-BW1980}
	\bgroup\scshape{}G.~L. Sewell\egroup{}, Relativity of temperature and the
	{Hawking} effect,  \emph{Physics Letters A} \textbf{79} (1980), 23--24.
	\doi{10.1016/0375-9601(80)90306-0}.
	
	\bibitem[Sew82]{Sewell1982}
	\bgroup\scshape{}G.~L. Sewell\egroup{}, Quantum fields on manifolds: {PCT} and
	gravitationally induced thermal states,  \emph{Annals of Physics}
	\textbf{141} (1982), 201--224. \doi{10.1016/0003-4916(82)90285-8}.
	
	\bibitem[Sew08]{Sewell2008}
	\bgroup\scshape{}G.~L. Sewell\egroup{}, On the question of temperature
	transformations under {Lorentz} and {Galilei} boosts,  \emph{Journal of
		Physics A: Mathematical and Theoretical} \textbf{41} (2008), 382003.
	\doi{10.1088/1751-8113/41/38/382003}.
	
	\bibitem[Sew09]{Sewell-rep2009}
	\bgroup\scshape{}G.~L. Sewell\egroup{}, Statistical thermodynamics of moving
	bodies,  \emph{Reports on Mathematical Physics} \textbf{64} (2009), 285--291.
	\doi{10.1016/s0034-4877(09)90033-7}.
	
	\bibitem[Sew10]{Sewell2010}
	\bgroup\scshape{}G.~L. Sewell\egroup{}, Note on the relativistic thermodynamics
	of moving bodies,  \emph{Journal of Physics A: Mathematical and Theoretical}
	\textbf{43} (2010), 485001. \doi{10.1088/1751-8113/43/48/485001}.
	
	\bibitem[SU01]{Sexl-Urbantke}
	\bgroup\scshape{}R.~U. Sexl\egroup{} and \bgroup\scshape{}H.~K.
	Urbantke\egroup{}, \emph{{Relativity, Groups, Particles: Special Relativity
			and Relativistic Symmetry in Field and Particle Physics}}, Springer Vienna,
	2001. \doi{10.1007/978-3-7091-6234-7}.
	
	\bibitem[SW70]{Sirugue-Winnink1970}
	\bgroup\scshape{}M.~Sirugue\egroup{} and \bgroup\scshape{}M.~Winnink\egroup{},
	Constraints imposed upon a state of a system that satisfies the {K.M.S.}
	boundary condition,  \emph{Communications in Mathematical Physics}
	\textbf{19} (1970), 161--168. \doi{10.1007/bf01646632}.
	
	\bibitem[Sol10]{Solveen2010}
	\bgroup\scshape{}C.~Solveen\egroup{}, Local thermal equilibrium in quantum
	field theory on flat and curved spacetimes,  \emph{Classical and Quantum
		Gravity} \textbf{27} (2010), 235002. \doi{10.1088/0264-9381/27/23/235002}.
	
	\bibitem[Sol12]{Solveen2012}
	\bgroup\scshape{}C.~Solveen\egroup{}, Local thermal equilibrium and {KMS}
	states in curved spacetime,  \emph{Classical and Quantum Gravity} \textbf{29}
	(2012), 245015. \doi{10.1088/0264-9381/29/24/245015}.
	
	\bibitem[Sor24]{Sorce-type-vN}
	\bgroup\scshape{}J.~Sorce\egroup{}, Notes on the type classification of von
	{Neumann} algebras,  \emph{Reviews in Mathematical Physics} \textbf{36}
	(2024), 2430002. \doi{10.1142/s0129055x24300024}.
	
	\bibitem[Sum06]{Summers-TT2006}
	\bgroup\scshape{}S.~J. Summers\egroup{}, {Tomita-Takesaki Modular Theory},  in
	\emph{Encyclopedia of Mathematical Physics} (\bgroup\scshape{}J.-P.
	Fran\c{c}oise\egroup{}, \bgroup\scshape{}G.~L. Naber\egroup{}, and
	\bgroup\scshape{}T.~S. Tsun\egroup{}, eds.), Elsevier, 2006, pp.~251--257.
	\doi{10.1016/b0-12-512666-2/00019-5}.
	
	\bibitem[Tak86]{Takagi1986}
	\bgroup\scshape{}S.~Takagi\egroup{}, {Vacuum Noise and Stress Induced by
		Uniform Acceleration: Hawking-Unruh Effect in Rindler Manifold of Arbitrary
		Dimension},  \emph{Progress of Theoretical Physics Supplement} \textbf{88}
	(1986), 1--142. \doi{10.1143/PTP.88.1}.
	
	\bibitem[Tak70a]{Takesaki1970-disj}
	\bgroup\scshape{}M.~Takesaki\egroup{}, Disjointness of the {KMS}-states of
	different temperatures,  \emph{Communications in Mathematical Physics}
	\textbf{17} (1970), 33--41. \doi{10.1007/bf01649582}.
	
	\bibitem[Tak70b]{Takesaki1970}
	\bgroup\scshape{}M.~Takesaki\egroup{}, \emph{{Tomita's Theory of Modular
			Hilbert Algebras and its Applications}}, Springer Berlin Heidelberg, 1970.
	\doi{10.1007/bfb0065832}.
	
	\bibitem[TW73]{Takesaki-Winnink1973}
	\bgroup\scshape{}M.~Takesaki\egroup{} and \bgroup\scshape{}M.~Winnink\egroup{},
	{Local Normality in Quantum Statistical Mechanics},  \emph{Communications in
		Mathematical Physics} \textbf{30} (1973), 129--152. \doi{10.1007/bf01645976}.
	
	\bibitem[Tes14]{Teschl2014}
	\bgroup\scshape{}G.~Teschl\egroup{}, \emph{{Mathematical Methods in Quantum
			Mechanics}}, 2nd ed., American Mathematical Society, 2014.
	\doi{10.1090/gsm/157}.
	
	\bibitem[Thi92]{Thirring1992}
	\bgroup\scshape{}W.~Thirring\egroup{}, Ergodic {Properties} in {Quantum}
	{Systems},  \emph{Les rencontres physiciens-math\'{e}maticiens de Strasbourg
		-- RCP25} \textbf{42} (1992), 67--79. Available at
	\url{http://www.numdam.org/item/RCP25_1992__42__67_0}.
	
	\bibitem[TG24]{Tjoa-Gray2024}
	\bgroup\scshape{}E.~Tjoa\egroup{} and \bgroup\scshape{}F.~Gray\egroup{}, The
	{Unruh-DeWitt} model and its joint interacting {Hilbert} space,
	\emph{Journal of Physics A: Mathematical and Theoretical} \textbf{57} (2024),
	325301. \doi{10.1088/1751-8121/ad6365}.
	
	\bibitem[{Tor}24]{Torres2024}
	\bgroup\scshape{}B.~{\relax de S. L}. {Torres}\egroup{}, Particle detector
	models from path integrals of localized quantum fields,  \emph{Physical
		Review D} \textbf{109} (2024), 065004. \doi{10.1103/physrevd.109.065004}.
	
	\bibitem[TMCA22]{Trushechkin2022}
	\bgroup\scshape{}A.~S. Trushechkin\egroup{},
	\bgroup\scshape{}M.~Merkli\egroup{}, \bgroup\scshape{}J.~D. Cresser\egroup{},
	and \bgroup\scshape{}J.~Anders\egroup{}, Open quantum system dynamics and the
	mean force {Gibbs} state,  \emph{AVS Quantum Science} \textbf{4} (2022),
	012301. \doi{10.1116/5.0073853}.
	
	\bibitem[Unr76]{Unruh1976}
	\bgroup\scshape{}W.~G. Unruh\egroup{}, Notes on black-hole evaporation,
	\emph{Physical Review D} \textbf{14} (1976), 870--892.
	\doi{10.1103/physrevd.14.870}.
	
	\bibitem[Unr86]{Unruh1986}
	\bgroup\scshape{}W.~G. Unruh\egroup{}, Accelerated monopole detector in odd
	spacetime dimensions,  \emph{Physical Review D} \textbf{34} (1986),
	1222--1223. \doi{10.1103/physrevd.34.1222}.
	
	\bibitem[UW84]{Unruh-Wald1984}
	\bgroup\scshape{}W.~G. Unruh\egroup{} and \bgroup\scshape{}R.~M. Wald\egroup{},
	What happens when an accelerating observer detects a {Rindler} particle,
	\emph{Physical Review D} \textbf{29} (1984), 1047--1056.
	\doi{10.1103/physrevd.29.1047}.
	
	\bibitem[Ver69]{Verboven1969}
	\bgroup\scshape{}E.~J. Verboven\egroup{}, {Kinematical Properties of
		Equilibrium States},  in \emph{The Many-Body Problem}
	(\bgroup\scshape{}A.~Cruz\egroup{} and \bgroup\scshape{}T.~W.
	Preist\egroup{}, eds.), \emph{Mallorca International School of Physics August
		1969}, Springer New York, 1969, pp.~73--124.
	\doi{10.1007/978-1-4899-6319-2_2}.
	
	\bibitem[Ver93]{Verch1993}
	\bgroup\scshape{}R.~Verch\egroup{}, Nuclearity, split property, and duality for
	the {Klein}-{Gordon} field in curved spacetime,  \emph{Letters in
		Mathematical Physics} \textbf{29} (1993), 297--310. \doi{10.1007/bf00750964}.
	
	\bibitem[Ver94]{Verch1994}
	\bgroup\scshape{}R.~Verch\egroup{}, Local definiteness, primarity and
	quasiequivalence of quasifree {Hadamard} quantum states in curved spacetime,
	\emph{Communications in Mathematical Physics} \textbf{160} (1994), 507--536.
	\doi{10.1007/bf02173427}.
	
	\bibitem[Ver12]{Verch-local-thermal-eq}
	\bgroup\scshape{}R.~Verch\egroup{}, {Local Covariance, Renormalization
		Ambiguity, and Local Thermal Equilibrium in Cosmology},  in \emph{Quantum
		Field Theory and Gravity} (\bgroup\scshape{}F.~Finster\egroup{},
	\bgroup\scshape{}O.~M\"{u}ller\egroup{},
	\bgroup\scshape{}M.~Nardmann\egroup{},
	\bgroup\scshape{}J.~Tolksdorf\egroup{}, and
	\bgroup\scshape{}E.~Zeidler\egroup{}, eds.), Springer Basel, 2012,
	pp.~229--256. \doi{10.1007/978-3-0348-0043-3_12}.
	
	\bibitem[Wal75]{Wald1975}
	\bgroup\scshape{}R.~M. Wald\egroup{}, On particle creation by black holes,
	\emph{Communications in Mathematical Physics} \textbf{45} (1975), 9--34.
	\doi{10.1007/bf01609863}.
	
	\bibitem[Wal94]{Wald-book}
	\bgroup\scshape{}R.~M. Wald\egroup{}, \emph{{Quantum Field Theory in Curved
			Spacetime and Black Hole Thermodynamics}}, \emph{Chicago Lectures in
		Physics}, University of Chicago Press, 1994.
	
	\bibitem[Yng05]{Yngvason2005}
	\bgroup\scshape{}J.~Yngvason\egroup{}, The role of type {III} factors in
	quantum field theory,  \emph{Reports on Mathematical Physics} \textbf{55}
	(2005), 135--147. \doi{10.1016/s0034-4877(05)80009-6}.
	
\end{thebibliography}
\end{document}